\documentclass[]{lmcs}

\usepackage[utf8]{inputenc}
\usepackage{amsmath,
	    amsfonts,
	    amstext,
	    amssymb,
	    mathrsfs,
	    mathpartir,
	    stmaryrd,
	    latexsym,
	    enumitem,
	    proof,
	    graphicx,
	    xcolor,
	    hyperref,
	    comment}

\usepackage{booktabs}
\usepackage{bm}
\usepackage{mymacros}

\usetikzlibrary{arrows, positioning, arrows.meta, automata, shapes, shapes.misc, shapes.geometric}

\tikzset{modal/.style = {>= stealth', shorten >= 0pt, shorten <= 0pt, auto,
			 node distance = 1cm, semithick},
	 point/.style = {circle, draw, fill = black, inner sep = 0.5mm},
	 elliptic state/.style = {draw, rounded rectangle, rounded rectangle arc length = 180}}

\begin{document}

\title[A Lumpability-Driven Taxonomy of Stochastic Bisimilarities]
      {A Lumpability-Driven Taxonomy of \\
       Strong and Weak Stochastic Bisimilarities \\
       with Their Congruence Properties}

\author[R.~Romanello]{Riccardo Romanello\lmcsorcid{0000-0002-2855-1221}}[a]
\author[A.~Esposito]{Andrea Esposito\lmcsorcid{0009-0009-2259-902X}}[b]
\author[M.~Bernardo]{Marco Bernardo\lmcsorcid{0000-0003-0267-6170}}[c]
\author[C.~Piazza]{Carla Piazza\lmcsorcid{0000-0002-2072-1628}}[a]
\author[S.~Rossi]{Sabina Rossi\lmcsorcid{0000-0002-1189-4439}}[d]

\address{Dipartimento di Scienze Matematiche, Informatiche e Fisiche, Universit\`a di Udine, Italy}
\email{riccardo.romanello@uniud.it, carla.piazza@uniud.it}

\address{Dipartimento di Ingegneria Elettrica e delle Tecnologie dell'Informazione, Universit\`a  di Napoli
Federico II, Italy}
\email{andrea.esposito9@unina.it}

\address{Dipartimento di Scienze Pure e Applicate, Universit\`a di Urbino, Italy}
\email{marco.bernardo@uniurb.it}

\address{Dipartimento di Scienze Ambientali, Informatica e Statistica, Universit\`a Ca' Foscari, Venezia,
Italy}
\email{sabina.rossi@unive.it}



\begin{abstract}
The aim of this paper is to study the relationships among the stochastic bisimulation-style equivalences
over PEPA -- Performance Evaluation Process Algebra definable according to the well known notions of
lumpability for the continuous-time Markov chains (CTMCs) underlying process terms. Lumpability is a central
tool in the analysis of a CTMC, because it results in aggregations of the state space enjoying properties
that are useful for efficiently computing the state probability distribution of the original chain.
\linebreak At the level of process terms, various stochastic bisimilarities accounting for activity types
and cumulative rates can be defined over PEPA, which induce different kinds of lumping. Since the
formalisations of some of them are scattered across the literature, where they appear under different, and
sometimes clashing, names, we collect them within a single, uniform framework, renaming each bisimilarity in
a consistent way after the kind of lumping it induces. Recalling that bisimilarities are said weak or strong
depending on whether they abstract from activities of unobservable types or not, we present strong and weak
variants of what we call ordinary, exact, and strict bisimilarities and show that they respectively induce
ordinary, exact, and strict lumpings. We then organise the six bisimilarities into a taxonomy establishing
all and only the inclusions holding among them, by means of some distinguishing examples. We also analyse
how the taxonomy changes in three special cases: process terms whose underlying CTMCs are time reversible,
process terms with no activities of unobservable types, and process terms with no recursion. The paper
concludes by investigating the compositionality properties of the six bisimilarities. Some of them are not
congruences with respect to the prefix and/or choice operators of PEPA. In that case we single out either a
set of process terms over which congruence with respect to those operators is achieved, or the coarsest
congruence with respect to them that is contained in the considered bisimilarity. This is especially
important for the three weak bisimilarities due to their use in equivalence-checking-based noninterference
analysis.
\end{abstract}

\maketitle

%
%
\section{Introduction}
\label{sec:intro}
%
%

Behavioural equivalences are a cornerstone of concurrency theory. They formalise when two systems exhibit
the same observable behaviour, support compositional reasoning, and enable state space minimisation up to
the chosen notion of observation. The paradigmatic examples are strong and weak bisimilarities over labelled
transition systems, introduced by Milner to express the capability of mimicking each other's behaviour
stepwise, possibly abstracting from unobservable activities, in the setting of nondeterministic process
algebra~\cite{Mil89a}. This is a compositional specification language including operators such as activity
prefix, choice, parallel composition, activity hiding, and recursion.

In the quantitative setting of stochastic process algebras, in which every activity is accompanied by the
rate (of an exponentially distributed random variable) at which it is executed, behavioural equivalences
play a twofold role~\cite{Ber07b}. On the one hand, as in the classical setting, they allow one to compare
and substitute process terms at the syntactic level. On the other hand, since the semantics of a process
term is an action-labelled variant of a continuous-time Markov chain (CTMC), an equivalence over terms
induces a partition of the state space of the underlying chain. When this partition satisfies suitable
conditions, such as ordinary (originally called strong)
lumpability~\cite{kemeny76:basic.concepts,buchholz:lumpability}, exact
lumpability~\cite{schweitzer:aggregations,buchholz:lumpability}, or their intersection named strict
lumpability~\cite{sumita.reiders:lumpability,buchholz:lumpability}, the aggregate process is still a CTMC
from which the precise performance indices of the original chain can be computed more efficiently, as the
induced aggregations enjoy useful properties with respect to the state probability distribution. For
instance, under ordinary lumpability the steady-state probability of each macro-state is the sum of the
probabilities of the states it aggregates, while under exact lumpability the states belonging to the same
macro-state are equiprobable at steady state.

Two levels of abstraction thus coexist: the level of the chains, where lumpability is a property of a
partition of the state space, and the level of the terms, where behavioural equivalences are defined on the
basis of a notion of observation. The aim of this paper is to investigate the relationships between these
two levels, focusing on the stochastic bisimilarities definable over PEPA -- Performance Evaluation Process
Algebra~\cite{hillston:book}. Some of them have been proposed in the literature, each inducing a different
kind of lumping on the underlying CTMC: the strong equivalence of~\cite{hillston:book}, the lumpable
bisimilarity of~\cite{marin:valuetools13}, the exact (performance) equivalence of~\cite{Buc94b,inf18}, the
exact lumpable bisimilarity of~\cite{valuetools25-exact}, and the Markovian forward-reverse bisimilarity
of~\cite{BR23}. Their definitions and properties are scattered across the literature, where the terminology
does not always reflect the kind of induced lumping, thus making it hard to compare the aforementioned
equivalences with each other and choose the right one for a given application.

This paper, which is an integrated and revised version of~\cite{marin:valuetools13,inf18,
valuetools25-exact} and the second part of~\cite{BR23}, aims to collect in a single work all the definitions
concerning the above notions and relate them in a unifying taxonomy. More precisely:

	\begin{itemize}

\item We gather, within a uniform framework, the theory of lumpability for CTMCs -- i.e., ordinary, exact,
and strict lumpabilities -- together with their interplay with time
reversibility~\cite{marin:mascots14,MR-acta17}, as well as six stochastic bisimilarities over PEPA, proving
for each equivalence the kind of lumping it induces. In doing so, we rename the five aforementioned
bisimilarities in a coherent and unambiguous way, so that the name of every equivalence explicitly recalls
the kind of lumping it induces on the underlying CTMC: strong ordinary bisimilarity $\sim_{\rm o}$, weak
ordinary bisimilarity $\approx_{\rm o}$, strong exact bisimilarity $\sim_{\rm e}$, weak exact bisimilarity
$\approx_{\rm e}$, strong strict bisimilarity $\sim_{\rm s}$, and weak strict bisimilarity $\approx_{\rm s}$
(which has not appeared yet in the literature). Weak and strong respectively refer to the capability of
abstracting from activities of unobservable types or not. The forthcoming Table~\ref{tab:bisim_renaming}
maps our terminology to the names appeared in the literature.

\item We organise the six lumpability-driven bisimilarities into a taxonomy, depicted in the forthcoming
Figure~\ref{fig:spectrum}, proving that some strict inclusions, and no other inclusions, hold among them.
The proof relies on some distinguishing examples, as a consequence of which $\sim_{\rm o}$, $\sim_{\rm e}$,
and $\approx_{\rm s}$ as well as $\approx_{\rm o}$ and $\approx_{\rm e}$ turn out to be pairwise
incomparable.

\item We analyse how the taxonomy changes in three special cases: process terms whose underlying CTMCs are
time reversible, process terms with no activities of unobservable types, and process terms with no
recursion.

	\end{itemize}

A further contribution of this paper is the investigation of the compositionality properties of the six
bisimilarities in a process algebraic setting. Some of them are not congruences with respect to the prefix
and/or choice operators. In that case we single out either a set of process terms over which congruence with
respect to those operators is achieved, or the coarsest congruence with respect to them that is contained in
the considered bisimilarity.

Besides providing a systematic account of a body of definitions that are currently dispersed, our taxonomy
acts as a reference map for the equivalence-based verification of stochastic systems. Knowing how the
discriminating powers of the six bisimilarities compare to each other, as well as their compositionality
properties, allows one to select, for a given application, the coarsest equivalence that is still adequate
for the properties of interest. A prominent example is the security analysis of stochastic systems, where
following~\cite{FG01} weak stochastic bisimilarities can be used to define noninterference
properties~\cite{GM82}. We will return to this application, which originally motivated the present work, in
Section~\ref{sec:conclusion}.

%
\subsection*{Related Work}
%

The classical references for bisimulation-based equivalences, both in their strong and weak variants, go
back to Milner~\cite{Mil89a}. Taxonomies of behavioural equivalences are valuable because they relate the
discriminating powers of relations that observe different aspects of computation. The prominent work in this
direction is van Glabbeek's linear-time/branching-time spectrum~\cite{Gla01}, which brought order to the
wide variety of equivalences over labelled transition systems. In the nondeterministic and probabilistic
setting, spectra of behavioural equivalences have been studied in~\cite{CR11}, in the case of (bi)simulation
relations only, and in~\cite{BDL14b}, which includes testing and trace relations too along with several
approaches to the definition of each family of equivalences. Taxonomies have also been investigated for
truly concurrent equivalences~\cite{GG01,Fec04,PU12,EB26}, which are not covered by van Glabbeek's
interleaving spectrum.

On the CTMC side, ordinary and exact lumpabilities were respectively introduced
in~\cite{kemeny76:basic.concepts} and~\cite{schweitzer:aggregations} and further investigated
in~\cite{sumita.reiders:lumpability,buchholz:lumpability,franceschinis:bounds1,franceschinis:lumping.peva,
inf18}. The relationships between lumpability and reversibility of CTMCs, which we will exploit in
Section~\ref{sec:taxonomy}, were studied in~\cite{marin:mascots14,MR-acta17}. For the sake of completeness,
we also mention \linebreak W-lumpability~\cite{Ber15}, which is the aggregation induced by a variant of
ordinary bisimilarity that abstracts from suitable combinations of activities of unobservable types by
preserving their expected durations and execution probabilities, and T-lumpability~\cite{Ber07a}, which is
the aggregation induced by stochastic versions of testing and trace equivalences~\cite{BC00}.

To the best of our knowledge, no previous work organised into a single spectrum the stochastic
bisimilarities definable over PEPA according to the kind of lumping they induce.

%
\subsection*{Structure of the Paper}
%

Section~\ref{sec:ctmc} recalls the terminology on CTMCs followed by ordinary, exact, and strict
lumpabilities, reversibility, and their mutual relationships, while Section~\ref{sec:pepa} recalls syntax,
operational semantics, and underlying CTMCs for PEPA process terms. Section~\ref{sec:bisim} presents, under
our uniform naming scheme, the six bisimilarities $\sim_{\rm o}$, $\approx_{\rm o}$, $\sim_{\rm e}$,
$\approx_{\rm e}$, $\sim_{\rm s}$, and $\approx_{\rm s}$ and proves the lumpings they induce.
Section~\ref{sec:taxonomy} develops the taxonomy of the six bisimilarities and analyses how it changes in
the three special cases. Section~\ref{sec:congruence} investigates the compositionality properties of the
six bisimilarities with respect to all PEPA operators. Section~\ref{sec:conclusion} concludes the paper and
discusses future work.

%
%
\section{Continuous-Time Markov Chains}
\label{sec:ctmc}
%
%

In this section we recall the notions about CTMCs that will be used throughout the paper. We first fix the
terminology concerning ergodicity, equilibrium distributions, and infinitesimal generators (Section
\ref{subsec:ctmc}). We then review the three notions of lumpability around which our taxonomy is built:
ordinary (Section~\ref{subsec:ordinary_lump}), exact (Section~\ref{subsec:exact_lump}), and strict
(Section~\ref{subsec:strict_lump}). Finally, we recall the concept of time reversibility
(Section~\ref{subsec:time_rev}) and its relationships with lumpability (Section~\ref{subsec:rev_lump}). The
expert reader may safely skim through this section, whose purpose is to fix notation and collect in one
place results that are scattered in the literature.

%
\subsection{CTMC Terminology}
\label{subsec:ctmc} 
%

A stochastic process is a collection of random variables describing the evolution of a system. A CTMC is a
stochastic process $X(t)$ for $t \in \mathbb{R}_{\ge 0}$ taking values in a discrete state space
$\mathcal{S}$ that enjoys the \emph{Markov or memoryless property}, i.e., the conditional (on both past and
present states) probability distribution of its future behaviour is independent of its past evolution until
the present state. Formally, for all $n \in \mathbb{N}$, time instants $t_0 < t_1 < \cdots < t_n < t_{n + 1}
\in \mathbb{R}_{\ge 0}$, and states $s_0, s_1, \ldots, s_n, s_{n + 1} \in \mathcal{S}$ it holds that:
\begin{multline*}
\mathit{Prob}(X(t_{n+1}) = s_{n+1} \mid X(t_{0}) = s_{0}, X(t_{1}) = s_{1}, \ldots, X(t_{n}) = s_{n}) \: =
\\
\mathit{Prob}(X(t_{n+1}) = s_{n+1} \mid X(t_{n}) = s_{n})
\end{multline*}

A CTMC can be represented as a state-transition graph or a state-indexed matrix. \linebreak In the first
case, each transition is labelled with some probabilistic information describing the evolution from the
source state to the target state of the transition. In the second case, the same information is stored into
an entry, indexed by those two states, of a matrix. The value of this probabilistic information is a
function of the time at which the state change takes place.

For the sake of simplicity, we restrict ourselves to \emph{time-homogeneous} CTMCs, in which conditional
probabilities of the form $\mathit{Prob}(X(t + t') = s' \mid X(t) = s)$ do not depend on $t$, so that the
aforementioned information is simply a positive real number $r = \lim_{t' \rightarrow 0}
\frac{\mathit{Prob}(X(t + t') = s' \mid X(t) = s)}{t'}$. This is called the \emph{rate} at which the CTMC
moves from state~$s$ to state~$s'$ and uniquely characterizes the exponentially distributed time taken by
the considered move: the probability that this time is at most $t$ is $1 - e^{- r \cdot t}$. It can be shown
that the sojourn time in any state $s \in \mathcal{S}$ is exponentially distributed with rate given by the
sum of the rates of the moves out of $s$. The average sojourn time in $s$ is the inverse of such a sum and
the probability of moving from $s$ to $s'$ is proportional to the corresponding rate.

A CTMC is \emph{irreducible} iff each of its states is reachable from every other state with probability
greater than $0$. A state $s \in \mathcal{S}$ is \emph{recurrent} iff the CTMC will eventually return to $s$
with probability $1$, in which case $s$ is \emph{positive recurrent} iff the expected time until the CTMC
returns to it is finite. A CTMC is \emph{ergodic} iff it is irreducible and all of its states are positive
recurrent; ergodicity coincides with irreducibility in the case that the CTMC has finitely many states, as
they form a finite strongly connected component in the graph.

Every time-homogeneous and ergodic CTMC $X(t)$ is \emph{stationary}, which means that $(X(t_{i} + t'))_{1
\le i \le n}$ has the same joint distribution as $(X(t_{i}))_{1 \le i \le n}$ for all $n \in \mathbb{N}_{\ge
1}$ and $t_{1} < \dots < t_{n}, t' \in \mathbb{R}_{\ge 0}$. In this case, $X(t)$ has a unique
\emph{equilibrium (or steady-state) probability distribution} $\bm{\pi}$ that for all $s \in \mathcal{S}$
fulfills:
\[
\pi(s) \: = \: \lim_{t \rightarrow \infty} \mathit{Prob}(X(t) = s \mid X(0) = s')
\]
where its values are independent of the specific initial state $s' \in \mathcal{S}$. These probabilities can
be computed by solving the linear system of \emph{global balance equations} $\bm{\pi} \cdot \mathbf{Q} =
\mathbf{0}$ subject to $\sum_{s \in \mathcal{S}} \pi(s) = 1$ and $\pi(s) \in \mathbb{R}_{> 0}$ for all $s
\in \mathcal{S}$. The \emph{infinitesimal generator matrix} $\mathbf{Q}$ contains for each pair of distinct
states $s$ and $s'$ the rate $q(s, s')$ of the corresponding move, which is~$0$ in the absence of a direct
move between them, while $q(s, s) = - \sum_{s' \neq s} q(s, s')$ for all $s \in \mathcal{S}$, i.e., every
diagonal element contains the opposite of the total exit rate of the corresponding state, so that each row
of $\mathbf{Q}$ sums up to $0$.

In the context of performance and reliability analysis, the notion of \emph{lumpability} provides a model
simplification technique that can be used for generating an aggregate CTMC that is smaller than the original
one but allows one to determine exact results for the original stochastic process. The concept of
lumpability can be formalized in terms of equivalence relations over the state space of the CTMC. Any such
equivalence induces a \emph{partition} of the state space of the CTMC and aggregation is achieved by
clustering equivalent states into macro-states, thus reducing the overall state space.

Intuitively, the coarser the partition, the greater the reduction of the state space; the art consists of
aggregating as many states as possible while preserving the exactness of the analysis. The three notions
recalled in Sections~\ref{subsec:ordinary_lump} to~\ref{subsec:strict_lump} differ in the direction of the
transition rates they constrain: ordinary lumpability looks at the flows \emph{leaving} the states of a
class, exact lumpability at the flows \emph{entering} them, and strict lumpability at both. This simple
observation is the key to most of the results relating lumpability and time reversibility presented in
Section~\ref{subsec:rev_lump}, as reversing time swaps incoming and outgoing flows.

%
\subsection{Ordinary Lumpability}
\label{subsec:ordinary_lump}
%

Ordinary lumpability was introduced in \cite{kemeny76:basic.concepts} and further studied in
\cite{sumita.reiders:lumpability,buchholz:lumpability,MR-acta17,inf18}. Notice that in the literature
ordinary lumpability is also referred to as strong lumpability, which we prefer to avoid so as not to
generate confusion with the use of strong and weak for bisimilarities.

	\begin{defi}[Ordinary lumpability]\label{def:ordinary_lump}
Let $X(t)$ be a CTMC with state space $\mathcal{S}$ and infinitesimal generator $\mathbf{Q}$ and $\sim$ be
an equivalence relation over $\mathcal{S}$. We say that $X(t)$ is \emph{ordinarily lumpable} with respect to
$\sim$ -- or equivalently that $\sim$ is an \emph{ordinary lumpability} for $X(t)$ -- iff $\sim$ induces a
partition of $\mathcal{S}$ -- called an \emph{ordinary lumping} -- such that for all equivalence classes
$C, C' \in \mathcal{S} / {\sim}$ with $C \neq C'$ and for all states $s_{1}, s_{2} \in C$:
\[
\sum_{s' \in C'} q(s_{1}, s') \: = \: \sum_{s' \in C'} q(s_{2}, s')
\]
	\end{defi}

Thus, an equivalence relation over the state space of a CTMC is an ordinary lumpability when it induces a
partition into equivalence classes such that, for any two states within an equivalence class, their
cumulative transition rates to any other class are the same. We refer to the equality above as the
\emph{outgoing condition}, since it constrains the total rate at which the states of a class move
\emph{towards} each other class. Notice that every CTMC is ordinarily lumpable with respect to the identity
relation -- which brings no advantage in terms of state space reduction -- as well as the relation having
only one equivalence class -- which is not useful either because the contribution to performance measures
typically differ from state to~state.

In \cite{kemeny76:basic.concepts} the authors prove that, for an equivalence relation $\sim$ over the state
space $\mathcal{S}$ of a CTMC $X(t)$, the aggregate process is still a CTMC for every initial distribution
if, and only if, $\sim$ is an ordinary lumpability for $X(t)$. Moreover, when the ordinary lumpability
condition holds, the infinitesimal generator of the lumped chain can be directly computed from the original
generator, as expressed by the following proposition (see also~\cite{buchholz:lumpability}). Its proof
follows from the general aggregation equation, where $q_{\mathrm{ag}}(C, C')$ defines the transition rate
between two aggregate states $C, C' \in \mathcal{S} / {\sim}$:
\[
q_{\mathrm{ag}}(C, C') \: = \: \frac{\sum_{s \in C} (\pi(s) \cdot \sum_{s' \in C'} q(s, s'))}{\sum_{s \in C}
\pi(s)}
\]

Intuitively, $q_{\mathrm{ag}}(C, C')$ is the steady-state conditional expectation of the transition rate
from $C$ to $C'$, given that the chain is in $C$. In general, this quantity depends on the equilibrium
distribution $\bm{\pi}$ and the aggregate process is not even Markovian, whereas under ordinary lumpability
it simplifies to an expression that does not involve $\bm{\pi}$ as shown by the following proposition.

	\begin{prop}[Aggregate process for ordinary lumpability]\label{prop:aggr_proc_ordinary_lump}
Let $X(t)$ be a CTMC with state space $\mathcal{S}$, infinitesimal generator $\mathbf{Q}$, and equilibrium
distribution $\bm{\pi}$. If $\sim$ is an ordinary lumpability for $X(t)$, then the aggregate process
$\widetilde{X}(t)$ with state space $\mathcal{S} / {\sim}$ has infinitesimal generator
$\widetilde{\mathbf{Q}}$ defined by letting for all equivalence classes $C, C' \in \mathcal{S} / {\sim}$
with $C \neq C'$ and for an arbitrary state $s \in C$:
\[
\widetilde{q}(C, C') \: = \: \sum_{s' \in C'} q(s, s')
\]
Moreover, the equilibrium distribution $\widetilde{\bm{\pi}}$ of $\widetilde{X}(t)$ exists and is such that
$\widetilde{\pi}(C) = \sum_{s \in C} \pi(s)$ for all equivalence classes $C \in \mathcal{S} / {\sim}$.
	\end{prop}

%
\subsection{Exact Lumpability}
\label{subsec:exact_lump}
%

Exact lumpability, introduced in~\cite{schweitzer:aggregations} and further investigated
in~\cite{buchholz:lumpability}, yields an aggregation in which not only the probability of any macro-state
is the sum of the probabilities of the states it contains, but all states within any macro-state are
equiprobable as well.

	\begin{defi}[Exact lumpability]\label{def:exact_lump}
Let $X(t)$ be a CTMC with state space $\mathcal{S}$ and infinitesimal generator $\mathbf{Q}$ and $\sim$ be
an equivalence relation over $\mathcal{S}$. We say that $X(t)$ is \emph{exactly lumpable} with respect to
$\sim$ -- or equivalently that $\sim$ is an \emph{exact lumpability} for $X(t)$ -- iff $\sim$ induces a
partition of $\mathcal{S}$ -- called an \emph{exact lumping} -- such that for all equivalence classes $C,
C' \in \mathcal{S} / {\sim}$ and for all states $s_{1}, s_{2} \in C$:
\[
\sum_{s' \in C'} q(s', s_{1}) \: = \: \sum_{s' \in C'} q(s', s_{2})
\]
	\end{defi}

An equivalence relation is an exact lumpability when it induces a partition of the state space of a CTMC
such that, for any two states within an equivalence class, the cumulative transition rates into such states
from any class are the same. Note the symmetry with Definition~\ref{def:ordinary_lump}: the outgoing
condition of ordinary lumpability constrains the rates from each state of a class, whereas the
\emph{incoming condition} of exact lumpability constrains the rates \emph{into} each state of a class.

Two further differences are worth pointing out. First, in Definition~\ref{def:exact_lump} the condition is
required to hold also for $C = C'$, i.e., within each class; as we shall see below in
Proposition~\ref{prop:exact_equiprobable}, it is precisely this within-class requirement that guarantees the
equiprobability, at steady state, of the states belonging to the same class. Second, the two notions behave
differently on the trivial partitions: the identity relation is both an ordinary and an exact lumpability
for every chain, whereas the relation having only one equivalence class, which is always an ordinary
lumpability, is not an exact one in general, because it requires all the column sums of $\mathbf{Q}$ to
coincide due to the within-class requirement.

The next proposition establishes a sufficient condition for exact
lumpability~\cite{franceschinis:lumping.peva}.

	\begin{prop}\label{prop:exact_sufficient}
Let $X(t) $ be a CTMC with state space $\mathcal{S}$ and infinitesimal generator $\mathbf{Q}$ and $\sim$ be
an equivalence relation over $\mathcal{S}$. Then $X(t)$ is exactly lumpable with respect to $\sim$
\linebreak if $\sim$ induces a partition of $\mathcal{S}$ such that:
		\begin{itemize}
\item $\sum_{s' \in C'} q(s', s_{1}) = \sum_{s' \in C'} q(s', s_{2})$ for all $C, C' \in \mathcal{S} /
{\sim}$ with $C \neq C'$ and for all $s_{1}, s_{2} \in C$.

\item $\sum_{s' \in C \setminus \{ s_{1} \}} q(s', s_{1}) = \sum_{s' \in C \setminus \{ s_{2} \}} q(s',
s_{2})$ for all $C \in \mathcal{S} / {\sim}$ and $s_{1}, s_{2} \in C$.

\item $q(s_{1}, s_{1}) = q(s_{2}, s_{2})$ for all $C \in \mathcal{S} / {\sim}$ and $s_{1}, s_{2} \in C$.
		\end{itemize}
	\end{prop}

Since $q(s, s) = - \sum_{s' \neq s} q(s, s')$, the last condition of Proposition~\ref{prop:exact_sufficient}
amounts to requiring that all the states in the same class have the same total exit rate. It is worth
remarking that this requirement is \emph{not} implied by Definition~\ref{def:exact_lump}: taking $C = C'$
there, one obtains that, within each class, the difference between the cumulative rate entering any state
$s$ from its own class $C$ and the total exit rate of that state, i.e., $\sum_{s' \in C \setminus \{ s \}}
q(s', s) - \sum_{s' \neq s} q(s, s')$, is constant. Hence, states of the same class may well have different
exit rates, provided that the difference is compensated for by the flows they receive from their classmates.

The next proposition leads to equiprobability. An \emph{invariant measure} of a CTMC $X(t)$ is a vector of
positive real numbers $\bm{\mu}$ satisfying the system of global balance equations $\bm{\mu} \cdot
\mathbf{Q} = \mathbf{0}$. If $X(t)$ is irreducible and $\bm{\mu_1}$ and $\bm{\mu_2}$ are two invariant
measures of $X(t)$, then there exists a constant $k \in \mathbb{R}_{> 0}$ such that $\bm{\mu_1} = k \cdot
\bm{\mu_2}$. If the CTMC is ergodic, then there exists a unique invariant measure whose elements sum to $1$,
which is $\bm{\pi}$.

	\begin{prop}\label{prop:exact_equiprobable}
Let $X(t)$ be an ergodic CTMC with state space $\mathcal{S}$ and $\sim$ be an equivalence relation over
$\mathcal{S}$. If $X(t)$ is exactly lumpable with respect to $\sim$, then for all states $s_{1}, s_{2} \in
\mathcal{S}$ such that $s_{1} \sim s_{2}$ it holds that $\mu(s_{1}) = \mu(s_{2})$ for all invariant measures
$\bm{\mu}$ of $X(t)$.
	\end{prop}

In the case of exact lumpability, if the CTMC is ergodic then its steady-state distribution is
equidistributed among the states of every class of the partition. Therefore, with the knowledge of the
lumped chain generator, one may compute the steady-state distribution and deduce (by local equidistribution)
the steady-state distribution of the original chain. \linebreak It must be emphasized that this last step is
impossible with ordinary lumpability, because it does not ensure equiprobability of the states in a
macro-state.

	\begin{prop}[Aggregate process for exact lumpability]\label{prop:aggr_proc_exact_lump}
Let $X(t)$ be a CTMC with state space $\mathcal{S}$, infinitesimal generator $\mathbf{Q}$, and equilibrium
distribution $\bm{\pi}$. If $\sim$ is an exact lumpability for $X(t)$, then the aggregate process
$\widetilde{X}(t)$ with state space $\mathcal{S} / {\sim}$ has infinitesimal generator
$\widetilde{\mathbf{Q}}$ defined by letting for all equivalence classes $C, C' \in \mathcal{S} / {\sim}$
with $C \neq C'$ and for an arbitrary state $s' \in C'$:
\[
\widetilde{q}(C, C') \: = \: \frac{|C'|}{|C|} \cdot \sum_{s \in C} q(s, s')
\]
Moreover, the equilibrium distribution $\widetilde{\bm{\pi}}$ of $\widetilde{X}(t)$ exists and is such that
$\widetilde{\pi}(C) = |C| \cdot \pi(s)$ for all equivalence classes $C \in \mathcal{S} / {\sim}$ and for an
arbitrary state $s \in C$.
	\end{prop}

%
\subsection{Strict Lumpability}
\label{subsec:strict_lump}
%

Ordinary and exact lumpabilities are, in general, incomparable notions: neither condition implies the other,
as they constrain transition flows in opposite directions. It is then natural to consider their
intersection, which yields the notion of strict lumpability introduced in~\cite{sumita.reiders:lumpability}
and further examined in~\cite{buchholz:lumpability}.

	\begin{defi}[Strict lumpability]\label{def:strict_lump}
Let $X(t)$ be a CTMC with state space $\mathcal{S}$ and $\sim$ be an equivalence relation over
$\mathcal{S}$. We say that $X(t)$ is \emph{strictly lumpable} with respect to $\sim$ -- or equivalently that
$\sim$ is a \emph{strict lumpability} for $X(t)$ -- iff it is both ordinarily and exactly lumpable with
respect to $\sim$.
	\end{defi}

Notice that the identity relation is a strict lumpability for every CTMC, so the notion is never vacuous.
Non-trivial strict lumpabilities are however quite demanding, as they constrain the flows in both
directions.

	\begin{exa}\label{ex:ord_exact_strict_lump_diff}
Consider the CTMC $X(t)$ with state space $\mathcal{S} = \{ s_{0}, s_{1}, s'_{1}, s''_{1}, s_{2},
\bar{s}_{2} \}$ and rates $q(s_{0}, s_{1}) = r$, $q(s_{0}, s_{2}) = r$, $q(s_{1}, s'_{1}) = r'$, $q(s_{1},
s''_{1}) = r''$, and $q(s_{2}, \bar{s}_{2}) = r' + r''$. Then:

		\begin{itemize}
\item The partition $\{ \{ s_{0} \}, \{ s_{1}, s_{2} \}, \{ s'_{1}, s''_{1}, \bar{s}_{2} \} \}$ is the
coarsest non-trivial ordinary lumping.

\item The partition $\{ \{ s_{0} \}, \{ s_{1}, s_{2} \}, \{ s'_{1} \}, \{ s''_{1} \}, \{ \bar{s}_{2} \} \}$
is the only non-trivial exact lumping when $r' \neq r''$, while $\{ \{ s_{0} \}, \{ s_{1}, s_{2} \}, \{
s'_{1}, s''_{1} \}, \{ \bar{s}_{2} \} \}$ is the coarsest one when $r' = r''$.

\item None of them is a strict lumping, only the partition induced by the identity relation is.
		\end{itemize}
	\end{exa}

The intrinsically complex nature of strict lumpability comes with pros and cons. On the one hand, proving an
equivalence relation to be a non-trivial strict lumpability is a complex task. On the other hand, when such
a relation can be defined, it is possible to exploit nice properties of the resulting aggregate process.

	\begin{prop}\label{prop:strict_flow}
Let $X(t)$ be a CTMC with state space $\mathcal{S}$ and infinitesimal generator $\mathbf{Q}$ and let $\sim$
be a strict lumpability for $X(t)$. Then for all equivalence classes $C, C' \in \mathcal{S} / {\sim}$ and
for all states $s \in C$ and $s' \in C'$ it holds that:
\[
|C| \cdot \sum_{s'' \in C'} q(s, s'') \: = \: |C'| \cdot \sum_{s'' \in C} q(s'', s')
\]
	\end{prop}

The above equality amounts to computing the same number, which is the total transition rate from $C$ to
$C'$, in two different ways: summing the (common, by ordinary lumpability) outgoing rate over the $|C|$
states of $C$ or summing the (common, by exact lumpability) incoming rate over the $|C'|$ states of $C'$.

	\begin{prop}[Aggregate process for strict lumpability]\label{prop:aggr_proc_strict_lump}
Let $X(t)$ be a CTMC with state space $\mathcal{S}$, infinitesimal generator $\mathbf{Q}$, and equilibrium
distribution $\bm{\pi}$. If $\sim$ is a strict lumpability for $X(t)$, then the aggregate process
$\widetilde{X}(t)$ with state space $\mathcal{S} / {\sim}$ has infinitesimal generator
$\widetilde{\mathbf{Q}}$ defined by letting for all equivalence classes $C, C' \in \mathcal{S} / {\sim}$
with $C \neq C'$ and for an arbitrary state $s \in C$:
\[
\widetilde{q}(C, C') \: = \: \sum_{s' \in C'} q(s, s')
\]
Moreover, it holds that $\pi(s_{1}) = \pi(s_{2})$ for all states $s_{1}, s_{2} \in \mathcal{S}$ such that
$s_{1} \sim s_{2}$, hence the equilibrium distribution $\widetilde{\bm{\pi}}$ of $\widetilde{X}(t)$, which
exists, is such that $\widetilde{\pi}(C) = |C| \cdot \pi(s)$ for all equivalence classes $C \in \mathcal{S}
/ {\sim}$ and for an arbitrary state $s \in C$.
	\end{prop}

%
\subsection{Time Reversibility}
\label{subsec:time_rev}
%

The analysis of an ergodic CTMC with an equilibrium distribution is significantly simplified if the CTMC
satisfies the property of \emph{time reversibility}~\cite{kelly:reversibility.networks}. Indeed, the
equilibrium distribution of a reversible chain can be computed from the \emph{detailed (or partial) balance
equations} recalled below, which are considerably easier to solve than the global balance equations.
Classical models such as birth-death processes, and hence many queueing systems, fall within this category
\cite{kelly:reversibility.networks}.

Given a stationary CTMC $X(t)$, its \emph{reversed process} is defined as $X^{\mathrm{R}}(t) = X(t' - t)$
for a fixed $t' \in \mathbb{R}_{\ge 0}$: intuitively, $X^{\mathrm{R}}(t)$ is the process obtained by
observing the evolution of $X(t)$ backwards in time (by stationarity, the choice of the time origin $t'$ is
unimportant). \linebreak It can be demonstrated that $X^{\mathrm{R}}(t)$ is a stationary CTMC as well.
Notice that the reversed process exists for any stationary CTMC; its transition rates are determined by the
transition rates and the equilibrium distribution of the forward process, as stated by the following
proposition~\cite{kelly:reversibility.networks}.

	\begin{prop}\label{prop:reversibility}
Let $X(t)$ be a stationary CTMC with state space $\mathcal{S}$, infinitesimal generator $\mathbf{Q}$, and
equilibrium distribution $\bm{\pi}$. For all states $s, s' \in \mathcal{S}$ with $s \neq s'$, the transition
rates of $X^{\mathrm{R}}(t)$, whose infinitesimal generator is denoted by $\mathbf{Q}^{\mathrm{R}}$, are
given by:
\[
q^{\mathrm{R}}(s', s) \: = \: \frac{\pi(s)}{\pi(s')} \cdot q(s, s')
\]
with $X(t)$ and $X^{\mathrm{R}}(t)$ sharing the same equilibrium distribution $\bm{\pi}$.
	\end{prop}

A CTMC $X(t)$ is \emph{time reversible} iff it is stochastically identical to $X^{\mathrm{R}}(t)$, meaning
that $(X(t_{i}))_{1 \le i \le n}$ has the same joint distribution as $(X(t' - t_{i}))_{1 \le i \le n}$ for
all $n \in \mathbb{N}_{\ge 1}$ and $t_{1} < \dots < t_{n}, t' \in \mathbb{R}_{\ge 0}$. In other words, a
CTMC is time reversible when it coincides with its own reversed process: an observer of the process at
equilibrium cannot tell whether time is flowing forwards or backwards. While the reversed process always
exists, time reversibility is a distinguishing property which, for ergodic CTMCs, can be characterized in
terms of the equilibrium distribution $\bm{\pi}$ and the infinitesimal generator
$\mathbf{Q}$~\cite{kelly:reversibility.networks}.

	\begin{thm}[Detailed balance equations]\label{thm:detailed_balance}
A stationary CTMC $X(t)$ with state space $\mathcal{S}$ and infinitesimal generator $\mathbf{Q}$ is time
reversible iff there exists a vector $\bm{\pi}$ of positive real numbers summing to $1$ such that for all
states $s, s' \in \mathcal{S}$ with $s \neq s'$:
\[
\pi(s) \cdot q(s, s') \: = \: \pi(s') \cdot q(s', s)
\]
If such a vector exists, then it is the equilibrium distribution of $X(t)$.
	\end{thm}

The above proposition ties the reversibility of a CTMC to the values of its stationary distribution. Indeed,
it can be derived by imposing that the infinitesimal generators of $X(t)$ and $X^{\mathrm{R}}(t)$ coincide:
by Proposition~\ref{prop:reversibility}, the equality $q(s', s) = q^{\mathrm{R}}(s', s)$ amounts exactly to
the detailed balance equation for the pair of states $s$ and $s'$.

However, it turns out that another condition can be used to decide whether a CTMC is time reversible, known
as \emph{Kolmogorov's criterion}~\cite{kelly:reversibility.networks}, which has the advantage of not
requiring the knowledge of the equilibrium distribution $\bm{\pi}$.

	\begin{thm}[Kolmogorov's criterion]\label{thm:kolmogorov}
A stationary CTMC $X(t)$ with state space $\mathcal{S}$ and infinitesimal generator $\mathbf{Q}$ is time
reversible iff for all $n \in \mathbb{N}_{\ge 2}$ and distinct $s_1, \dots, s_n \in \mathcal{S}$:
\[
q(s_1, s_2) \cdot \ldots \cdot q(s_{n - 1}, s_{n}) \cdot q(s_n, s_1) \: = \: q(s_1, s_n) \cdot q(s_n, s_{n -
1}) \cdot \ldots \cdot q(s_2, s_1)
\]
	\end{thm}

%
\subsection{Relationship between Lumpability and Reversibility}
\label{subsec:rev_lump}
%

Since reversing time swaps outgoing and incoming flows, with ordinary and exact lumpabilities respectively
constraining precisely such flows, it is natural to expect a connection between lumpability and
reversibility, which was proved in~\cite{marin:mascots14,MR-acta17}. In particular, the authors showed how
ordinary and exact lumpabilities behave when related to a CTMC and its reversed counterpart.

	\begin{thm}\label{thm:exact_ordinary_reversed}
Let $X(t)$ be a stationary CTMC and $\sim$ be an exact lumpability for $X(t)$. Then $\sim$ is an ordinary
lumpability for $X^{\mathrm{R}}(t)$.
	\end{thm}

Intuitively, by Proposition~\ref{prop:reversibility} the rates of $X^{\mathrm{R}}(t)$ are the reversed rates
of $X(t)$, weighted by the equilibrium probabilities; since exact lumpability guarantees, by
Proposition~\ref{prop:exact_equiprobable}, that equivalent states are equiprobable, the incoming-flow
condition on $X(t)$ translates exactly into the outgoing-flow condition on $X^{\mathrm{R}}(t)$.

Unfortunately, the symmetric result does not hold: an ordinary lumpability for $X(t)$ is not, in general, an
exact lumpability for $X^{\mathrm{R}}(t)$.

The tight relationship between exact and ordinary lumpabilities, when considered on a CTMC and its reversed
version, induces a non-trivial result for strict lumpability.

	\begin{thm}\label{thm:strict_reversed}
Let $X(t)$ be a stationary CTMC. An equivalence relation $\sim$ is a strict lumpability for $X(t)$ iff it is
a strict lumpability for $X^{\mathrm{R}}(t)$.
	\end{thm}

The result just introduced also allows us to establish a bridge between strict lumpability, reversed
processes, and the Markov property.

	\begin{thm}\label{thm:strict_markov}
Let $X(t)$ be a stationary CTMC and $\sim$ be an equivalence relation. If $\sim$ is a strict lumpability for
$X(t)$, then both aggregate processes $\widetilde{X}(t)$ and $\widetilde{X^{\mathrm{R}}}(t)$ satisfy the
Markov property.
	\end{thm}

This is an immediate consequence of the previous results: a strict lumpability for $X(t)$ is in particular
an ordinary lumpability for $X(t)$ and, by Theorem~\ref{thm:strict_reversed}, also for $X^{\mathrm{R}}(t)$,
with ordinary lumpability being precisely the condition ensuring that the aggregate process is a~CTMC.

Before introducing the last result relating reversibility and lumpability, recall that, while the reversed
process always exists, reversibility holds only under the special circumstances characterized by
Theorems~\ref{thm:detailed_balance} and~\ref{thm:kolmogorov}. The following result highlights a relationship
between exact and strict lumpabilities over reversible CTMCs.

	\begin{thm}\label{thm:rev_exact_strict}
Let $X(t)$ be a reversible CTMC and $\sim$ be an equivalence relation. $\sim$ is a strict lumpability for
$X(t)$ iff it is an exact lumpability for $X(t)$.
	\end{thm}

This final result can be interpreted as the property of exact lumpability to \emph{intrinsically} behave as
an ordinary one as well when the CTMC is reversible.

%
%
\section{Performance Evaluation Process Algebra}
\label{sec:pepa}
%
%

We now move from the level of CTMCs, where lumpability is a property of a relation over the state space, to
the level of process terms, whose underlying semantics in the form of action-labelled variants of CTMCs can
be compared via stochastic behavioral equivalences -- defined over all terms -- each inducing a different
kind of lumping. Since we will focus on stochastic bisimilarities over PEPA -- Performance Evaluation
Process Algebra~\cite{hillston:book}, in this section we recall the syntax (Section~\ref{subsec:syntax}),
operational semantics (Section~\ref{subsec:semantics}), and underlying CTMC
(Section~\ref{subsec:underl_ctmc}) of the considered language.

%
\subsection{Syntax of PEPA}
\label{subsec:syntax}
%

PEPA is an algebraic process calculus enhanced with stochastic timing information that may be used to
calculate performance measures in addition to prove functional properties. The basic elements of PEPA are
\emph{activities} and \emph{components}.

Each activity is represented as a pair $(\alpha, r)$ where $\alpha$ is the \emph{activity type} while $r$ is
the \emph{activity rate}. We assume that there is a countable set $\cA$ of activity types, including a
distinguished one, traditionally denoted by $\tau$, for \emph{unobservable} activities that are private to
the components executing them. An activity rate may be a positive real number, which uniquely identifies the
exponentially distributed random variable determining the duration of the activity, or the distinguished
symbol $\top$, which should be read as \emph{unspecified}; we denote by $\cR$ the set of activity rates.

The PEPA language provides a small set of operators. These allow process terms to be constructed by defining
the behaviour of components via the activities they undertake and the interactions among them. The syntax
for PEPA is defined by the following grammar:
\[
P \: ::= \: \underline{0} \mid A \mid (\alpha, r) \, . \, P \mid P + P \mid P \sync{L} P \mid P \, / \, L
\]
where $P$ denotes a component, $A$ is a constant, and $L \subseteq \cA \setminus \{ \tau \}$. We write $\cC$
for the set of all components. The meaning of the various operators (whose precedence from highest to lowest
is determined by the order in which they appear in the syntax) can be informally described as follows:

	\begin{itemize}
\item Inaction: $\underline{0}$ performs no activity.

\item Constant: each constant $A$ is assumed to have a defining equation of the form $A \rmdef P$
establishing that the behavior of $A$ is given by the behavior of $P$. Since $A$ can occur within $P$, the
defining equation may be recursive. Each constant present in $P$ is assumed to be \emph{guarded}, i.e., to
occur in the scope of a prefix.

\item Prefix: $(\alpha, r) \, . \, P$ can carry out the activity $(\alpha, r)$ of type $\alpha$ at rate $r$
and subsequently behaves as $P$. When $a = (\alpha, r)$, component $(\alpha, r) \, . \, P$ will be written
as $a \, . \, P$.

\item Choice: $P_{1} + P_{2}$ behaves as either $P_{1}$ or $P_{2}$ depending on whether an activity of
$P_{1}$ or an activity of $P_{2}$ terminates first. In other words, all activities executable by $P_{1}$ and
$P_{2}$ compete under a \emph{race policy}. It is worth recalling that the minimum of several exponentially
distributed random variables is still exponentially distributed with rate given by the sum of the rates of
the original variables. As a consequence, each activity involved in the race has a probability of
terminating first that is proportional to its rate. Moreover, the continuous nature of exponential
distributions ensures that the probability of $P_{1}$ and $P_{2}$ completing an activity at the same time is
$0$.

\item Cooperation: $P_{1} \sync{L} P_{2}$ is the parallel execution of $P_{1}$ and $P_{2}$. They proceed
independent of each other based on the race policy in the case of activities whose type does not belong to
the \emph{cooperation set} $L$ (\emph{individual activities}), while they have to synchronise on activities
of the same type when this belongs to $L$ (\emph{shared activities}). In the latter case, the resulting
activity has the same type as the two original ones and a rate reflecting the rate of the slower
subcomponent (if either original activity has an unspecified rate, then that activity is passive and the
rate is completely determined by the one of the other original activity). Unlike choice, here no
subcomponent is discarded after a race.

\item Hiding: $P \, / \, L$ behaves as $P$ except that any executed activity of type belonging to $L$
\linebreak is hidden, i.e., changed to the unobservable type $\tau$.
	\end{itemize}

%
\subsection{Operational Semantics}
\label{subsec:semantics}
%

The operational semantic rules for PEPA are given in Table~\ref{tab:sem_rules}. The semantics of every
component $P$ is a labelled \emph{multi-transition} system $\cD\cG(P)$, named \emph{derivation graph} of
$P$, which is obtained by inductively applying the rules to the syntactical structure of $P$. Multiplicities
of transitions, intended as all possible different ways of deriving a transition from the semantic rules,
are significant: for instance, $(\alpha, r_{1}) \, . \, P + (\alpha, r_{2}) \, . \, P$ has two transitions
to $P$ even when $r_{1} = r_{2}$. The set $\mathit{ds}(P)$ of states reachable from $P$ is termed the
\emph{derivative set} of $P$; states correspond to syntactically different components. Every state in
$\mathit{ds}(P)$ has finitely many outgoing transitions thanks to guardedness, with $\mathit{ds}(P)$ being
finite if no cooperation or hiding occurs in recursive constant definitions.

	\begin{table}[t]

\caption{Operational semantics rules for PEPA.}
\label{tab:sem_rules}

\begin{center}
\begin{tabular}{ccc}
\toprule
\multicolumn{3}{c}
{$\dfrac{}{(\alpha, r) \, . \, P \transits{(\alpha, r)} P}$
\qquad
$\dfrac{P \transits{(\alpha, r)} P' \quad A \rmdef P}{A \transits{(\alpha, r)} P'}$} \\[8mm]
\multicolumn{3}{c}
{$\dfrac{P_{1} \transits{(\alpha, r)} P'_{1}}{P_{1} + P_{2} \transits{(\alpha, r)} P'_{1}}$
\qquad
$\dfrac{P_{2} \transits{(\alpha, r)} P'_{2}}{P_{1} + P_{2} \transits{(\alpha, r)} P'_{2}}$} \\[8mm]
\multicolumn{3}{c}
{$\dfrac{P_{1} \transits{(\alpha, r)} P'_{1} \quad \alpha \notin L}{P_{1} \sync{L} P_{2} \transits{(\alpha,
r)} P'_{1} \sync{L} P_{2}}$
\qquad
$\dfrac{P_{2} \transits{(\alpha, r)} P'_{2} \quad \alpha \notin L}{P_{1} \sync{L} P_{2} \transits{(\alpha,
r)} P_{1} \sync{L} P'_{2}}$} \\[8mm]
\multicolumn{3}{c}
{$\dfrac{P_{1} \transits{(\alpha, r_1)} P'_{1} \quad P_{2} \transits{(\alpha, r_2)} P'_{2} \quad \alpha \in
L}{P_{1} \sync{L} P_{2} \transits{(\alpha, r)} P'_{1} \sync{L} P'_{2}}$ \quad where $r =
\dfrac{r_1}{r_{\alpha}(P_{1})} \cdot \dfrac{r_2}{r_{\alpha}(P_{2})} \cdot \mathrm{min}(r_{\alpha}(P_{1}),
r_{\alpha}(P_{2}))$} \\[8mm]
\multicolumn{3}{c}
{$\dfrac{P \transits{(\alpha, r)} P' \quad \alpha \notin L}{P \, / \, L \transits{(\alpha, r)} P' \, / \,
L}$
\qquad
$\dfrac{P \transits{(\alpha, r)} P' \quad \alpha \in L}{P \, /L \transits{(\tau, r)} P' \, / \, L}$} \\
\bottomrule 
\end{tabular}
\end{center}

	\end{table}

The first two rules for cooperation formalize an \emph{interleaving} view of concurrency. Consider for
instance $(\alpha_1, r_1) \, . \, \underline{0} \sync{\emptyset} (\alpha_2, r_2) \, . \, \underline{0}$.
Since the cooperation set is empty, the two activities are executed in parallel. Recalling that the
probability that $\alpha_1$ and $\alpha_2$ terminate simultaneously is $0$, if $\alpha_1$ (resp.\
$\alpha_2$) terminates first then the residual time to the termination of $\alpha_2$ (resp.\ $\alpha_1$)
\linebreak is still exponentially distributed with rate $r_2$ (resp.\ $r_1$). This is a consequence of the
memoryless property of exponential distributions.

The third rule for cooperation relies on the notion of \emph{apparent rate}: $r_{\alpha}(P)$ is the sum of
the rates of all the activities of type $\alpha$ that $P$ can execute. If the considered activities have
unspecified rates, then $r_{\alpha}(P)$ is $\top$ preceded by the number of such activities whereas
$\frac{\top}{r_{\alpha}(P)}$ is the inverse of that number. Moreover, $\mathrm{min}(r, n \top) = r$ when $r
\in \mathbb{R}_{> 0}$.

%
\subsection{Underlying Stochastic Process}
\label{subsec:underl_ctmc}
%

We denote by $\cC_{\rm pc}$ the set of components that are \emph{performance closed}, i.e., such that their
derivation graphs have no transitions labelled with $\top$. Consider $P \in \cC_{\rm pc}$ that is
\emph{finite}, i.e., such that $\cD\cG(P)$ has finitely many states $\mathit{ds}(P) = \{ P_1, \ldots, P_n
\}$ and finitely many transitions. If we define the stochastic process $X(t)$ for $t \in \mathbb{R}_{\ge 0}$
by letting $X(t) = P_i$ when $P$ behaves as $P_i$ at time $t$, then $X(t)$ turns out to be a
CTMC~\cite{hillston:book} whose infinitesimal generator $\mathbf{Q}$ is constructed as follows.

The \emph{transition rate} $q(P_i, P_j)$ from $P_i$ to $P_j$ is the sum of the activity rates labeling the
transitions from the state corresponding to $P_i$ to the state corresponding to $P_j$:
\[
q(P_i, P_j) \: = \sum_{a \in \Ac(P_i | P_j)} r_a
\]
where $P_i \neq P_j$, $\Ac(P_i| P_j) = \bms a \in \Ac(P_i) \mid P_i \transits{a} P_j \ems$, $\Ac(P_i)$ is
the multiset of activities executable by $P_i$, and $r_a$ is the rate of $a$. Clearly, if $P_j$ is not a
one-step derivative of $P_i$, then $q(P_i, P_j) = 0$. The values $q(P_i, P_j)$ are the off-diagonal elements
of the infinitesimal generator $\mathbf{Q}$ of the CTMC. Its diagonal elements are formed as the negative
sums of the non-diagonal elements of the corresponding rows. We use the following notation: $q(P_i) =
\sum_{P_j \neq P_i} q(P_i, P_j)$ and $q(P_i, P_i) = -q(P_i)$. If the CTMC is irreducible (in addition to
being finite), then its steady-state distribution $\bm{\pi}$ exists and can be computed by solving the
global balance equations $\bm{\pi} \cdot \mathbf{Q} = \mathbf{0}$ under the constraint $\sum_{P_i \in
\mathit{ds}(P)} \pi(P_i) = 1$ with every $\pi(P_i) \in \mathbb{R}_{> 0}$.

We conclude by introducing related notation that will be used in the bisimilarities:

	\begin{itemize}
\item The \emph{conditional transition rate} $q(P_i, P_j, \alpha)$ from $P_i$ to $P_j$ via activities of
type $\alpha$ -- where $P_i$ and $P_j$ are not necessarily distinct -- is the sum of the rates labelling all
$\alpha$-transitions from the state corresponding to $P_i$ to the state corresponding to $P_j$:
\[
q(P_i, P_j, \alpha) \: = \sum_{P_i \transits{(\alpha, r)} P_j} r
\]

\item The \emph{cumulative conditional transition rate} $q[P_i, C, \alpha]$ from $P_i$ to $C \subseteq
\mathit{ds}(P)$ via activities of type $\alpha$ is the sum of the rates labelling all $\alpha$-transitions
from the state corresponding to $P_i$ to some state in $C$:
\[
q[P_i, C, \alpha] \: = \: \sum_{P' \in C} q(P_i, P' , \alpha)
\]

\item The \emph{total transition rate} $q[P_i, \alpha]$ from $P_i$ via activities of type $\alpha$ is the
sum of the rates labelling all $\alpha$-transitions from the state corresponding to $P_i$:
\[
q[P_i, \alpha] \: = \: \sum_{P' \in \cC} q(P_i, P', \alpha)
\]

\item The \emph{incoming cumulative conditional transition rate} $q[C, P_i, \alpha]$ into $P_i$ from $C
\subseteq \cC$ via activities of type $\alpha$ is the sum of the rates labelling all $\alpha$-transitions
from some state in $C$ to the state corresponding to $P_i$:
\[
q[C, P_i, \alpha] \: = \: \sum_{P' \in C} q(P', P_i, \alpha)
\]
	\end{itemize}

%
%
\section{Lumpability-Driven Bisimilarities over PEPA}
\label{sec:bisim}
%
%

In this section we present the definitions of six bisimulation-style equivalences over PEPA named after the
kind of lumping they induce: strong and weak ordinary bisimilarities (Section~\ref{subsec:ord_bisim}),
strong and weak exact bisimilarities (Section~\ref{subsec:exact_bisim}), and strong and weak strict
bisimilarities (Section~\ref{subsec:strict_bisim}). We also discuss the form that incoming cumulative rates
should take in order to avoid trivial exact and strict bisimilarities (Section~\ref{subsec:from_out_to_in}).

Each of the six bisimilarities collected in Table~\ref{tab:bisim_renaming} follows the pattern
of~\cite{LS91} and is defined as the union of all the bisimulations of its kind. It can be shown to be the
largest bisimulation of that kind by first proving that the transitive closure of the countable union of
bisimulations of that kind is a bisimulation of that kind too~\cite{GSS95,hillston:book}. For each of the
three pairs of bisimilarities, we first introduce the strong variant, which handles $\tau$-activities like
all the other activities, followed by the weak variant, in which $\tau$-activities receive a special
treatment, then we prove which of the three lumpings they induce.

	\begin{table}[t]

\caption{The six stochastic bisimilarities under our lumpability-driven naming scheme.}
\label{tab:bisim_renaming}

\begin{tabular}{|c|l|c|l|}
\hline
\emph{} & \emph{our name} & \emph{reference} & \emph{also known as} \\
\hline
\hline
$\sim_{\rm o}$ & strong ordinary bisimilarity & \cite{hillston:book} & strong equivalence \\
\hline
$\approx_{\rm o}$ & weak ordinary bisimilarity & \cite{marin:valuetools13} & lumpable bisimilarity \\
\hline
\hline
$\sim_{\rm e}$ & strong exact bisimilarity & \cite{Buc94b} & exact performance equivalence \\
& & \cite{inf18} & exact equivalence \\
\hline
$\approx_{\rm e}$ & weak exact bisimilarity & \cite{valuetools25-exact} & exact lumpable bisimilarity \\
\hline
\hline
$\sim_{\rm s}$ & strong strict bisimilarity & \cite{BR23} & Markovian forward-reverse bisimilarity \\
\hline
$\approx_{\rm s}$ & weak strict bisimilarity & --- & --- \\
\hline
\end{tabular}

	\end{table}

%
\subsection{Ordinary Bisimilarities}
\label{subsec:ord_bisim}
%

We start by recalling the notion of \emph{strong equivalence} for PEPA introduced in~\cite{hillston:book}:
two components are strongly equivalent iff there is an equivalence relation relating them such that, for any
activity type $\alpha$, the cumulative conditional transition rates from those components to any equivalence
class, via activities of type $\alpha$, coincide.

As the title of this subsection suggests, this is the first place where we settle the naming conventions for
the stochastic bisimilarities that are scattered in the literature. The name strong equivalence was
originally chosen because of the tight relationship between this notion and that of strong lumpability on
the underlying CTMC. Since, following Section~\ref{sec:ctmc}, we adopted the name ordinary lumpability for
the latter, we consistently rename the equivalence as ordinary bisimilarity.

	\begin{defi}[Strong ordinary bisimilarity]\label{def:strong_ordinary_bisim}
We say that $P_{1}, P_{2} \in \cC$ are \emph{strongly ordinary bisimilar}, written $P_{1} \sim_{\rm o}
P_{2}$, iff $(P_{1}, P_{2}) \in \cB$ for some strong ordinary bisimulation $\cB$. An equivalence relation
$\cB \subseteq \cC \times \cC$ is a \emph{strong ordinary bisimulation} iff, whenever $(P_{1}, P_{2}) \in
\cB$, then for all activity types $\alpha \in \cA$ and for all equivalence classes $C \in \cC / \cB$:
\[
q[P_{1}, C, \alpha] \: = \: q[P_{2}, C, \alpha]
\]
	\end{defi}

	\begin{thm}[{\cite{hillston:book}}]\label{thm:strong_ordinary_bisim_induces}
For all $P \in \cC_{\rm pc}$, $\sim_{\rm o}$ induces a partition of $\mathit{ds}(P)$ that is an ordinary
lumping.
	\end{thm}

As a consequence, in case of ergodicity the aggregate CTMC satisfies the property that the steady-state
probability of each macro-state is equal to the sum of the steady-state probabilities of the original states
it contains. Note that, unlike ordinary lumpability, strong ordinary bisimilarity compares cumulative
conditional transition rates also towards the same class to which $P_{1}$ and $P_{2}$ belong.

The above bisimilarity deals with all the activity types in the same fashion. However, the unobservable type
$\tau$ deserves a special treatment: $\tau$-activities performed \emph{within} an equivalence class amount
to invisible changes of state and should not concur to distinguish components in that class. For this very
reason, the notion of \emph{lumpable bisimilarity} was introduced in~\cite{marin:valuetools13}: the
conditions on $\tau$-activities are only imposed towards the classes that do not contain the compared
components. We point out that this is similar to what happens with the weak probabilistic bisimilarities
examined in~\cite{BKHW05} in the setting of action-labelled discrete-time Markov chains, whose transitions
are labelled with probabilities.

	\begin{defi}[Weak ordinary bisimilarity]\label{def:weak_ordinary_bisim}
We say that $P_{1}, P_{2} \in \cC$ are \emph{weakly ordinary bisimilar}, written $P_{1} \approx_{\rm o}
P_{2}$, iff $(P_{1}, P_{2}) \in \cB$ for some weak ordinary bisimulation $\cB$. An equivalence relation $\cB
\subseteq \cC \times \cC$ is a \emph{weak ordinary bisimulation} iff, whenever $(P_{1}, P_{2}) \in \cB$,
then for all activity types $\alpha \in \cA$ and for all equivalence classes $C \in \cC / \cB$ such that:
		\begin{itemize}
\item either $\alpha \neq \tau$

\item or $\alpha = \tau$ and $P_{1}, P_{2} \notin C$
		\end{itemize}
it holds that:
\[
q[P_{1}, C, \alpha] \: = \: q[P_{2}, C, \alpha]
\]
	\end{defi}

As we will exemplify in Section~\ref{sec:taxonomy}, weak ordinary bisimilarity is coarser than strong
ordinary bisimilarity, because the former abstracts from activities of type $\tau$ among components
belonging to the same equivalence class.

We show that weak ordinary bisimilarity retains the fundamental property of inducing an ordinary lumping on
the underlying CTMC, which aggregates more than the one induced by strong ordinary bisimilarity.

	\begin{thm}\label{thm:weak_ordinary_bisim_induces}
For all $P \in \cC_{\rm pc}$, $\approx_{\rm o}$ induces a partition of $\mathit{ds}(P)$ that is an ordinary
lumping.
	\end{thm}

	\begin{proof}
Consider the partition induced by $\approx_{\rm o}$ over $\mathit{ds}(P)$ and two arbitrary components
$P_{1}, P_{2} \in \mathit{ds}(P) \cap C$ for some equivalence class $C \in \cC / {\approx_{\rm o}}$. From
the outgoing condition in the definition of $\approx_{\rm o}$, for all equivalence classes $C' \in \cC /
{\approx_{\rm o}}$ with $C \neq C'$ we derive that $q[P_{1}, C', \alpha] = q[P_{2}, C', \alpha]$ for all
$\alpha \in \mathcal{A}$ (in the case $\alpha = \tau$ note that $C \neq C'$ implies $P_{1}, P_{2} \notin
C'$). Since the cumulative transition rate from $P' \in \cC_{\rm pc} \setminus C'$ to $C'$ is given by
$q(P', C') = \sum_{\alpha \in \mathcal{A}} q[P', C', \alpha]$, we have that for all $C' \in \cC /
{\approx_{\rm o}}$ with $C \neq C'$:
\[
q(P_{1}, C') \: = \: \sum_{\alpha \in \mathcal{A}} q[P_{1}, C', \alpha] \: = \: \sum_{\alpha \in
\mathcal{A}} q[P_{2}, C', \alpha] \: = \: q(P_{2}, C')
\]
which is precisely the condition of Definition~\ref{def:ordinary_lump}.
\qedhere
	\end{proof}

%
\subsection{From Outgoing to Incoming Cumulative Rates}
\label{subsec:from_out_to_in}
%

The main difference between bisimilarities based on exact lumpability and bisimilarities based on ordinary
lumpability is the consideration of incoming cumulative rates of the form $q[C, P, \alpha]$ instead of
outgoing cumulative rates of the form $q[P, C, \alpha]$. In the latter, the components that are naturally
considered in $C$ are those belonging to $\mathit{ds}(P)$. In the former, further components may come into
play, unless we restrict ourselves to components $P$ whose underlying CTMC is ergodic. The reason is that
the state corresponding to $P$ can be reached in one transition by infinitely many states corresponding to
components not necessarily belonging to $\mathit{ds}(P)$: think, e.g., of $a \, . \, P$ for all $a \in \cA
\times \cR$.

If we really consider those further components, then $q[C, P, \alpha]$ can only take two values, which are
$0$ in the case that no component in $C$ reaches $P$ via $\alpha$ and $\infty$ otherwise. For instance, if
$(\alpha, r) \, . \, P \in C$ then $(\alpha, r / n) \, . \, P + \ldots + (\alpha, r / n) \, . \, P \in C$ as
well for all $n \in \mathbb{N}_{\ge 2}$, where the last component has $n$ summands. Following the discussion
after Proposition~\ref{prop:exact_sufficient}, in the definition of exact bisimilarity we should require
$q[P_{1}, \alpha] = q[P_{2}, \alpha]$ and $q[C, P_{1}, \alpha] = q[C, P_{2}, \alpha]$. However, this would
boil down to requiring just $q[P_{1}, \alpha] = q[P_{2}, \alpha]$.

To avoid that as well as the limitation to components $P$ whose underlying CTMC is ergodic, i.e., to define
a more interesting bisimilarity in the spirit of exact lumpability over the entire set $\cC$ of components,
the incoming cumulative rates that should be considered are those of the form $q[C \cap \mathit{ds}(P), P,
\alpha]$. In this way, we rule out all components in $C$ that reach $P$ but are not reachable from $P$.

%
\subsection{Exact Bisimilarities}
\label{subsec:exact_bisim}
%

Strong and weak ordinary bisimilarities can be interpreted as the process algebraic counterparts of ordinary
lumpability: their defining conditions constrain the flows \emph{leaving} equivalent components. It is then
natural to ask for the counterparts of exact lumpability, i.e., bisimilarities constraining the
\emph{incoming} flows. These are known in the literature as \emph{exact performance equivalence}, introduced
in~\cite{Buc94b}, and \emph{exact equivalence}, introduced in~\cite{inf18}. We respectively rename them
strong and weak exact bisimilarities and redefine them by taking into account the remarks in
Section~\ref{subsec:from_out_to_in}.

	\begin{defi}[Strong exact bisimilarity]\label{def:strong_exact_bisim}
We say that $P_{1}, P_{2} \in \cC$ are \emph{strongly exact bisimilar}, written $P_{1} \sim_{\rm e} P_{2}$,
iff $(P_{1}, P_{2}) \in \cB$ for some strong exact bisimulation $\cB$. An equivalence relation $\cB
\subseteq \cC \times \cC$ is a \emph{strong exact bisimulation} iff, whenever $(P_{1}, P_{2}) \in \cB$, then
for all activity types $\alpha \in \cA$ and for equivalence classes $C \in \cC / \cB$:
\[
q[P_{1}, \alpha] \: = \: q[P_{2}, \alpha] \qquad
q[C \cap \mathit{ds}(P_{1}), P_{1}, \alpha] \: = \: q[C \cap \mathit{ds}(P_{2}), P_{2}, \alpha]
\]
	\end{defi}

The first condition requires equivalent components to execute activities of the same type with the same
total transition rate; the second one requires them to receive the same flow from every equivalence class.
Notice that the first condition cannot be dropped. Indeed, by summing the second condition over all the
equivalence classes, one only obtains the equality of the total incoming transition rates of $P_{1}$ and
$P_{2}$, which says nothing about the rates of the activities they can execute. This mirrors, at the level
of PEPA components, the situation discussed after Proposition~\ref{prop:exact_sufficient}: there, the
equality of the total exit rates had to be imposed as a separate condition, with the first condition above
playing precisely that role.

We demonstrate that strong exact bisimilarity induces an exact lumping on the underlying CTMC, thus allowing
the steady-state probability distribution of the original chain to be recovered from the aggregate one. This
holds over the set $\cC_{\rm pc, sf}$ of performance closed components whose derivation graphs are
\emph{selfloop free}. The reason is that rates labeling transitions from a state to itself are taken into
account by incoming cumulative rates of the form $q[C \cap \mathit{ds}(P), P, \alpha]$, while they are
ignored by the infinitesimal generator of the CTMC underlying $P$. For instance, if we consider the
derivation graph on the left in the figure of the forthcoming Example~\ref{ex:exact_vs_ordinary}, then
$A_{5} \sim_{\rm e} A'_{5}$ but the partition $\{ \{ A_{5}, A'_{5} \}, \{ A''_{5} \} \}$ induced over
$\mathit{ds}(A_{5})$ is not an exact lumping because $q(\{ A_{5}, A'_{5} \}, A_{5}) = 5 \neq 1 = q(\{ A_{5},
A'_{5} \}, A'_{5})$. A similar issue does not arise in the case of ordinary bisimilarities because ordinary
lumpability does not consider the rates of transitions between states in the same equivalence class.

	\begin{thm}\label{thm:strong_exact_bisim_induces}
For all $P \in \cC_{\rm pc, sf}$, $\sim_{\rm e}$ induces a partition of $\mathit{ds}(P)$ that is an exact
lumping.
	\end{thm}

	\begin{proof}
Consider the partition induced by $\sim_{\rm e}$ over $\mathit{ds}(P)$ and two arbitrary components $P_{1},
P_{2} \in \mathit{ds}(P) \cap C$ for some equivalence class $C \in \cC / {\sim_{\rm e}}$. From the incoming
condition in the definition of $\sim_{\rm e}$, for all equivalence classes $C' \in \cC / {\sim_{\rm e}}$ we
derive that $q[C' \cap \mathit{ds}(P_{1}), P_{1}, \alpha] = q[C' \cap \mathit{ds}(P_{2}), P_{2}, \alpha]$
for all $\alpha \in \mathcal{A}$. Since the cumulative transition rate from $C'$ to $P' \in \cC_{\rm pc,
sf}$ is given by $q(C', P') = \sum_{\alpha \in \mathcal{A}} q[C' \cap \mathit{ds}(P'), P', \alpha]$, we have
that for all $C' \in \cC / {\sim_{\rm e}}$:
\[
q(C', P_{1}) \: = \: \sum_{\alpha \in \mathcal{A}} q[C' \cap \mathit{ds}(P_{1}), P_{1}, \alpha] \: = \:
\sum_{\alpha \in \mathcal{A}} q[C' \cap \mathit{ds}(P_{2}), P_{2}, \alpha] \: = \: q(C', P_{2})
\]
which is precisely the condition of Definition~\ref{def:exact_lump}.
\qedhere
	\end{proof}

As in the duality between strong and weak ordinary bisimilarities, strong exact bisimilarity does not handle
$\tau$-activities in a special way, which is instead fundamental to define useful behavioural equivalences.
For this very reason, in~\cite{valuetools25-exact} the authors introduced the notion of \emph{exact lumpable
bisimilarity}. Its definition accomplishes the endeavour of both inducing an exact lumping on the underlying
CTMC, as we will prove, and allowing some unobservable behaviours to be ignored.

	\begin{defi}[Weak exact bisimilarity]\label{def:weak_exact_bisim}
We say that $P_{1}, P_{2} \in \cC$ are \emph{weakly exact bisimilar}, written $P_{1} \approx_{\rm e} P_{2}$,
iff $(P_{1}, P_{2}) \in \cB$ for some weak exact bisimulation $\cB$. An equivalence relation $\cB \subseteq
\cC \times \cC$ is a \emph{weak exact bisimulation} iff, whenever $(P_{1}, P_{2}) \in \cB$, then for all
activity types $\alpha \in \cA$ and for all equivalence classes $C \in \cC / \cB$:
		\begin{itemize}
\item either $\alpha \neq \tau$ and:
\[
q[P_{1}, \alpha] \: = \: q[P_{2}, \alpha] \qquad
q[C \cap \mathit{ds}(P_{1}), P_{1}, \alpha] \: = \: q[C \cap \mathit{ds}(P_{2}), P_{2}, \alpha]
\]

\item or $\alpha = \tau$ and:
			\begin{itemize}
\item either $P_{1}, P_{2} \notin C$ and:
\[
q[C \cap \mathit{ds}(P_{1}), P_{1}, \tau] \: = \: q[C \cap \mathit{ds}(P_{2}), P_{2}, \tau]
\]

\item or $P_{1}, P_{2} \in C$ and:
\[
q[C \cap \mathit{ds}(P_{1}), P_{1}, \tau] - q[P_{1}, \tau] \: = \: q[C \cap \mathit{ds}(P_{2}), P_{2}, \tau]
- q[P_{2}, \tau]
\]
			\end{itemize}
		\end{itemize}
	\end{defi}

The special treatment of $\tau$-activities (see the comments after Proposition~\ref{prop:exact_sufficient})
relaxes the constraints on the internal flows exchanged within an equivalence class, while preserving the
distinguishing feature of exact lumpability: the equiprobability, at steady-state, of the states belonging
to the same class. Thus, weak exact bisimilarity retains the fundamental property of the strong version.

	\begin{thm}\label{thm:weak_exact_bisim_induces}
For all $P \in \cC_{\rm pc, sf}$, $\approx_{\rm e}$ induces a partition of $\mathit{ds}(P)$ that is an exact
lumping.
	\end{thm}

	\begin{proof}
Consider the partition induced by $\approx_{\rm e}$ over $\mathit{ds}(P)$ and two arbitrary components
$P_{1}, P_{2} \in \mathit{ds}(P) \cap C$ for some equivalence class $C \in \cC / {\approx_{\rm e}}$:

		\begin{itemize}
\item From the incoming condition in the definition of $\approx_{\rm e}$, for all equivalence classes $C'
\in \cC / {\approx_{\rm e}}$ with $C \neq C'$ we derive that $q[C' \cap \mathit{ds}(P_{1}), P_{1}, \alpha] =
q[C' \cap \mathit{ds}(P_{2}), P_{2}, \alpha]$ for all $\alpha \in \mathcal{A}$ (in the case $\alpha = \tau$
note that $C \neq C'$ implies $P_{1}, P_{2} \notin C'$). Since the cumulative transition rate from $C'$ to
$P' \in \cC_{\rm pc} \setminus C'$ is given by $q(C', P') = \sum_{\alpha \in \mathcal{A}} q[C' \cap
\mathit{ds}(P'), P', \alpha]$, we have that for all $C' \in \cC / {\approx_{\rm e}}$ with $C \neq C'$:
\[
q(C', P_{1}) \: = \: \sum_{\alpha \in \mathcal{A}} q[C' \cap \mathit{ds}(P_{1}), P_{1}, \alpha] \: = \:
\sum_{\alpha \in \mathcal{A}} q[C' \cap \mathit{ds}(P_{2}), P_{2}, \alpha] \: = \: q(C', P_{2})
\]

\item From the other conditions in the definition of $\approx_{\rm e}$ and $P_{1}, P_{2} \in \cC_{\rm pc,
sf}$, we derive that:
\[\begin{array}{rcl}
q(C, P_{1}) & \!\! = \!\! & q(C \setminus \{ P_{1} \}, P_{1}) + q(P_{1}, P_{1}) \\
& \!\! = \!\! & (\sum_{\alpha \in \mathcal{A} \setminus \{ \tau \}} q[C \cap \mathit{ds}(P_{1}), P_{1},
\alpha] + q[C \cap \mathit{ds}(P_{1}), P_{1}, \tau]) \\
& & - \, (\sum_{\alpha \in \mathcal{A} \setminus \{ \tau \}} q[P_{1}, \alpha] + q[P_{1}, \tau]) \\
& \!\! = \!\! & \sum_{\alpha \in \mathcal{A} \setminus \{ \tau \}} (q[C \cap \mathit{ds}(P_{1}), P_{1},
\alpha] - q[P_{1}, \alpha]) \\
& & + \, (q[C \cap \mathit{ds}(P_{1}), P_{1}, \tau] - q[P_{1}, \tau]) \\
& \!\! = \!\! & \sum_{\alpha \in \mathcal{A} \setminus \{ \tau \}} (q[C \cap \mathit{ds}(P_{2}), P_{2},
\alpha] - q[P_{2}, \alpha]) \\
& & + \, (q[C \cap \mathit{ds}(P_{1}), P_{2}, \tau] - q[P_{2}, \tau]) \\
& \!\! = \!\! & (\sum_{\alpha \in \mathcal{A} \setminus \{ \tau \}} q[C \cap \mathit{ds}(P_{2}), P_{2},
\alpha] + q[C \cap \mathit{ds}(P_{2}), P_{2}, \tau]) \\
& & - \, (\sum_{\alpha \in \mathcal{A} \setminus \{ \tau \}} q[P_{2}, \alpha] + q[P_{2}, \tau]) \\
& \!\! = \!\! & q(C \setminus \{ P_{2} \}, P_{2}) + q(P_{2}, P_{2}) \\
& \!\! = \!\! & q(C, P_{2}) \\
\end{array}\]
		\end{itemize}
The two equalities above yield precisely the condition of Definition~\ref{def:exact_lump}.
\qedhere
	\end{proof}

%
\subsection{Strict Bisimilarities}
\label{subsec:strict_bisim}
%

In parallel with what we did for CTMCs in Section~\ref{sec:ctmc}, we now discuss the strict versions of
strong and weak bisimilarities. They are obtained by requiring at the same time the outgoing conditions of
ordinary bisimilarities and the incoming conditions of exact bisimilarities. The strong version was
originally introduced in~\cite{BR23} over a reversible stochastic process algebra. In the non-reversible
setting of PEPA, we have to take into account the remarks in Section~\ref{subsec:from_out_to_in}.

	\begin{defi}[Strong strict bisimilarity]\label{def:strong_strict_bisim}
We say that $P_{1}, P_{2} \in \cC$ are \emph{strongly strict bisimilar}, written $P_{1} \sim_{\rm s} P_{2}$,
iff $(P_{1}, P_{2}) \in \cB$ for some strong strict bisimulation $\cB$. An equivalence relation $\cB
\subseteq \cC \times \cC$ is a \emph{strong strict bisimulation} iff, whenever $(P_{1}, P_{2}) \in \cB$,
then for all activity types $\alpha \in \cA$ and for all equivalence classes $C \in \cC / \cB$:
\[
q[P_{1}, C, \alpha] \: = \: q[P_{2}, C, \alpha] \qquad
q[C \cap \mathit{ds}(P_{1}), P_{1}, \alpha] \: = \: q[C \cap \mathit{ds}(P_{2}), P_{2}, \alpha]
\]
	\end{defi}

Notice that the condition $q[P_{1}, \alpha] = q[P_{2}, \alpha]$ of Definition~\ref{def:strong_exact_bisim}
needs not be required explicitly. Since $q[P, \alpha] = \sum_{C \in \cC / \cB} q[P, C, \alpha]$, that
condition follows from the first condition above by summing over all equivalence classes.

	\begin{thm}\label{thm:strong_strict_bisim_induces}
For all $P \in \cC_{\rm pc, sf}$, $\sim_{\rm s}$ induces a partition of $\mathit{ds}(P)$ that is a strict
lumping.
	\end{thm}

	\begin{proof}
A straightforward consequence of Definition~\ref{def:strong_strict_bisim} and
Theorems~\ref{thm:strong_ordinary_bisim_induces} and~\ref{thm:strong_exact_bisim_induces}.
\qedhere
	\end{proof}

	\begin{defi}[Weak strict bisimilarity]\label{def:weak_strict_bisim}
We say that $P_{1}, P_{2} \in \cC$ are \emph{weakly strict bisimilar}, written $P_{1} \approx_{\rm s}
P_{2}$, iff $(P_{1}, P_{2}) \in \cB$ for some weak strict bisimulation $\cB$. An equivalence relation $\cB
\subseteq \cC \times \cC$ is a \emph{weak strict bisimulation} iff, whenever $(P_{1}, P_{2}) \in \cB$, then
for all activity types $\alpha \in \cA$ and for all equivalence classes $C \in \cC / \cB$:
		\begin{itemize}
\item either $\alpha \neq \tau$, or $\alpha = \tau$ and $P_{1}, P_{2} \notin C$, and:
\[
q[P_{1}, C, \alpha] \: = \: q[P_{2}, C, \alpha] \qquad
q[C \cap \mathit{ds}(P_{1}), P_{1}, \alpha] \: = \: q[C \cap \mathit{ds}(P_{2}), P_{2}, \alpha]
\]

\item or $\alpha = \tau$ and $P_{1}, P_{2} \in C$ and:
\[
q[C \cap \mathit{ds}(P_{1}), P_{1}, \tau] - q[P_{1}, \tau] \: = \: q[C \cap \mathit{ds}(P_{2}), P_{2}, \tau]
- q[P_{2}, \tau]
\]
		\end{itemize}
	\end{defi}

	\begin{thm}\label{thm:weak_strict_bisim_induces}
For all $P \in \cC_{\rm pc, sf}$, $\approx_{\rm s}$ induces a partition of $\mathit{ds}(P)$ that is a strict
lumping.
	\end{thm}

	\begin{proof}
A straightforward consequence of Definition~\ref{def:weak_strict_bisim} and
Theorems~\ref{thm:weak_ordinary_bisim_induces} and~\ref{thm:weak_exact_bisim_induces}.
\qedhere
	\end{proof}

%
%
\section{Taxonomy}
\label{sec:taxonomy}
%
%

We now compare the distinguishing power of the six stochastic bisimilarities over PEPA in order to develop
their taxonomy. We first study the taxonomy over the entire set $\cC$ of components
(Section~\ref{subsec:taxonomy_all_comp}), then we show how it changes over the set $\cC_{\rm n\tau}$ of
components with no $\tau$-activities and the set $\cC_{\rm nr}$ of components with no recursion
(Section~\ref{subsec:taxonomy_ntau_nr_comp}).

%
\subsection{Taxonomy over $\cC$}
\label{subsec:taxonomy_all_comp}
%

We introduce three pairs of components that will act as distinguishing witnesses for the taxonomy. The first
pair tells the weak bisimilarities apart from the strong ones: the two components differ only for an
activity of type $\tau$ that can be repeatedly performed within the class to which the components belong.

	\begin{exa}[Weak vs.\ strong]\label{ex:weak_vs_strong}
Consider the two recursive components:
\[
A_{1} \: \rmdef \: (\alpha, r_{\alpha}) \, . \, A_{1}
\qquad
A_{2} \: \rmdef \: (\alpha, r_{\alpha}) \, . \, A_{2} + (\tau, r_{\tau}) \, . \, A_{2}
\]
with $\alpha \neq \tau$ and $r_{\alpha}, r_{\tau} \in \mathbb{R}_{> 0}$. The derivation graph of the former
has a single transition, which is labelled with $\alpha$ and goes from the state corresponding to $A_{1}$ to
itself, whilst the derivation graph of the latter has two transitions, one labelled with $\alpha$ and the
other with $\tau$, both going from the state corresponding to $A_{2}$ to itself.

The equivalence relation whose only non-singleton equivalence class is $C = \{ A_{1}, A_{2} \}$ is a weak
ordinary/exact/strict bisimulation. In particular, either component in $C$ reaches $C$ via~$\alpha$ with
cumulative rate $r_{\alpha}$ and is reached from $C$ via $\alpha$ with cumulative rate $r_{\alpha}$.
Moreover, in the exact and strict cases, it additionally holds that $q[C \cap \mathit{ds}(A_{1}), A_{1},
\tau] - q[A_{1}, \tau] = 0 - 0 = r_{\tau} - r_{\tau} = q[C \cap \mathit{ds}(A_{2}), A_{2}, \tau] - q[A_{2},
\tau]$.

On the other hand, there is no strong ordinary/exact/strict bisimulation relating $A_{1}$ and~$A_{2}$. Any
equivalence class $C$ containing both would yield $q[A_{1}, C, \tau] = 0 \neq r_{\tau} = q[A_{2}, C, \tau]$
in the ordinary and strict cases and $q[A_{1}, \tau] = 0 \neq r_{\tau} = q[A_{2}, \tau]$ in the exact case.
	\end{exa}

Before moving to the next two pairs, we observe that the classical $\tau$-abstracting bisimulation
equivalences developed for nondeterministic process algebras, like weak bisimilarity~\cite{Mil89a} and
branching bisimilarity~\cite{GW96}, are capable of ignoring $\tau$-activities in various situations that
respect the branching structure of the considered process terms. A typical law they all satisfy is that
$\alpha \, . \, \tau \, . \, P$ is equivalent to $\alpha \, .  \, P$. This is preserved to some extent by
the weak Markovian bisimilarity of~\cite{Ber15}, where $(\alpha, r) \, . \, (\tau, r_{1}) \, . \, \cdots \,
. \, (\tau, r_{n}) \, . \, P$ is equivalent to $(\alpha, r) \, . \, (\tau, \bar{r}) \, . \, P$ with $1 /
\bar{r} = \sum_{1 \le i \le n} 1 / r_{i}$, i.e., a sequence of $n \in \mathbb{N}_{\ge 2}$ exponentially
timed $\tau$-activities can be replaced by a single $\tau$-activity that possesses the same average
duration. In contrast, that law is no longer satisfied by weak ordinary/exact/strict bisimilarities.

	\begin{exa}[Weak discrimination]\label{ex:no_alpha_tau_p}
Consider the two recursive components:
\[
A'_{1} \: \rmdef \: (\alpha, r) \, . \, (\tau, r_{\tau}) \, . \, (\alpha', r') \, . \, A'_{1}
\qquad
A'_{2} \: \rmdef \: (\alpha, r) \, . \, (\alpha', r') \, . \, A'_{2}
\]
with $\alpha, \alpha' \neq \tau$ and $r, r', r_{\tau} \in \mathbb{R}_{> 0}$. The derivation graph of the
former has a sequence formed by one $\alpha$-transition, one $\tau$-transition, and one $\alpha'$-transition
going back to the state corresponding to $A'_{1}$, whereas the derivation graph of the latter has a sequence
formed by one $\alpha$-transition and one $\alpha'$-transition going back to the state corresponding to
$A'_{2}$.

The only way to abstract from the $\tau$-activity is to consider an equivalence relation having $(\tau,
r_{\tau}) \, . \, (\alpha', r') \, . \, A'_{1}$ and $(\alpha', r') \, . \, A'_{1}$ in the same equivalence
class. However, this would not result in a weak ordinary/exact/strict bisimulation because $(\alpha', r') \,
. \, A'_{1}$ has an $\alpha'$-transition while $(\tau, r_{\tau}) \, . \, (\alpha', r') \, . \, A'_{1}$ has
not and hence the outgoing conditions of the three notions would not be met. Those two components would be
weakly equivalent not even in the case in which $\alpha' = \tau$, because $(\alpha', r') \, . \, A'_{1}$ has
a $\tau$-transition towards a different equivalence class while $(\tau, r_{\tau}) \, . \, (\alpha', r') \, .
\, A'_{1}$ has not.
	\end{exa}

The second pair features two components that are ordinary, but not exact, bisimilar: intuitively, they agree
on their outgoing flows, but not on their incoming flows.

	\begin{exa}[Ordinary vs.\ exact]\label{ex:ordinary_vs_exact}
Consider the two recursive components:
\[
A_{3} \: \rmdef \: (\alpha, r) \, . \, (\alpha', r') \, . \, A_{3} + (\alpha, r) \, . \, ((\alpha', r') \, .
\, A_{3} + \underline{0})
\qquad
A_{4} \: \rmdef \: (\alpha, 2 \cdot r) \, . \, (\alpha', r') \, . \, A_{4}
\]
with $\alpha, \alpha' \neq \tau$ and $r, r' \in \mathbb{R}_{> 0}$. The derivation graph of the former has
one $\alpha$-transition with rate $r$ to the state corresponding to $(\alpha', r') \, . \, A_{3}$ followed
by one $\alpha'$-transition with rate $r'$ that goes back to the state corresponding to $A_{3}$, together
with one $\alpha$-transition with rate $r$ to the state corresponding to $(\alpha', r') \, . \, A_{3} +
\underline{0}$ followed by one $\alpha'$-transition with rate $r'$ that goes back to the state corresponding
to $A_{3}$. The latter has one $\alpha$-transition with rate $2 \cdot r$ followed by one
$\alpha'$-transition with rate $r'$ that goes back to the state corresponding to $A_{4}$.

The equivalence relation whose only non-singleton equivalence classes are $C = \{ A_{3}, A_{4} \}$ and $C' =
\{ (\alpha', r') \, . \, A_{3}, (\alpha', r') \, . \, A_{3} + \underline{0}, (\alpha', r') \, . \, A_{4} \}$
is a strong ordinary bisimulation. Indeed, either component in $C$ reaches $C'$ via $\alpha$ with cumulative
rate $2 \cdot r$ and any component in $C'$ reaches~$C$ via $\alpha'$ with cumulative rate $r'$.

In contrast, no equivalence relation with the same equivalence classes $C$ and $C'$ is a strong exact
bisimulation because $q[C' \cap \mathit{ds}(A_{3}), A_{3}, \alpha'] = 2 \cdot r' \neq r' = q[C' \cap
\mathit{ds}(A_{4}), A_{4}, \alpha']$. If we refine $C'$ into $C'_{1} = \{ (\alpha', r') \, . \, A_{3} \}$
and $C'_{2} = \{ (\alpha', r') \, . \, A_{3} + \underline{0}, (\alpha', r') \, . \, A_{4} \}$, then
\linebreak $q[C'_{2} \cap \mathit{ds}(A_{3}), A_{3}, \alpha'] = r' = q[C'_{2} \cap \mathit{ds}(A_{4}),
A_{4}, \alpha']$, but $q[C'_{1} \cap \mathit{ds}(A_{3}), A_{3}, \alpha'] = r' \neq 0 = q[C'_{1} \cap
\mathit{ds}(A_{4}), A_{4}, \alpha']$.
	\end{exa}

The third pair contains two components that are exact, but not ordinary, bisimilar: unlike the previous two,
they agree on their incoming flows, but not on their outgoing flows.

	\begin{exa}[Exact vs.\ ordinary]\label{ex:exact_vs_ordinary}
Consider the two recursive components:
\begin{align*}
A_{5} & \: \rmdef \: (\alpha, 1) \, . \, A_{5} + (\alpha, 1) \, . \, A'_{5} + (\alpha, 8) \, . \, A''_{5} &
\qquad
A_{6} & \: \rmdef \: (\alpha, 5) \, . \, A_{6} + (\alpha, 5) \, . \, A'_{6} \\
A'_{5} & \: \rmdef \: (\alpha, 5) \, . \, A_{5} + (\alpha, 5) \, . \, A'_{5} &
A'_{6} & \: \rmdef \: (\alpha, 1) \, . \, A_{6} + (\alpha, 1) \, . \, A'_{6} + (\alpha, 8) \, . \, A''_{6}
\\
A''_{5} & \: \rmdef \: (\beta, 5) \, . \, A_{5} + (\beta, 5) \, . \, A'_{5} &
A''_{6} & \: \rmdef \: (\beta, 5) \, . \, A_{6} + (\beta, 5) \, . \, A'_{6}
\end{align*}
with $\alpha, \beta \neq \tau$, whose underlying derivation graphs are depicted in
Figure~\ref{fig:exact_vs_ordinary}.

	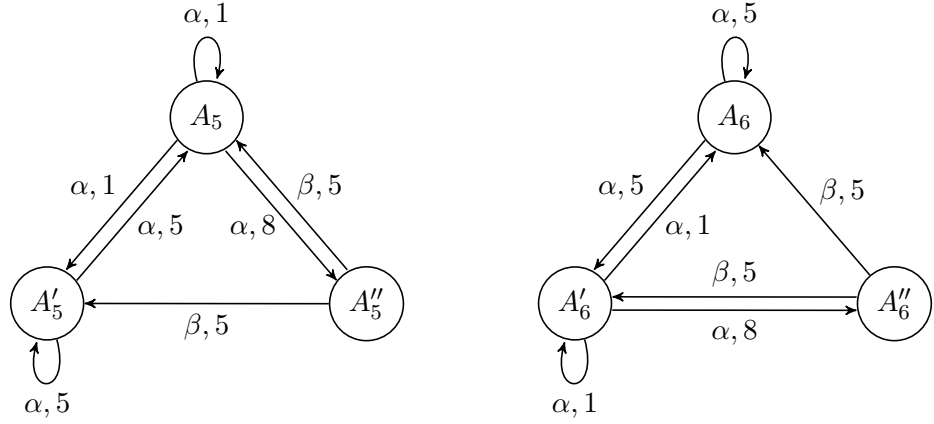
\begin{figure}[t]

\begin{center}
\begin{tikzpicture}[modal, node distance = 2cm]


\node[state] (a5)				          {$A_{5}$};
\node[state] (a5')  [below left = of a5, yshift = -10pt]  {$A'_{5}$};
\node[state] (a5'') [below right = of a5, yshift = -10pt] {$A''_{5}$};


\node[state] (a6)   [right = of a5, xshift = 4cm]         {$A_{6}$};
\node[state] (a6')  [below left = of a6, yshift = -10pt]  {$A'_{6}$};
\node[state] (a6'') [below right = of a6, yshift = -10pt] {$A''_{6}$};


\path[->] (a5) edge[loop above] node[above] {$\alpha, 1$} (a5);

\path[->] (a5)  edge[transform canvas = {xshift = -2pt, yshift = 2pt}] 
                node[above left, xshift = 2pt, yshift = -2pt] {$\alpha, 1$} (a5');

\path[->] (a5') edge[transform canvas = {xshift = 2pt, yshift = -2pt}] 
                node[below right, xshift = -2pt, yshift = 4pt] {$\alpha, 5$} (a5);

\path[->] (a5'') edge[transform canvas = {xshift = 2pt, yshift = 2pt}] 
                 node[above right, xshift = -2pt, yshift = -2pt] {$\beta, 5$} (a5);

\path[->] (a5)   edge[transform canvas = {xshift = -2pt, yshift = -2pt}] 
                 node[below left, xshift = 2pt, yshift = 4pt] {$\alpha, 8$} (a5'');

\path[->] (a5') edge[loop below] node[below] {$\alpha, 5$} (a5');

\path[->] (a5'') edge node[below] {$\beta, 5$} (a5');


\path[->] (a6)   edge[loop above] node[above] {$\alpha, 5$} (a6);

\path[->] (a6)   edge[transform canvas = {xshift = -2pt, yshift = 2pt}] 
                 node[above left, xshift = 2pt, yshift = -2pt] {$\alpha, 5$} (a6');

\path[->] (a6')  edge[transform canvas = {xshift = 2pt, yshift = -2pt}] 
                 node[below right,xshift = -2pt, yshift = 4pt] {$\alpha, 1$} (a6);

\path[->] (a6')  edge[transform canvas = {yshift = -2.5pt}]
                 node[below] {$\alpha, 8$} (a6'');

\path[->] (a6'') edge[transform canvas = {yshift = 2.5pt}] 
                 node[above] {$\beta, 5$} (a6');

\path[->] (a6'') edge node[above right, xshift = -2pt, yshift = -2pt] {$\beta, 5$} (a6);

\path[->] (a6')  edge[loop below] node[below] {$\alpha, 1$} (a6');

\end{tikzpicture}
\end{center}

\caption{Derivation graphs for the components in Example~\ref{ex:exact_vs_ordinary}.}
\label{fig:exact_vs_ordinary}

	\end{figure}

The equivalence relation whose only non-singleton equivalence classes are $C_{\alpha} = \{ A_{5}, A'_{5},
\linebreak A_{6}, A'_{6} \}$ and $C_{\beta} = \{ A''_{5}, A''_{6} \}$ is a strong exact bisimulation. In
particular, for all $P_{\alpha} \in C_{\alpha}$ it holds that $q[P_{\alpha}, \alpha] = 10$, $q[P_{\alpha},
\beta] = 0$, $q[C_{\alpha} \cap \mathit{ds}(P_{\alpha}), P_{\alpha}, \alpha] = 6$, and $q[C_{\beta} \cap
\mathit{ds}(P_{\alpha}), P_{\alpha}, \beta] = 5$, while for all $P_{\beta} \in C_{\beta}$ it holds that
$q[P_{\beta}, \alpha] = 0$, $q[P_{\beta}, \beta] = 10$, $q[C_{\alpha} \cap \mathit{ds}(P_{\beta}),
P_{\beta}, \alpha] = 8$, and $q[C_{\beta} \cap \mathit{ds}(P_{\beta}), P_{\beta}, \beta] = 0$.

On the other hand, the equivalence relation above is not a strong ordinary bisimilarity because $q[A_{5},
C_{\beta}, \alpha] = 8 \neq 0 = q[A_{6}, C_{\beta}, \alpha]$. If we refine $C_{\beta}$ then $q[A_{5}, \{
A''_{6} \}, \alpha] = 0 = q[A_{6}, \{ A''_{6} \}, \alpha]$ but $q[A_{5}, \{ A''_{5} \}, \alpha] = 8 \neq 0 =
q[A_{6}, \{ A''_{5} \}, \alpha]$.
	\end{exa}

We are now ready to state the main result of this section.

	\begin{thm}[Taxonomy over $\cC$]\label{thm:taxonomy}
The following relationships hold among the six stochastic bisimilarities:
\[
\sim_{\rm s} {\subset} \sim_{\rm o} {\subset} \approx_{\rm o} \qquad
\sim_{\rm s} {\subset} \sim_{\rm e} {\subset} \approx_{\rm e} \qquad
\approx_{\rm s} {\subset} \approx_{\rm o} \qquad
\approx_{\rm s} {\subset} \approx_{\rm e} \qquad
\sim_{\rm s} {\subset} \approx_{\rm s}
\]
Moreover, no other inclusion holds between any two of the six bisimilarities.
	\end{thm}

	\begin{proof}
We first prove the inclusions by showing that every bisimulation of the finer kind is also a bisimulation
of the coarser kind (the inclusion between the corresponding bisimilarities then follows by taking the
largest bisimulations):

		\begin{itemize}
\item $\sim_{\rm s} {\subseteq} \sim_{\rm o}$ and $\approx_{\rm s} {\subseteq} \approx_{\rm o}$: the
conditions of Definitions~\ref{def:strong_ordinary_bisim} and~\ref{def:weak_ordinary_bisim} coincide with
the outgoing conditions of Definitions~\ref{def:strong_strict_bisim} and~\ref{def:weak_strict_bisim},
respectively.

\item $\sim_{\rm s} {\subseteq} \sim_{\rm e}$ and $\approx_{\rm s} {\subseteq} \approx_{\rm e}$: the
incoming conditions of Definitions~\ref{def:strong_exact_bisim} and~\ref{def:weak_exact_bisim} respectively
coincide with the incoming conditions of Definitions~\ref{def:strong_strict_bisim}
and~\ref{def:weak_strict_bisim}, while the outgoing conditions $q[P_{1}, \alpha] = q[P_{2}, \alpha]$ of the
former definitions follow by summing the outgoing conditions $q[P_{1}, C, \alpha] = q[P_{2}, C,\alpha]$ of
the latter definitions over all equivalence classes $C$.

\item $\sim_{\rm o} {\subseteq} \approx_{\rm o}$, $\sim_{\rm e} {\subseteq} \approx_{\rm e}$, and $\sim_{\rm
s} {\subseteq} \approx_{\rm s}$: weak bisimulations relax the conditions of the corresponding strong ones on
the $\tau$-flows within every class of related components. In particular, for the weak exact and weak strict
conditions on the class $C$ containing $P_{1}$ and~$P_{2}$, note that a strong relation guarantees both $q[C
\cap \mathit{ds}(P_{1}), P_{1}, \tau] = q[C \cap \mathit{ds}(P_{2}), P_{2}, \tau]$ and $q[P_{1}, \tau] =
q[P_{2}, \tau]$, hence $q[C \cap \mathit{ds}(P_{1}), P_{1}, \tau] - q[P_{1}, \tau] = q[C \cap
\mathit{ds}(P_{2}), P_{2}, \tau] - q[P_{2}, \tau]$.
		\end{itemize}

The strictness of the inclusions follows from the three pairs of witnesses in the previous examples:

		\begin{itemize}
\item $A_{1}$ and $A_{2}$ of Example~\ref{ex:weak_vs_strong} are related by $\approx_{\rm o}$, $\approx_{\rm
e}$, $\approx_{\rm s}$ and distinguished by $\sim_{\rm o}$, $\sim_{\rm e}$, $\sim_{\rm s}$.

\item $A_{3}$ and $A_{4}$ of Example~\ref{ex:ordinary_vs_exact} are related by $\sim_{\rm o}$, $\approx_{\rm
o}$ and distinguished by $\sim_{\rm e}$, $\approx_{\rm e}$, $\sim_{\rm s}$, $\approx_{\rm s}$.

\item $A_{5}$ and $A_{6}$ of Example~\ref{ex:exact_vs_ordinary} are related by $\sim_{\rm e}$, $\approx_{\rm
e}$ and distinguished by $\sim_{\rm o}$, $\approx_{\rm o}$, $\sim_{\rm s}$, $\approx_{\rm s}$.
		\end{itemize}

Finally, the absence of any further inclusion follow from the same examples.
\qedhere
	\end{proof}

An immediate consequence of Theorem~\ref{thm:taxonomy} is that $\sim_{\rm o}$, $\sim_{\rm e}$, and
$\approx_{\rm s}$ are pairwise incomparable. Each of them observes a genuinely different aspect of the
behaviour of a component: its outgoing flows, its incoming flows, and both of them up to unobservable
activities, respectively. Also $\approx_{\rm o}$ and $\approx_{\rm e}$ are incomparable.
Figure~\ref{fig:spectrum} depicts the resulting taxonomy.

	\begin{figure}[t]

\begin{center}
\begin{tikzpicture}[modal]

\node (strict)				               {\textcolor{orange}{$\sim_{\rm s}$}};
\node (strong)   [below left = of strict]              {\textcolor{blue}{$\sim_{\rm o}$}};
\node (exact)    [below right = of strict]             {\textcolor{green}{$\sim_{\rm e}$}};
\node (lump)     [below = of strong]                   {\textcolor{blue}{$\approx_{\rm o}$}};
\node (w_exact)  [below = of exact]                    {\textcolor{green}{$\approx_{\rm e}$}};
\node (w_strict) [below = of strict, yshift = -0.05cm] {\textcolor{orange}{$\approx_{\rm s}$}};

\path[->] (strict)   edge (strong);
\path[->] (strict)   edge (exact);
\path[->] (strict)   edge (w_strict);
\path[->] (strong)   edge (lump);
\path[->] (exact)    edge (w_exact);
\path[->] (w_strict) edge (lump);
\path[->] (w_strict) edge (w_exact);

\path[<->, brown] (strict)   edge[bend left] (exact);
\path[<->, brown] (w_strict) edge[bend left] (w_exact);

\end{tikzpicture}
\end{center}

\caption{Taxonomy of the six stochastic bisimilarities. Arrows denote strict inclusion, while the absence of
directed paths between two equivalences denotes incomparability. Blue, orange, and green mark the
bisimilarities inducing, respectively, an ordinary, strict, or exact lumping on the underlying CTMC. The
brown arrows refer to components whose underlying CTMC is time reversible, over which the lumpings induced
by exact and strict bisimilarities coincide (Proposition~\ref{prop:rev_collapse}), although the equivalences
themselves remain distinct (Example~\ref{ex:rev_no_collapse}).}
\label{fig:spectrum}

	\end{figure}
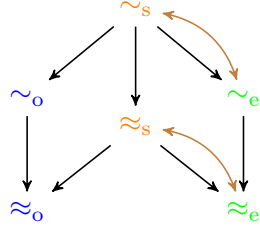

Finally, let us consider the special case of components whose underlying CTMCs are time reversible. At the
level of the induced lumpings, the exact and strict notions coincide.

	\begin{prop}\label{prop:rev_collapse}
Let $P \in \cC_{\rm pc}$ be such that its underlying CTMC is time reversible. Then the partitions of
$\mathit{ds}(P)$ induced by $\sim_{\rm e}$ and $\approx_{\rm e}$ are strict lumpings.
	\end{prop}

	\begin{proof}
By Theorems~\ref{thm:strong_exact_bisim_induces} and~\ref{thm:weak_exact_bisim_induces}, the partition
induced by $\sim_{\rm e}$ and $\approx_{\rm e}$ on the CTMC underlying $P$ is an exact lumping. Since the
CTMC is time reversible, the partition is also a strict lumping due to Theorem~\ref{thm:rev_exact_strict}.
\qedhere
	\end{proof}

Note that this fact concerns the induced lumpings only, as in general the bisimilarities themselves remain
distinct even over time-reversible CTMCs. The reason is that bisimilarities observe the activity types
labelling transitions too, whereas reversibility constrains only the cumulative rates of the underlying
CTMC. The following example distinguishes $\sim_{\rm e}$ from $\sim_{\rm s}$, as well as $\approx_{\rm e}$
from $\approx_{\rm s}$, over two components with underlying time-reversible CTMCs.

	\begin{exa}\label{ex:rev_no_collapse}
Consider the recursive components:
\begin{align*}
A_{7} &\rmdef (\alpha, r) \, . \, A_{9} + (\beta, r) \, . \, A_{10} &
A_{8} &\rmdef (\alpha, r) \, . \, A_{10} + (\beta, r) \, . \, A_{9} \\
A_{9} &\rmdef (\alpha, r) \, . \, A_{7} + (\alpha, r) \, . \, A_{8} &
A_{10} &\rmdef (\beta, r) \, . \, A_{7} + (\beta, r) \, . \, A_{8}
\end{align*}
with $\alpha, \beta \neq \tau$ and $r \in \mathbb{R}_{> 0}$, whose underlying derivation graph is depicted
in Figure~\ref{fig:rev_no_collapse}.

	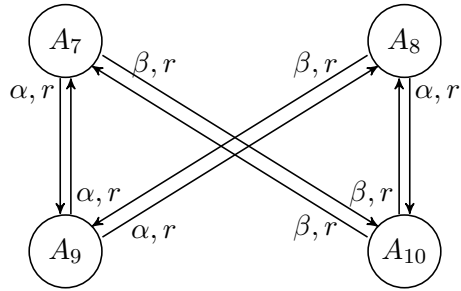
\begin{figure}[b]

\begin{center}
\begin{tikzpicture}[modal, node distance = 1.8cm]


\node[state] (a7)		         {\makebox[-4pt][c]{$A_{7}$}};
\node[state] (a8)  [right = 3.5cm of a7] {\makebox[-4pt][c]{$A_{8}$}};
\node[state] (a9)  [below = of a7]       {\makebox[-4pt][c]{$A_{9}$}};
\node[state] (a10) [right = 3.5cm of a9] {\makebox[-4pt][c]{$A_{10}$}};


\path[->] (a7)  edge[transform canvas = {xshift = -2pt}] 
                node[above left, xshift = 2pt, yshift = 12pt] {$\alpha, r$} (a9);

\path[->] (a9)  edge[transform canvas = {xshift = 2pt}] 
                node[below right, xshift = -2pt, yshift = -12pt] {$\alpha, r$} (a7);

\path[->] (a8)  edge[transform canvas = {xshift = 2pt}] 
                node[above right, xshift = -2pt, yshift = 12pt] {$\alpha, r$} (a10);

\path[->] (a10) edge[transform canvas = {xshift = -2pt}] 
                node[below left, xshift = 2pt, yshift = -9pt] {$\beta, r$} (a8);

\path[->] (a7)  edge[transform canvas = {xshift = 2pt, yshift = 2pt}] 
                node[above left, xshift = -20pt, yshift = 20pt] {$\beta, r$} (a10);

\path[->] (a10) edge[transform canvas = {xshift = -2pt, yshift = -2pt}] 
                node[below right, xshift = 20pt, yshift = -20pt] {$\beta, r$} (a7);

\path[->] (a8)  edge[transform canvas = {xshift = -2pt, yshift = 2pt}] 
                node[above right, xshift = 20pt, yshift = 20pt] {$\beta, r$} (a9);
                 
\path[->] (a9)  edge[transform canvas = {xshift = 2pt, yshift = -2pt}] 
                node[below left, xshift = -20pt, yshift = -24pt] {$\alpha, r$} (a8);

\end{tikzpicture}
\end{center}

\caption{Derivation graph for the components in Example~\ref{ex:rev_no_collapse}.}
\label{fig:rev_no_collapse}

	\end{figure}

The underlying CTMC has symmetric rates, i.e., $q(s, s') = q(s', s)$ for all distinct states $s, s' \in \{
A_{7}, A_{8}, A_{9}, A_{10} \}$. Therefore, the detailed balance equations hold with the uniform state
probability distribution (which assigns probability $1 / 4$ to each state), meaning that the CTMC is time
reversible due to Theorem~\ref{thm:detailed_balance}.

Notice that the partition $\{ \{ A_{7}, A_{8} \}, \{ A_{9} \}, \{ A_{10} \} \}$ is a strict lumping of the
underlying CTMC. Indeed, for the outgoing condition it holds that $q(A_{7}, \{ A_{9} \}) = r = q(A_{8}, \{
A_{9} \})$ and $q(A_{7}, \{ A_{10} \}) = r = q(A_{8}, \{ A_{10} \})$, while for the incoming condition we
have that $q(\{ A_{7}, A_{8} \}, A_{7}) = 0 = q(\{ A_{7}, A_{8} \}, A_{8})$, $q(\{ A_{9} \}, A_{7}) = r =
q(\{ A_{9} \}, A_{8})$, and $q(\{ A_{10} \}, A_{7}) = r = q(\{ A_{10} \}, A_{8})$.

That partition can be used as a basis to derive a strong exact bisimulation identifying $A_{7}$ and $A_{8}$,
because each of them can be reached with the same cumulative rate $r$ both via~$\alpha$ from the equivalence
class of $A_{9}$ and via $\beta$ from the equivalence class of $A_{10}$. In contrast, the same partition
cannot yield a strong ordinary bisimulation because $A_{7}$ (resp.\ $A_{8}$) reaches via $\alpha$ (resp.\
$\beta$) only the equivalence class of $A_{9}$ (resp.\ $A_{10}$). To overcome this, $A_{9}$ and $A_{10}$
could be placed in the same equivalence class, but then $A_{9}$ (resp.\ $A_{10}$) would reach via $\alpha$
(resp.\ $\beta$) the equivalence class of $A_{7}$ and $A_{8}$ with cumulative rate $2 \cdot r$, while
$A_{10}$ (resp.\ $A_{9}$) would do with cumulative rate $0$.
	\end{exa}

	\begin{rem}\label{rem:single_type}
The distinction disappears when the activity types convey no information. \linebreak If all the activities
of the components at hand share a single action type, then the per-type conditions of bisimilarities boil
down to the overall conditions of lumpings. In that case, the argument proving
Theorem~\ref{thm:rev_exact_strict} lifts from the underlying CTMCs to the components, yielding $\sim_{\rm e}
{=} \sim_{\rm s}$ and $\approx_{\rm e} {=} \approx_{\rm s}$ under time reversibility.
	\end{rem}

%
\subsection{Taxonomies over $\cC_{\rm n\tau}$ and $\cC_{\rm nr}$}
\label{subsec:taxonomy_ntau_nr_comp}
%

As expected, for each of the three pairs of stochastic bisimilarities, the strong version and the weak
version coincide in the absence of $\tau$-activities within the considered components.

	\begin{thm}[Taxonomy over $\cC_{\rm n\tau}$]\label{thm:taxonomy_nt}
The following relationships hold among the six stochastic bisimilarities over $\cC_{\rm n\tau}$:
\[
\sim_{\rm s} {\subset} \sim_{\rm o} {=} \approx_{\rm o} \qquad
\sim_{\rm s} {\subset} \sim_{\rm e} {=} \approx_{\rm e} \qquad
\sim_{\rm s} {=} \approx_{\rm s} \qquad
\]
	\end{thm}

	\begin{proof}
It follows from the fact that all the conditions on $\tau$-activities appearing in
Definitions~\ref{def:weak_ordinary_bisim}, \ref{def:weak_exact_bisim}, and~\ref{def:weak_strict_bisim} are
vacuously satisfied and the remaining conditions coincide with those of the corresponding strong
equivalences in Definitions~\ref{def:strong_ordinary_bisim}, \ref{def:strong_exact_bisim},
and~\ref{def:strong_strict_bisim}.
\qedhere
	\end{proof}

Given their relationships with ordinary/exact/strict lumpability, the weak bisimilarities considered in this
paper are more restrictive than the classical ones. They abstract from the cumulative rates of
$\tau$-activities as long as these activities do not cause the departure from an equivalence class. This has
to do with repetitive behaviors, which in process algebras are expressed through recursion, and affects the
strong bisimilarities considered in this paper too, as ordinary and exact bisimilarities are no longer
incomparable in the absence of recursion. It is therefore instructive to see what happens if we rule out
recursion with a few examples, which revisit those of Section~\ref{subsec:taxonomy_all_comp}.

	\begin{exa}[Weak vs.\ strong revisited]\label{ex:weak_vs_strong_nr}
Consider the two non-recursive components:
\[
P_{1} \: = \: (\alpha, r_{\alpha}) \, . \, \underline{0}
\qquad
P_{2} \: = \: (\alpha, r_{\alpha}) \, . \, \underline{0} + (\tau, r_{\tau}) \, . \, \underline{0}
\]
with $\alpha \neq \tau$ and $r_{\alpha}, r_{\tau} \in \mathbb{R}_{> 0}$. The derivation graph of the former
has a single transition, which is labelled with $\alpha$ and goes to the terminal state corresponding to
$\underline{0}$, whilst the derivation graph of the latter has two transitions, one labelled with $\alpha$
and the other with $\tau$, both going to the terminal state corresponding to $\underline{0}$.

No equivalence relation equating $P_{1}$ and $P_{2}$ is a weak ordinary/exact/strict bisimulation. In the
ordinary and strict cases, $P_{1}$ reaches via $\tau$ the equivalence class of $\underline{0}$ with
cumulative rate $0$, while $P_{2}$ reaches via $\tau$ the same equivalence class with cumulative rate
$r_{\tau}$; to avoid such a comparison, $\underline{0}$ should be in the same equivalence class as $P_{1}$
and $P_{2}$, but this is not possible because, unlike $P_{1}$ and $P_{2}$, $\underline{0}$ cannot execute
any $\alpha$-activity. In the exact case, the total rate of $P_{1}$ via $\tau$ is $0$, while for $P_{2}$ it
is $r_{\tau}$; to avoid such a comparison, again $\underline{0}$ should be in the same equivalence class as
$P_{1}$ and $P_{2}$, but this is not possible due to $\alpha$-activities.

A fortiori, there is no strong ordinary/exact/strict bisimulation relating $P_{1}$ and $P_{2}$.
	\end{exa}

	\begin{exa}[Weak discrimination revisited]\label{ex:no_alpha_tau_p_nr}
Consider the two non-recursive components:
\[
P'_{1} \: = \: (\alpha, r) \, . \, (\tau, r_{\tau}) \, . \, \underline{0}
\qquad
P'_{2} \: = \: (\alpha, r) \, . \, \underline{0}
\]
with $\alpha \neq \tau$ and $r, r_{\tau} \in \mathbb{R}_{> 0}$. The derivation graph of the former has a
sequence formed by one $\alpha$-transition and one $\tau$-transition going to the terminal state
corresponding to $\underline{0}$, whereas the derivation graph of the latter has one $\alpha$-transition
going to the terminal state corresponding to $\underline{0}$.

The equivalence relation whose only non-singleton equivalence classes are $C = \{ P'_{1}, P'_{2} \}$ and $C'
= \{ (\tau, r_{\tau}) \, . \, \underline{0}, \underline{0} \}$ is a weak ordinary bisimulation. Indeed,
either component in~$C$ reaches $C'$ via $\alpha$ with cumulative rate $r$ and either component in $C'$
reaches no other class via an activity of type different from $\tau$. Note that it is not a strong ordinary
bisimulation. Furthermore $(\tau, r_{\tau}) \, . \, \underline{0}$ and $\underline{0}$ are weakly ordinary
bisimilar too.

In contrast, no equivalence relation with the same equivalence classes $C$ and $C'$ is a weak exact/strict
bisimulation because $q[C \cap \mathit{ds}((\tau, r_{\tau}) \, . \, \underline{0}), (\tau, r_{\tau}) \, . \,
\underline{0}, \alpha] = 0 = q[C \cap \mathit{ds}(\underline{0}), \underline{0}, \alpha]$ but $q[C' \cap
\mathit{ds}((\tau, r_{\tau}) \, . \, \underline{0}), (\tau, r_{\tau}) \, . \, \underline{0}, \tau] -
q[(\tau, r_{\tau}) \, . \, \underline{0}, \tau] = 0 - r_{\tau} \neq 0 - 0 = q[C' \cap
\mathit{ds}(\underline{0}), \underline{0}, \tau] - q[\underline{0}, \tau]$ (hence $(\tau, r_{\tau}) \, . \,
\underline{0}$ and $\underline{0}$ cannot be related by any weak exact/strict bisimulation). If we refine
$C'$ into $\{ (\tau, r_{\tau}) \, . \, \underline{0} \}$ and $\{ \underline{0} \}$ then we obtain a weak
exact bisimulation relating $P'_{1}$ and $P'_{2}$ that provides no $\tau$-abstraction capability though,
i.e., a strong exact bisimulation.

Also note that $(\tau, r_{1}) \, . \, \underline{0}$ and $(\tau, r_{2}) \, . \, \underline{0}$ are related
by weak exact/strict bisimilarity iff $r_{1} = r_{2}$. Indeed, if they belong to the same equivalence class
$C$, then it must be the case that $q[C \cap \mathit{ds}((\tau, r_{1}) \, . \, \underline{0}), (\tau, r_{1})
\, . \, \underline{0}, \tau] - q[(\tau, r_{1}) \, . \, \underline{0}, \tau] = 0 - r_{1} = 0 - r_{2} = q[C
\cap \mathit{ds}((\tau, r_{2}) \, . \, \underline{0}), \underline{0}, \tau] - q[(\tau, r_{2}) \, . \,
\underline{0}, \tau]$.
	\end{exa}

	\begin{exa}[Ordinary vs.\ exact revisited]\label{ex:ordinary_vs_exact_nr}
Consider the two non-recursive components:
\[
P_{3} \: = \: (\alpha, r) \, . \, \underline{0} + (\alpha, r) \, . \, (\underline{0} + \underline{0})
\qquad
P_{4} \: = \: (\alpha, 2 \cdot r) \, . \, (\underline{0} + \underline{0} + \underline{0})
\]
with $\alpha \neq \tau$ and $r \in \mathbb{R}_{> 0}$. The derivation graph of the former has one
$\alpha$-transition with rate $r$ to the terminal state corresponding to $\underline{0}$ and one
$\alpha$-transition with rate $r$ to the terminal state corresponding to $\underline{0} + \underline{0}$,
whereas the derivation graph of the latter has just one transition, with type $\alpha$ and rate $2 \cdot r$,
to the terminal state corresponding to $\underline{0} + \underline{0} + \underline{0}$ (see
Example~\ref{ex:ord_exact_strict_lump_diff}).

The equivalence relation whose only non-singleton equivalence classes are $C = \{ P_{3}, P_{4} \}$ and $C' =
\{ \underline{0}, \underline{0} + \underline{0}, \underline{0} + \underline{0} + \underline{0} \}$ is a
strong ordinary bisimulation. Indeed, either component in~$C$ reaches $C'$ via $\alpha$ with cumulative rate
$2 \cdot r$ and any component in $C'$ reaches no class.

The same equivalence relation is a strong exact/strict bisimulation too, because \linebreak $q[C \cap
\mathit{ds}(\underline{0}), \underline{0}, \alpha'] = q[C \cap \mathit{ds}(\underline{0} + \underline{0}),
\underline{0} + \underline{0}, \alpha'] = q[C \cap \mathit{ds}(\underline{0} + \underline{0} +
\underline{0}), \underline{0} + \underline{0} + \underline{0}, \alpha'] = 0$ for all $\alpha' \in \cA$.
	\end{exa}

	\begin{exa}[Exact vs.\ ordinary revisited]\label{ex:exact_vs_ordinary_nr}
Consider the two non-recursive components:
\[
P_{5} \: = \: (\alpha, r) \, . \, (\alpha', r') \, . \, \underline{0}
\qquad
P_{6} \: = \: (\alpha, r) \, . \, (\alpha'', r'') \, . \, \underline{0}
\]
with $\alpha, \alpha', \alpha'' \neq \tau$, $r, r', r'' \in \mathbb{R}_{> 0}$, and $\alpha' \neq \alpha''$
or $r' \neq r''$. The derivation graph of the former has one $\alpha$-transition with rate $r$ followed by
one $\alpha'$-transition with rate $r'$ to the terminal state corresponding to $\underline{0}$, whilst the
derivation graph of the latter has one $\alpha$-transition with rate $r$ followed by one
$\alpha''$-transition with rate $r''$ to the terminal state corresponding to $\underline{0}$.

The equivalence relation whose only non-singleton equivalence class is $C = \{ P_{5}, P_{6} \}$ is a strong
exact bisimulation. Indeed, either component in $C$ has the same total exit rate $r$ via $\alpha$ and $0$
via any other action type and is reached with cumulative rate $0$ from any equivalence class via any action
type.

In contrast, the same equivalence relation is not a strong ordinary bisimulation because $q[P_{5}, \{
(\alpha', r') \, . \, \underline{0} \}, \alpha] \neq q[P_{6}, \{ (\alpha', r') \, . \, \underline{0} \},
\alpha]$ and $q[P_{5}, \{ (\alpha'', r'') \, . \, \underline{0} \}, \alpha] \neq q[P_{6}, \{ (\alpha'', r'')
\, . \, \underline{0} \}, \alpha]$ due to $\alpha' \neq \alpha''$ or $r' \neq r''$.
	\end{exa}

We conclude by showing how the taxonomy changes if we restrict ourselves to $\cC_{\rm nr}$. In particular,
exact bisimilarity boils down to require, for each action type, the same total exit rate via that type.
Likewise, only the outgoing condition matters for strict bisimilarity.

	\begin{lem}\label{lem:incoming_rate_nr}
Let $\sim$ be an equivalence relation over $\cC$. Then $q[C \cap \mathit{ds}(P), P, \alpha] = 0$ for all $P
\in \cC_{\rm nr}$, $\alpha \in \cA$, and $C \in \cC / {\sim}$.
	\end{lem}

	\begin{proof}
A straightforward consequence of the fact that the derivation graph of every $P \in \cC_{\rm nr}$ is a
directed acyclic graph, hence no component in $\mathit{ds}(P)$ has a transition to $P$.
\qedhere
	\end{proof}

	\begin{thm}[Taxonomy over $\cC_{\rm nr}$]\label{thm:taxonomy_nr}
The following relationships hold among the six stochastic bisimilarities over $\cC_{\rm nr}$:
\[
\sim_{\rm s} {=} \sim_{\rm o} {\subset} \sim_{\rm e} \qquad
\approx_{\rm s} {\subset} \approx_{\rm o} \qquad
\approx_{\rm s} {\subset} \approx_{\rm e} \qquad
\sim_{\rm o} {\subset} \approx_{\rm o} \qquad
\sim_{\rm e} {=} \approx_{\rm e} \qquad
\sim_{\rm s} {=} \approx_{\rm s}
\]
	\end{thm}

	\begin{proof}
We have that:

		\begin{itemize}
\item $\sim_{\rm s} {=} \sim_{\rm o}$ stems from Lemma~\ref{lem:incoming_rate_nr}.

\item $\sim_{\rm o} {\subset} \sim_{\rm e}$ stems from Lemma~\ref{lem:incoming_rate_nr} with strictness
being witnessed by Example~\ref{ex:exact_vs_ordinary_nr}.

\item $\approx_{\rm s} {\subset} \approx_{\rm o}$ stems from Theorem~\ref{thm:taxonomy} with strictness
being witnessed by Example~\ref{ex:no_alpha_tau_p_nr}.

\item $\approx_{\rm s} {\subset} \approx_{\rm e}$ stems from Theorem~\ref{thm:taxonomy} with strictness
being witnessed by Example~\ref{ex:exact_vs_ordinary_nr}.

\item $\sim_{\rm o} {\subset} \approx_{\rm o}$ stems from Theorem~\ref{thm:taxonomy} with strictness being
witnessed by Example~\ref{ex:no_alpha_tau_p_nr}.

\item $\sim_{\rm e} {\subseteq} \approx_{\rm e}$ stems from Theorem~\ref{thm:taxonomy}. \\
Suppose now that $P_{1} \approx_{\rm e} P_{2}$ for two arbitrary $P_{1}, P_{2} \in \cC_{\rm nr}$.
Then:

			\begin{itemize}
\item $q[P_{1}, \alpha] = q[P_{2}, \alpha]$ for all $\alpha \in \cA \setminus \{ \tau \}$.

\item $q[P_{1}, \tau] = q[P_{2}, \tau]$ because, if we denote by $C$ the equivalence class with respect to
$\approx_{\rm e}$ containing $P_{1}$ and $P_{2}$, then it must be the case that $q[C \cap
\mathit{ds}(P_{1}), P_{1}, \tau] - q[P_{1}, \tau] = q[C \cap \mathit{ds}(P_{2}), P_{2}, \tau] - q[P_{2},
\tau]$ with $q[C \cap \mathit{ds}(P_{1}), P_{1}, \tau] = 0 = q[C \cap \mathit{ds}(P_{2}), P_{2}, \tau]$ due
to Lemma~\ref{lem:incoming_rate_nr}.

\item $q[C \cap \mathit{ds}(P_{1}), P_{1}, \alpha] = 0 = q[C \cap \mathit{ds}(P_{2}), P_{2}, \alpha]$ due to
Lemma~\ref{lem:incoming_rate_nr} for all $\alpha \in \cA$ and $C \in \cC / {\approx_{\rm e}}$.
			\end{itemize}

\item $\sim_{\rm s} {\subseteq} \approx_{\rm s}$ stems from Theorem~\ref{thm:taxonomy}. \\
Suppose now that $P_{1} \approx_{\rm s} P_{2}$ for two arbitrary $P_{1}, P_{2} \in \cC_{\rm nr}$.
Then:

			\begin{itemize}
\item $q[P_{1}, C, \alpha] = q[P_{2}, C, \alpha]$ for all $\alpha \in \cA$ and $C \in \cC / {\approx_{\rm
s}}$ such that either $\alpha \neq \tau$, or $\alpha = \tau$ and $P_{1}, P_{2} \notin C$.

\item $q[P_{1}, C, \tau] = 0 = q[P_{2}, C, \tau]$ for $C \in \cC / {\approx_{\rm s}}$ such that $P_{1},
P_{2} \in C$ because $P_{1}, P_{2} \in \cC_{\rm nr}$.

\item $q[C \cap \mathit{ds}(P_{1}), P_{1}, \alpha] = 0 = q[C \cap \mathit{ds}(P_{2}), P_{2}, \alpha]$ due to
Lemma~\ref{lem:incoming_rate_nr} for all $\alpha \in \cA$ and $C \in \cC / {\approx_{\rm s}}$.
			\end{itemize}

\item The incomparability of $\approx_{\rm o}$ and $\approx_{\rm e}$ stems from
Examples~\ref{ex:no_alpha_tau_p_nr} and~\ref{ex:exact_vs_ordinary_nr}.
\qedhere
		\end{itemize}
	\end{proof}

%
%
\section{Compositionality Properties}
\label{sec:congruence}
%
%

In this section we investigate the congruence property of the six bisimilarities with respect to the
operators of PEPA. This is an important property as it enables compositional reasoning, in the sense that it
$(i)$ allows equivalent subcomponents to be replaced within arbitrary components without modifying the
overall functional and performance behavior of the latter and $(ii)$ supports compositional state space
minimization without altering the original semantics. We separately examine strong bisimilarities
(Section~\ref{sec:congruence_strong}) and weak bisimilarities (Section~\ref{sec:congruence_weak}) and show
that some of them are not congruences with respect to prefix and/or choice. In that case we single out
either a set of components over which congruence with respect to those operators is achieved, or the
coarsest congruence with respect to them that is contained in the considered bisimilarity.

%
\subsection{Compositionality of Strong Bisimilarities}
\label{sec:congruence_strong}
%

The congruence property of strong ordinary bisimilarity with respect to all the operators of PEPA was proved
in~\cite{hillston:book}, where it is called strong equivalence.

	\begin{thm}[Congruence of $\sim_{\rm o}$ {\cite{hillston:book}}]
	\label{thm:congr_strong_ordinary_bisim}
Let $P_{1}, P_{2} \in \cC$. If $P_{1} \sim_{\rm o} P_{2}$ then:

		\begin{itemize}
\item $a \, . \, P_{1} \sim_{\rm o} a \, . \, P_{2}$ for all $a \in \cA \times \cR$.

\item $P_{1} + P \sim_{\rm o} P_{2} + P$ and $P + P_{1} \sim_{\rm o} P + P_{2}$ for all $P \in \cC$.

\item $P_{1} \sync{L} P \sim_{\rm o} P_{2} \sync{L} P$ and $P \sync{L} P_{1} \sim_{\rm o} P \sync{L} P_{2}$
for all $L \subseteq \cA \setminus \{ \tau \}$ and $P \in \cC$.

\item $P_{1} \, / \, L \sim_{\rm o} P_{2} \, / \, L$ for all $L \subseteq \cA \setminus \{ \tau \}$.
		\end{itemize}
	\end{thm}

Now we prove that strong exact bisimilarity and strong strict bisimilarity are congruences with respect to
all the operators of PEPA other than prefix and choice.

	\begin{thm}[Congruence of $\sim_{\rm e}$]\label{thm:congr_strong_exact_bisim}
Let $P_{1}, P_{2} \in \cC$. If $P_{1} \sim_{\rm e} P_{2}$ then:

		\begin{itemize}
\item $P_{1} \sync{L} P \sim_{\rm e} P_{2} \sync{L} P$ and $P \sync{L} P_{1} \sim_{\rm e} P \sync{L} P_{2}$
for all $L \subseteq \cA \setminus \{ \tau \}$ and $P \in \cC$.

\item $P_{1} \, / \, L \sim_{\rm e} P_{2} \, / \, L$ for all $L \subseteq \cA \setminus \{ \tau \}$.
		\end{itemize}
	\end{thm}

	\begin{proof}
Let $\cB$ be a strong exact bisimulation such that $(P_{1}, P_{2}) \in \cB$. We prove that:

		\begin{itemize}
\item $\cB' = \cI_{\cC} \cup \{ (Q_{1} \sync{L} Q, Q_{2} \sync{L} Q) \mid (Q_{1}, Q_{2}) \in \cB \}$ and
$\cB'' = \cI_{\cC} \cup \{ (Q \sync{L} Q_{1}, Q \sync{L} Q_{2}) \mid (Q_{1}, Q_{2}) \in \cB \}$ are strong
exact bisimulations too, where $\cI_{\cC}$ is the identity relation over $\cC$. Due to the symmetry of the
operational semantic rules for cooperation, without loss of generality we can focus on a single relation,
say $\cB'$. The only interesting cases are those in which we consider $(Q_{1} \sync{L} Q, Q_{2} \sync{L} Q)
\in \cB'$, hence $(Q_{1}, Q_{2}) \in \cB$, and $C \in \cC / \cB'$ of the form $[Q']_{\cB} \, \sync{L}
\bar{Q}$, i.e., $\{ Q'' \sync{L} \bar{Q} \mid (Q'', Q') \in \cB \}$. There are two subcases based on $\alpha
\in \cA$:

			\begin{itemize}
\item If $\alpha \notin L$ then $q[Q_{1} \sync{L} Q, \alpha] = q[Q_{2} \sync{L} Q, \alpha]$ and $q[C \cap
\mathit{ds}(Q_{1} \sync{L} Q), Q_{1} \sync{L} Q, \alpha] = q[C \cap \mathit{ds}(Q_{2} \sync{L} Q), Q_{2}
\sync{L} Q, \alpha]$ because for $i \in \{ 1, 2 \}$ it holds that $q[Q_{i} \sync{L} Q, \alpha] = q[Q_{i},
\alpha] + q[Q, \alpha]$ and $q[C \cap \mathit{ds}(Q_{i} \sync{L} Q), Q_{i} \sync{L} Q, \alpha]$ is equal to:

				\begin{itemize}
\item $q[[Q']_{\cB} \cap \mathit{ds}(Q_{i}), Q_{i}, \alpha] + q[\{ \bar{Q} \}, Q, \alpha]$ if $Q_{i} \in
[Q']_{\cB}$ and $Q = \bar{Q}$

\item $q[[Q']_{\cB} \cap \mathit{ds}(Q_{i}), Q_{i}, \alpha]$ if $Q_{i} \notin [Q']_{\cB}$ and $Q = \bar{Q}$

\item $q[\{ \bar{Q} \}, Q, \alpha]$ if $Q_{i} \in [Q']_{\cB}$ and $Q \neq \bar{Q}$

\item $0$ if $Q_{i} \notin [Q']_{\cB}$ and $Q \neq \bar{Q}$
				\end{itemize}

with $(Q_{1}, Q_{2}) \in \cB$.

\item If $\alpha \in L$ then $q[Q_{1} \sync{L} Q, \alpha] = q[Q_{2} \sync{L} Q, \alpha]$ and $q[C \cap
\mathit{ds}(Q_{1} \sync{L} Q), Q_{1} \sync{L} Q, \alpha] = q[C \cap \mathit{ds}(Q_{2} \sync{L} Q), Q_{2}
\sync{L} Q, \alpha]$ because for $i \in \{ 1, 2 \}$ it holds that $q[Q_{i} \sync{L} Q, \alpha] =
\min(q[Q_{i}, \alpha], q[Q, \alpha])$ and $q[C \cap \mathit{ds}(Q_{i} \sync{L} Q), Q_{i} \sync{L} Q, \alpha]
= \frac{q[[Q']_{\cB} \cap \mathit{ds}(Q_{i}), Q_{i}, \alpha]}{q[Q', \alpha]} \cdot \frac{q[\{ \bar{Q} \}, Q,
\alpha]}{q[\bar{Q}, \alpha]} \cdot \min(q[Q', \alpha], q[\bar{Q}, \alpha])$ with $(Q_{1}, Q_{2}) \in \cB$.
			\end{itemize}

\item $\cB' = \cI_{\cC} \cup \{ (Q_{1} \, / \, L, Q_{2} \, / \, L) \mid (Q_{1}, Q_{2}) \in \cB \}$ is a
strong exact bisimulation too. The only interesting cases are those in which we consider $(Q_{1} \, / \, L,
Q_{2} \, / \, L) \in \cB'$, hence $(Q_{1}, Q_{2}) \in \cB$, and $C \in \cC / \cB'$ of the form $[Q]_{\cB} \,
/ \, L$, i.e., $\{ Q' \, / \, L \mid (Q, Q') \in \cB \}$. There are three subcases based on $\alpha \in
\cA$:

			\begin{itemize}
\item If $\alpha \in L$ then $q[Q_{1} \, / \, L, \alpha] = 0 = q[Q_{2} \, / \, L, \alpha]$ and $q[C \cap
\mathit{ds}(Q_{1} \, / \, L), Q_{1} \, / \, L, \alpha] = 0 = q[C \cap \mathit{ds}(Q_{2} \, / \, L), Q_{2} \,
/ \, L, \alpha]$.

\item If $\alpha = \tau$ then $q[Q_{1} \, / \, L, \alpha] = q[Q_{2} \, / \, L, \alpha]$ and $q[C \cap
\mathit{ds}(Q_{1} \, / \, L), Q_{1} \, / \, L, \alpha] = q[C \cap \mathit{ds}(Q_{2} \, / \, L), Q_{2} \, /
\, L, \alpha]$ because $q[Q_{i} \, / \, L, \tau] = \sum_{\alpha' \in L \cup \{ \tau \}} q[Q_{i}, \alpha']$
and $q[C \cap \mathit{ds}(Q_{i}  \, / \, L), \linebreak Q_{i} \, / \, L, \tau] = \sum_{\alpha' \in L \cup \{
\tau \}} q[[Q]_{\cB} \cap \mathit{ds}(Q_{i}), Q_{i}, \alpha']$ for $i \in \{ 1, 2 \}$ with $(Q_{1}, Q_{2})
\in \cB$.

\item If $\alpha \notin L \cup \{ \tau \}$ then $q[Q_{1} \, / \, L, \alpha] = q[Q_{2} \, / \, L, \alpha]$
and $q[C \cap \mathit{ds}(Q_{1} \, / \, L), Q_{1} \, / \, L, \alpha] = q[C \cap \mathit{ds}(Q_{2} \, / \,
L), Q_{2} \, / \, L, \alpha]$ because $q[Q_{i} \, / \, L, \alpha] = q[Q_{i}, \alpha]$ and $q[Q_{i} / L, C,
\alpha] = q[[Q]_{\cB} \cap \mathit{ds}(Q_{i}), Q_{i}, \alpha]$ for $i \in \{ 1, 2 \}$ with $(Q_{1}, Q_{2})
\in \cB$.
\qedhere
			\end{itemize}
		\end{itemize}
	\end{proof}

	\begin{thm}[Congruence of $\sim_{\rm s}$]\label{thm:congr_strong_strict_bisim}
Let $P_{1}, P_{2} \in \cC$. If $P_{1} \sim_{\rm s} P_{2}$ then:

		\begin{itemize}
\item $P_{1} \sync{L} P \sim_{\rm s} P_{2} \sync{L} P$ and $P \sync{L} P_{1} \sim_{\rm s} P \sync{L} P_{2}$
for all $L \subseteq \cA \setminus \{ \tau \}$ and $P \in \cC$.

\item $P_{1} \, / \, L \sim_{\rm s} P_{2} \, / \, L$ for all $L \subseteq \cA \setminus \{ \tau \}$.
		\end{itemize}
	\end{thm}

	\begin{proof}
Let $\cB$ be a strong strict bisimulation such that $(P_{1}, P_{2}) \in \cB$. We prove that:

		\begin{itemize}
\item $\cB' = \cI_{\cC} \cup \{ (Q_{1} \sync{L} Q, Q_{2} \sync{L} Q) \mid (Q_{1}, Q_{2}) \in \cB \}$ and
$\cB'' = \cI_{\cC} \cup \{ (Q \sync{L} Q_{1}, Q \sync{L} Q_{2}) \mid (Q_{1}, Q_{2}) \in \cB \}$ are strong
strict bisimulations too, where $\cI_{\cC}$ is the identity relation over $\cC$. Due to the symmetry of the
operational semantic rules for cooperation, without loss of generality we can focus on a single relation,
say $\cB'$. The only interesting cases are those in which we consider $(Q_{1} \sync{L} Q, Q_{2} \sync{L} Q)
\in \cB'$, hence $(Q_{1}, Q_{2}) \in \cB$, and $C \in \cC / \cB'$ of the form $[Q']_{\cB} \, \sync{L}
\bar{Q}$, i.e., $\{ Q'' \sync{L} \bar{Q} \mid (Q'', Q') \in \cB \}$.

\noindent
As for outgoing conditions, there are two subcases based on $\alpha \in \cA$:

			\begin{itemize}
\item If $\alpha \notin L$ then $q[Q_{1} \sync{L} Q, C, \alpha] = q[Q_{2} \sync{L} Q, C, \alpha]$ because
for $i \in \{ 1, 2 \}$ it holds that $q[Q_{i} \sync{L} Q, C, \alpha]$ is equal to:

				\begin{itemize}
\item $q[Q_{i}, [Q']_{\cB}, \alpha] + q[Q, \{ \bar{Q} \}, \alpha]$ if $Q_{i} \in [Q']_{\cB}$ and $Q =
\bar{Q}$

\item $q[Q_{i}, [Q']_{\cB}, \alpha]$ if $Q_{i} \notin [Q']_{\cB}$ and $Q = \bar{Q}$

\item $q[Q, \{ \bar{Q} \}, \alpha]$ if $Q_{i} \in [Q']_{\cB}$ and $Q \neq \bar{Q}$

\item $0$ if $Q_{i} \notin [Q']_{\cB}$ and $Q \neq \bar{Q}$
				\end{itemize}

with $(Q_{1}, Q_{2}) \in \cB$.

\item If $\alpha \in L$ then $q[Q_{1} \sync{L} Q, C, \alpha] = q[Q_{2} \sync{L} Q, C, \alpha]$ because for
$i \in \{ 1, 2 \}$ it holds that $q[Q_{i} \sync{L} Q, C, \alpha] = \frac{q[Q_{i}, [Q']_{\cB},
\alpha]}{q[Q_{i}, \alpha]} \cdot \frac{q[Q, \{ \bar{Q} \}, \alpha]}{q[Q, \alpha]} \cdot \min(q[Q_{i},
\alpha], q[Q, \alpha])$ with $(Q_{1}, Q_{2}) \in \cB$.
			\end{itemize}

\noindent
As for incoming conditions, we proceed like in the proof of the corresponding result of
Theorem~\ref{thm:congr_strong_exact_bisim}.

\item $\cB' = \cI_{\cC} \cup \{ (Q_{1} \, / \, L, Q_{2} \, / \, L) \mid (Q_{1}, Q_{2}) \in \cB \}$ is a
strong strict bisimulation too. The only interesting cases are those in which we consider $(Q_{1} \, / \, L,
Q_{2} \, / \, L) \in \cB'$, hence $(Q_{1}, Q_{2}) \in \cB$, and $C \in \cC / \cB'$ of the form $[Q]_{\cB} \,
/ \, L$, i.e., $\{ Q' \, / \, L \mid (Q, Q') \in \cB \}$.

\noindent
As for outgoing conditions, there are three subcases based on $\alpha \in \cA$:

			\begin{itemize}
\item If $\alpha \in L$ then $q[Q_{1} \, / \, L, C, \alpha] = 0 = q[Q_{2} \, / \, L, C, \alpha]$.

\item If $\alpha = \tau$ then $q[Q_{1} \, / \, L, C, \alpha] = q[Q_{2} \, / \, L, C, \alpha]$ because
$q[Q_{i} \, / \, L, C, \tau] = \sum_{\alpha' \in L \cup \{ \tau \}} q[Q_{i},$ \linebreak $[Q]_{\cB},
\alpha']$ for $i \in \{ 1, 2 \}$ with $(Q_{1}, Q_{2}) \in \cB$.

\item If $\alpha \notin L \cup \{ \tau \}$ then $q[Q_{1} / L, C, \alpha] = q[Q_{2} / L, C, \alpha]$ because
$q[Q_{i} / L, C, \alpha] = q[Q_{i}, [Q]_{\cB}, \alpha]$ for $i \in \{ 1, 2 \}$ with $(Q_{1}, Q_{2}) \in
\cB$.
			\end{itemize}

\noindent
As for incoming conditions, we proceed like in the proof of the corresponding result of
Theorem~\ref{thm:congr_strong_exact_bisim}.
\qedhere
		\end{itemize}
	\end{proof}

The following counterexamples show why strong exact bisimilarity and strong strict bisimilarity are not
congruences with respect to prefix and choice. To achieve compositionality for those two operators, we need
to restrict to components with no recursion.

	\begin{exa}[$\sim_{\rm s}$ w.r.t.\ prefix and choice]\label{ex:no_congr_prefix_choice_strict_bisim}
Consider the two recursive components:
\[
A_{11} \: \rmdef \: (\alpha, r) \, . \, A_{11}
\qquad
A_{12} \: \rmdef \: (\alpha, r) \, . \, (\alpha, r) \, . \, A_{12}
\]
with $\alpha \neq \tau$ and $r \in \mathbb{R}_{> 0}$, whose underlying derivation graphs are depicted in
Figure~\ref{fig:no_congr_prefix_choice_strict_bisim}.


	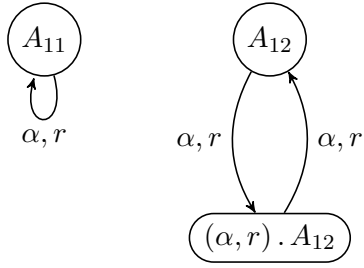
\begin{figure}[t]

\begin{center}
\begin{tikzpicture}[modal, node distance = 1.8cm]


\node[state]          (a11)		          {\makebox[-4pt][c]{$A_{11}$}};
\node[state]          (a12)  [right = 2cm of a11] {\makebox[-4pt][c]{$A_{12}$}};
\node[elliptic state] (a122) [below = of a12]     {$(\alpha, r) \, . \, A_{12}$};


\path[->] (a11)	 edge[loop below] node[below] {$\alpha, r$} (a11);
\path[->] (a12)	 edge[bend right] node[left]  {$\alpha, r$} (a122);
\path[->] (a122) edge[bend right] node[right] {$\alpha, r$} (a12);

\end{tikzpicture}
\end{center}

\caption{Derivation graphs for the original components in
Example~\ref{ex:no_congr_prefix_choice_strict_bisim}.}
\label{fig:no_congr_prefix_choice_strict_bisim}

	\end{figure}

	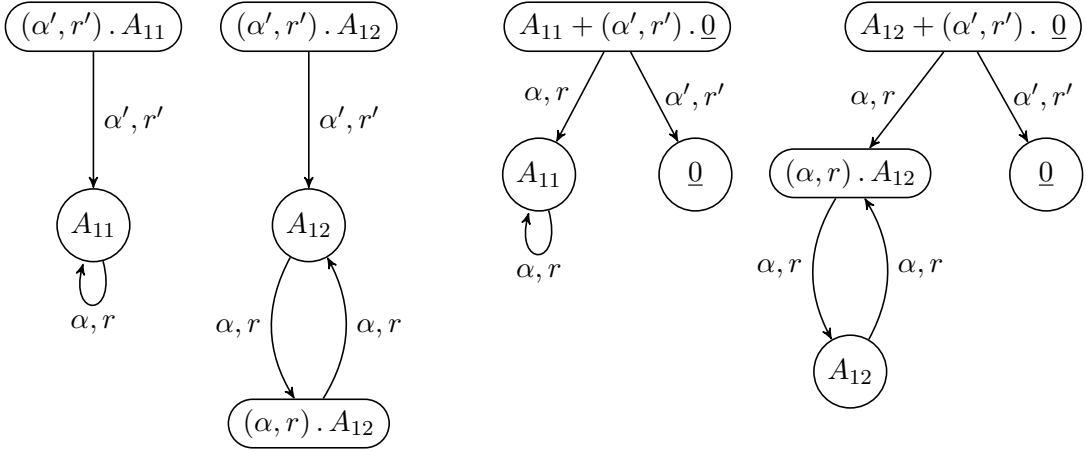
\begin{figure}[t]

\begin{center}
\begin{tikzpicture}[modal, node distance = 1.8cm]


\node[elliptic state] (alpha11p)			     {$(\alpha', r') \, . \, A_{11}$};
\node[state]          (a11p)     [below = of alpha11p]       {\makebox[-4pt][c]{$A_{11}$}};
\node[elliptic state] (alpha12p) [right = 0.5cm of alpha11p] {$(\alpha', r') \, . \, A_{12}$};
\node[state]          (a12p)     [below = of alpha12p]       {\makebox[-4pt][c]{$A_{12}$}};
\node[elliptic state] (a122p)    [below = of a12p]           {$(\alpha, r) \, . \, A_{12}$};
\node[elliptic state] (alpha11c) [right = 1.4cm of alpha12p] {$A_{11} + (\alpha', r') \, . \,
							      \underline{0}$};
\node[state]	      (a11c)	 [below left = of alpha11c, xshift = 50pt]
							     {\makebox[-4pt][c]{$A_{11}$}};
\node[state]	      (nilc)	 [below right = of alpha11c, xshift = -50pt]
							     {\makebox[-4pt][c]{$\underline{0}$}};
\node[elliptic state] (alpha12c) [right = 1.5cm of alpha11c]
							    {$A_{12} + (\alpha', r') \, . \, \underline{0}$};
\node[elliptic state] (a12c)	 [below left = of alpha12c, xshift = 50pt]
							     {$(\alpha, r) \, . \, A_{12}$};
\node[state]	      (nil2c)	 [below right = of alpha12c, xshift = -50pt]
							     {\makebox[-4pt][c]{$\underline{0}$}};
\node[state]	      (a122c)	 [below = of a12c]	     {\makebox[-4pt][c]{$A_{12}$}};


\path[->] (alpha11p) edge[]           node[right] {$\alpha', r'$} (a11p);
\path[->] (a11p)     edge[loop below] node[below] {$\alpha, r$}   (a11p);
\path[->] (alpha12p) edge[]           node[right] {$\alpha', r'$} (a12p);
\path[->] (a12p)     edge[bend right] node[left]  {$\alpha, r$}   (a122p);
\path[->] (a122p)    edge[bend right] node[right] {$\alpha, r$}   (a12p);
\path[->] (alpha11c) edge[]           node[left]  {$\alpha, r$}   (a11c);
\path[->] (alpha11c) edge[]           node[right] {$\alpha', r'$} (nilc);
\path[->] (a11c)     edge[loop below] node[below] {$\alpha, r$}   (a11c);
\path[->] (alpha12c) edge[]           node[left]  {$\alpha, r$}   (a12c);
\path[->] (alpha12c) edge[]           node[right] {$\alpha', r'$} (nil2c);
\path[->] (a12c)     edge[bend right] node[left]  {$\alpha, r$}   (a122c);
\path[->] (a122c)    edge[bend right] node[right] {$\alpha, r$}   (a12c);

\end{tikzpicture}
\end{center}

\caption{Derivation graphs for the additional components in
Example~\ref{ex:no_congr_prefix_choice_strict_bisim}.}
\label{fig:no_congr_prefix_choice_strict_bisim_additional}

	\end{figure}

It turns out that $A_{11} \sim_{\rm s} A_{12}$ as witnessed by the equivalence relation whose only
non-singleton equivalence class is $C = \{ A_{11}, A_{12}, (\alpha, r) \, . \, A_{12} \}$. Indeed, every
component in~$C$ reaches $C$ via $\alpha$ with cumulative rate $r$ and is reached from $C$ via $\alpha$ with
cumulative rate $r$.

However, $(\alpha', r') \, . \, A_{11} \not\sim_{\rm s} (\alpha', r') \, . \, A_{12}$, where $\alpha' \neq
\tau$ and $r' \in \mathbb{R}_{> 0}$ with $\alpha' \neq \alpha$ or $r \neq r'$. As can be seen in
Figure~\ref{fig:no_congr_prefix_choice_strict_bisim_additional}, the reason is that $(\alpha, r) \, . \,
A_{12}$ still has a single incoming $\alpha$-transition (from $A_{12}$) while $A_{12}$ now has both an
incoming $\alpha$-transition (from $(\alpha, r) \, . \, A_{12}$) and an incoming $\alpha'$-transition (from
$(\alpha', r') \, . \, A_{12}$), hence they can no longer belong to the same equivalence class. As a
consequence, $A_{11}$ (which reaches only itself) and $A_{12}$ (which reaches also $(\alpha, r) \, . \,
A_{12}$) can no longer belong to the same equivalence class either, hence $(\alpha', r') \, . \, A_{11}$ and
$(\alpha', r') \, . \, A_{12}$ are told apart.

Likewise, $A_{11} + (\alpha', r') \, . \, \underline{0} \not\sim_{\rm s} A_{12} + (\alpha', r') \, . \,
\underline{0}$. As can be seen in Figure~\ref{fig:no_congr_prefix_choice_strict_bisim_additional}, the
reason is that $A_{12}$ still has a single incoming $\alpha$-transition (from $(\alpha, r) \, . \, A_{12}$)
while $(\alpha, r) \, . \, A_{12}$ now has two incoming $\alpha$-transitions (one from $A_{12} + (\alpha',
r') \, . \, \underline{0}$ and the other from $A_{12}$), hence they can no longer belong to the same
equivalence class. As a consequence, $A_{11}$ (which reaches only itself) and $(\alpha, r) \, . \, A_{12}$
(which reaches also $A_{12}$) can no longer belong to the same equivalence class either, hence $A_{11} +
(\alpha', r') \, . \, \underline{0}$ and $A_{12} + (\alpha', r') \, . \, \underline{0}$ are told apart.

In contrast, $A_{11} \sim_{\rm e} A_{12}$ with $(\alpha', r') \, . \, A_{11} \sim_{\rm e} (\alpha', r') \, .
\, A_{12}$ and $A_{11} + (\alpha', r') \, . \, \underline{0} \sim_{\rm e} A_{12} + (\alpha', r') \, . \,
\underline{0}$, because none of the four additional components has incoming transitions and the first
(resp.\ last) two have the same total exit rate on any action type.
	\end{exa}

	\begin{exa}[$\sim_{\rm e}$ w.r.t.\ prefix]\label{ex:no_congr_prefix_exact_bisim}
Given the same two recursive components $A_{11}$ and $A_{12}$ as the previous example, which are related by
$\sim_{\rm e}$, consider $(\alpha, r) \, . \, A_{11}$ and $(\alpha, r) \, . \, A_{12}$, whose underlying
derivation graphs are depicted in Figure~\ref{fig:no_congr_prefix_exact_bisim}.


	\begin{figure}[t]

\begin{center}
\begin{tikzpicture}[modal, node distance = 1.8cm]


\node[elliptic state] (alpha11) []		     {$(\alpha, r) \, . \, A_{11}$};
\node[state]	      (a11)	[below = of alpha11] {\makebox[-4pt][c]{$A_{11}$}};
\node[elliptic state] (alpha12) [right = of alpha11] {$(\alpha, r) \, . \, A_{12}$};
\node[state]	      (a12)	[below = of alpha12] {\makebox[-4pt][c]{$A_{12}$}};


\path[->] (alpha11) edge[]	     node[right] {$\alpha, r$} (a11);
\path[->] (a11)	    edge[loop below] node[below] {$\alpha, r$} (a11);
\path[->] (a12)	    edge[bend right] node[right] {$\alpha, r$} (alpha12);
\path[->] (alpha12) edge[bend right] node[left]  {$\alpha, r$} (a12);

\end{tikzpicture}
\end{center}

\caption{Derivation graphs for the additional components in Example~\ref{ex:no_congr_prefix_exact_bisim}.}
\label{fig:no_congr_prefix_exact_bisim}

	\end{figure}
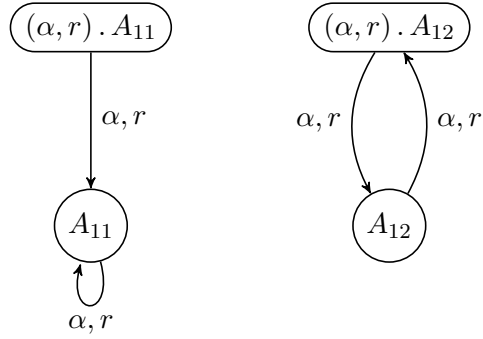

It holds that $(\alpha, r) \, . \, A_{11} \not\sim_{\rm e} (\alpha, r) \, . \, A_{12}$ because the former
has no incoming transitions, whereas the latter has one incoming $\alpha$-transition from $A_{12}$. Note
that $(\alpha, r) \, . \, A_{11}$ and~$A_{11}$ are syntactically different, which is the reason why the
operational semantic rules associate two different states with them. The defining equation for $A_{11}$,
which establishes that $A_{11}$ can perform $\alpha$ at rate $r$ and then repeats itself, does not create
any syntactical identification between $A_{11}$ and $(\alpha, r) \, . \, A_{11}$ as it is only used in the
operational semantic rules. This would be more evident if in the syntax we had variables and the rec binder
for them instead of constants and defining equations for them. Recalling that in process algebra the use of
the former is equivalent to the use of the latter~\cite{Mil89a}, $A_{11}$ would be modeled as $\textrm{rec}
\, X : (\alpha, r) \, .  \, X$ whilst $(\alpha, r) \, . \, A_{11}$ would be modeled as $(\alpha, r) \, . \,
\textrm{rec} \, X : (\alpha, r) \, . \, X$, which are syntactically different and not related by any
defining equation.

	\end{exa}

	\begin{exa}[$\sim_{\rm e}$ w.r.t.\ choice]\label{ex:no_congr_choice_exact_bisim}
Consider the recursive component:
\[
A \: \rmdef \: (\alpha, r) \, . \, ((\alpha, r) \, . \, A + (\alpha, r) \, . \, A)
\]
with $\alpha \neq \tau$ and $r \in \mathbb{R}_{> 0}$. Its derivation graph has one $\alpha$-transition to
the state corresponding to $(\alpha, r) \, . \, A + (\alpha, r) \, . \, A$, which in turn has two
$\alpha$-transitions to the state corresponding to $A$.

It holds that $(\alpha, r) \, . \, A \sim_{\rm e} (\alpha, r) \, . \, \underline{0}$ because both have the
same total exit rate on any action type and neither has incoming transitions. As can be seen in
Figure~\ref{fig:no_congr_choice_exact_bisim}, the derivation graph of the former has one $\alpha$-transition
to the state corresponding to $A$, which cannot reach the state corresponding to $(\alpha, r) \, . \, A$.

	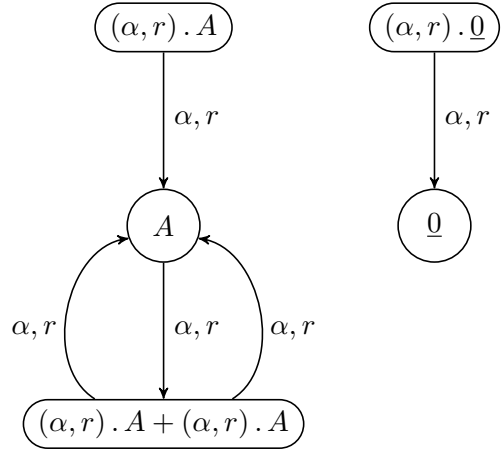
\begin{figure}[t]

\begin{center}
\begin{tikzpicture}[modal, node distance = 1.8cm]


\node[elliptic state] (alphaA)   []                    {$(\alpha, r) \, . \, A$};
\node[state]          (A)        [below = of alphaA]   {\makebox[-4pt][c]{$A$}};
\node[elliptic state] (A+A)      [below = of A]        {$(\alpha, r) \, . \, A + (\alpha, r) \, . \, A$};
\node[elliptic state] (alphanil) [right = of alphaA]   {$(\alpha, r) \, . \, \underline{0}$};
\node[state]          (nil)      [below = of alphanil] {\makebox[-4pt][c]{$\underline{0}$}}; 


\path[->] (alphaA)   edge[]                node[right]                {$\alpha,r$} (A);
\path[->] (A)        edge[]                node[right]                {$\alpha,r$} (A+A);
\path[->] (A+A)      edge[bend right = 70] node[right, yshift = -8pt] {$\alpha,r$} (A);
\path[->] (A+A)      edge[bend left = 70]  node[left, yshift = -8pt]  {$\alpha,r$} (A);
\path[->] (alphanil) edge[]                node[right]                {$\alpha,r$} (nil);

\end{tikzpicture}
\end{center}

\caption{Derivation graphs for the original components in Example~\ref{ex:no_congr_choice_exact_bisim}.}
\label{fig:no_congr_choice_exact_bisim}

	\end{figure}
 
	\begin{figure}[t]

\begin{center}
\begin{tikzpicture}[modal, node distance = 1.8cm]


\node[elliptic state] (A+A)		       {$(\alpha, r) \, . \, A + (\alpha, r) \, . \, A$};
\node[state]          (A)    [below = of A+A]			     {\makebox[-4pt][c]{$A$}};
\node[elliptic state] (root) [right = of A+A]
					       {$(\alpha, r) \, . \, \underline{0} + (\alpha, r) \, . \, A$};
\node[state]          (nil)  [below left = of root, xshift = 50pt]   {\makebox[-4pt][c]{$\underline{0}$}};
\node[state]          (r1)   [below right = of root, xshift = -50pt] {\makebox[-4pt][c]{$A$}}; 
\node[elliptic state] (r2)   [below = of r1]
					       {$(\alpha, r) \, . \, A + (\alpha, r) \, . \, A$};


\path[->] (A)	 edge[]		       node[right]		  {$\alpha, r$} (A+A);
\path[->] (A+A)	 edge[bend right = 70] node[left, yshift = 8pt]	  {$\alpha, r$} (A);
\path[->] (A+A)	 edge[bend left = 70]  node[right, yshift = 8pt]  {$\alpha, r$} (A);
\path[->] (root) edge[]                node[left]		  {$\alpha, r$} (nil);
\path[->] (root) edge[]                node[right]		  {$\alpha, r$} (r1);
\path[->] (r1)   edge[]                node[right]		  {$\alpha, r$} (r2);
\path[->] (r2)   edge[bend right = 70] node[right, yshift = -8pt] {$\alpha, r$} (r1);
\path[->] (r2)   edge[bend left = 70]  node[left, yshift = -8pt]  {$\alpha, r$} (r1);

\end{tikzpicture}
\end{center}

\caption{Derivation graphs for the additional components in Example~\ref{ex:no_congr_choice_exact_bisim}.}
\label{fig:no_congr_choice_exact_bisim_additional}

	\end{figure}
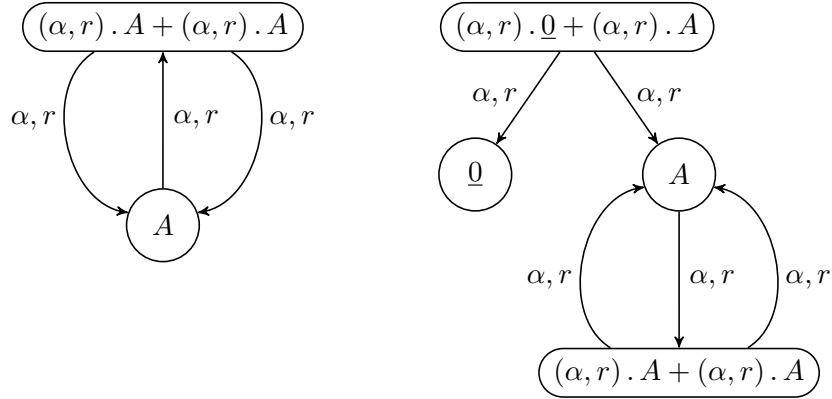

However, $(\alpha, r) \, . \, A + (\alpha, r) \, . \, A \not\sim_{\rm e} (\alpha, r) \, . \, \underline{0} +
(\alpha, r) \, . \, A$. As can be seen in Figure~\ref{fig:no_congr_choice_exact_bisim_additional}, the
reason is that the latter has no incoming transitions, while the former has one incoming transition from the
state corresponding to $A$.
	\end{exa}

	\begin{thm}[Congruence of $\sim_{\rm e}$ over $\cC_{\rm nr}$]
	\label{thm:congr_strong_exact_bisim_prefix_choice}
Let $P_{1}, P_{2} \in \cC_{\rm nr}$. If $P_{1} \sim_{\rm e} P_{2}$ then:

		\begin{itemize}
\item $a \, . \, P_{1} \sim_{\rm e} a \, . \, P_{2}$ for all $a \in \cA \times \cR$.

\item $P_{1} + P \sim_{\rm e} P_{2} + P$ and $P + P_{1} \sim_{\rm e} P + P_{2}$ for all $P \in \cC_{\rm
nr}$.
		\end{itemize}
	\end{thm}

	\begin{proof}
Let $\cB$ be a strong exact bisimulation over $\cC_{\rm nr}$ such that $(P_{1}, P_{2}) \in \cB$. We prove
that:

		\begin{itemize}
\item $\cB' = (\cB \cup \{ ((\alpha, r) \, . \, Q_{1}, (\alpha, r) \, . \, Q_{2}) \mid (Q_{1}, Q_{2}) \in
\cB \})^{+}$ is a strong exact bisimulation over $\cC_{\rm nr}$ too, where $^{+}$ denotes transitive
closure. There are only two interesting cases:

			\begin{itemize}
\item If we consider $(Q'_{1}, Q'_{2}) \in \cB' \cap \cB$ -- which thus already satisfies $q[Q'_{1}, \alpha]
= q[Q'_{2}, \alpha]$ for all $\alpha \in \cA$ -- and $C = [(\alpha, r) \, . \, Q_{1}]_{\cB} \cup [(\alpha,
r) \, . \, Q_{2}]_{\cB}$ for some $(Q_{1}, Q_{2}) \in \cB$, then $q[C \cap \mathit{ds}(Q'_{1}), Q'_{1},
\alpha'] = q[C \cap \mathit{ds}(Q'_{2}), Q'_{2}, \alpha']$ for all $\alpha' \in \cA$ because $q[C \cap
\mathit{ds}(Q'_{i}), Q'_{i}, \alpha'] = q[[(\alpha, r) \, . \, Q_{1}]_{\cB} \cap \mathit{ds}(Q'_{i}),
Q'_{i}, \alpha'] + q[[(\alpha, r) \, . \, Q_{2}]_{\cB} \cap \mathit{ds}(Q'_{i}), Q'_{i}, \alpha']$ for $i
\in \{ 1, 2 \}$ with $(Q'_{1}, Q'_{2}) \in \cB$.

\item If we consider $((\alpha, r) \, . \, Q_{1}, (\alpha, r) \, . \, Q_{2}) \in \cB'$, hence $(Q_{1},
Q_{2}) \in \cB$, then the result follows from:

				\begin{itemize}
\item $q[(\alpha, r) \, . \, Q_{1}, \alpha] = r = q[(\alpha, r) \, . \, Q_{2}, \alpha]$.

\item $q[(\alpha, r) \, . \, Q_{1}, \alpha'] = 0 = q[(\alpha, r) \, . \, Q_{2}, \alpha']$ for all $\alpha'
\in \cA \setminus \{ \alpha \}$.

\item $q[C \cap \mathit{ds}((\alpha, r) \, . \, Q_{1}), (\alpha, r) \, . \, Q_{1}, \alpha'] = 0 = q[C \cap
\mathit{ds}((\alpha, r) \, . \, Q_{2}), (\alpha, r) \, . \, Q_{2}, \alpha']$ for all $\alpha' \in \cA$ and
$C \in \cC / \cB'$ because we are considering components with no recursion.
				\end{itemize}
			\end{itemize}

\item $\cB' = (\cB \cup \{ (Q_{1} + Q, Q_{2} + Q) \mid (Q_{1}, Q_{2}) \in \cB \})^{+}$ and $\cB'' = (\cB
\cup \{ (Q + Q_{1}, Q + Q_{2}) \mid (Q_{1}, Q_{2}) \in \cB \})^{+}$ are strong exact bisimulations over
$\cC_{\rm nr}$ too. Due to the symmetry of the operational semantic rules for choice, without loss of
generality we can focus on a single relation, say $\cB'$. There are only two interesting cases:

			\begin{itemize}
\item If we consider $(Q'_{1}, Q'_{2}) \in \cB' \cap \cB$ -- which thus already satisfies $q[Q'_{1}, \alpha]
= q[Q'_{2}, \alpha]$ for all $\alpha \in \cA$ -- and $C = [Q_{1} + Q]_{\cB} \cup [Q_{2} + Q]_{\cB}$ for some
$(Q_{1}, Q_{2}) \in \cB$, then $q[C \cap \mathit{ds}(Q'_{1}), Q'_{1}, \alpha] = q[C \cap
\mathit{ds}(Q'_{2}), Q'_{2}, \alpha]$ for all $\alpha \in \cA$ because $q[C \cap \mathit{ds}(Q'_{i}),
Q'_{i}, \alpha] = q[[Q_{1} + Q]_{\cB} \cap \mathit{ds}(Q'_{i}), Q'_{i}, \alpha] + q[[Q_{2} + Q]_{\cB} \cap
\mathit{ds}(Q'_{i}), Q'_{i}, \alpha]$ for $i \in \{ 1, 2 \}$ with $(Q'_{1}, Q'_{2}) \in \cB$.

\item If we consider $(Q_{1} + Q, Q_{2} + Q) \in \cB'$, hence $(Q_{1}, Q_{2}) \in \cB$, then the result
follows from:

				\begin{itemize}
\item $q[Q_{1} + Q, \alpha] = q[Q_{2} + Q, \alpha]$ for all $\alpha \in \cA$ because $q[Q_{i} + Q, \alpha] =
q[Q_{i}, \alpha] + q[Q, \alpha]$ for $i \in \{ 1, 2 \}$ with $(Q_{1}, Q_{2}) \in \cB$.

\item $q[C \cap \mathit{ds}(Q_{1} + Q), Q_{1} + Q, \alpha] = 0 = q[C \cap \mathit{ds}(Q_{2} + Q), Q_{2} + Q,
\alpha]$ for all $\alpha \in \cA$ and $C \in \cC / \cB'$ because we are considering components with no
recursion.
\qedhere
				\end{itemize}
			\end{itemize}
		\end{itemize}
	\end{proof}

	\begin{thm}[Congruence of $\sim_{\rm s}$ over $\cC_{\rm nr}$]
	\label{thm:congr_strong_strict_bisim_prefix_choice}
Let $P_{1}, P_{2} \in \cC_{\rm nr}$. If $P_{1} \sim_{\rm s} P_{2}$ then:

		\begin{itemize}
\item $a \, . \, P_{1} \sim_{\rm s} a \, . \, P_{2}$ for all $a \in \cA \times \cR$.

\item $P_{1} + P \sim_{\rm s} P_{2} + P$ and $P + P_{1} \sim_{\rm s} P + P_{2}$ for all $P \in \cC_{\rm
nr}$.
		\end{itemize}
	\end{thm}

	\begin{proof}
Let $\cB$ be a strong strict bisimulation over $\cC_{\rm nr}$ such that $(P_{1}, P_{2}) \in \cB$. We prove
that:

		\begin{itemize}
\item $\cB' = (\cB \cup \{ ((\alpha, r) \, . \, Q_{1}, (\alpha, r) \, . \, Q_{2}) \mid (Q_{1}, Q_{2}) \in
\cB \})^{+}$ is a strong strict bisimulation over $\cC_{\rm nr}$ too, where $^{+}$ denotes transitive
closure.

\noindent
As for outgoing conditions, there are only two interesting cases:

			\begin{itemize}
\item If we consider $(Q'_{1}, Q'_{2}) \in \cB' \cap \cB$ and $C = [(\alpha, r) \, . \, Q_{1}]_{\cB} \cup
[(\alpha, r) \, . \, Q_{2}]_{\cB}$ for some $(Q_{1}, Q_{2}) \in \cB$, then $q[Q'_{1}, C, \alpha'] =
q[Q'_{2}, C, \alpha']$ for all $\alpha' \in \cA$ because $q[Q'_{i}, C, \alpha'] = q[Q'_{i}, [(\alpha, r) \,
. \, Q_{1}]_{\cB}, \alpha'] + q[Q'_{i}, [(\alpha, r) \, . \, Q_{2}]_{\cB}, \alpha']$ for $i \in \{ 1, 2 \}$
with $(Q'_{1}, Q'_{2}) \in \cB$.

\item If we consider $((\alpha, r) \, . \, Q_{1}, (\alpha, r) \, . \, Q_{2}) \in \cB'$, hence $(Q_{1},
Q_{2}) \in \cB'$, then $q[(\alpha, r) \, . \, Q_{1}, C, \alpha'] \linebreak = q[(\alpha, r) \, . \, Q_{2},
C, \alpha']$ for all $\alpha' \in \cA$ and $C \in \cC / \cB'$ because:

				\begin{itemize}
\item If $\alpha' = \alpha$ and $C$ contains $Q_{1}$ and $Q_{2}$, then $q[(\alpha, r) \, . \, Q_{i}, C,
\alpha'] = r$ for $i \in \{ 1, 2 \}$.

\item If $\alpha' \neq \alpha$ or $C$ does not contain $Q_{1}$ and $Q_{2}$, then $q[(\alpha, r) \, . \,
Q_{i}, C, \alpha'] = 0$ for $i \in \{ 1, 2 \}$.
				\end{itemize}
			\end{itemize}

\noindent
As for incoming conditions, we proceed like in the proof of the corresponding result of
Theorem~\ref{thm:congr_strong_exact_bisim_prefix_choice}.

\item $\cB' = (\cB \cup \{ (Q_{1} + Q, Q_{2} + Q) \mid (Q_{1}, Q_{2}) \in \cB \})^{+}$ and $\cB'' = (\cB
\cup \{ (Q + Q_{1}, Q + Q_{2}) \mid (Q_{1}, Q_{2}) \in \cB \})^{+}$ are strong strict bisimulations over
$\cC_{\rm nr}$ too. Due to the symmetry of the operational semantic rules for choice, without loss of
generality we can focus on a single relation, say $\cB'$.

\noindent
As for outgoing conditions, there are only two interesting cases:

			\begin{itemize}
\item If we consider $(Q'_{1}, Q'_{2}) \in \cB' \cap \cB$ and $C = [Q_{1} + Q]_{\cB} \cup [Q_{2} + Q]_{\cB}$
for some $(Q_{1}, Q_{2}) \in \cB$, then $q[Q'_{1}, C, \alpha] = q[Q'_{2}, C, \alpha]$ for all $\alpha \in
\cA$ because $q[Q'_{i}, C, \alpha] = q[Q'_{i}, [Q_{1} + Q]_{\cB}, \alpha] + q[Q'_{i}, [Q_{2} + Q]_{\cB},
\alpha]$ for $i \in \{ 1, 2 \}$ with $(Q'_{1}, Q'_{2}) \in \cB$.

\item If we consider $(Q_{1} + Q, Q_{2} + Q) \in \cB'$, hence $(Q_{1}, Q_{2}) \in \cB$, then $q[Q_{1} + Q,
C, \alpha] = q[Q_{2} + Q, C, \alpha]$ for all $\alpha \in \cA$ and $C \in \cC / \cB'$ because $q[Q_{i} + Q,
C, \alpha] = q[Q_{i}, C, \alpha] + q[Q, C, \alpha]$ for $i \in \{ 1, 2 \}$ with $(Q_{1}, Q_{2}) \in \cB$.
			\end{itemize}

\noindent
As for incoming conditions, we proceed like in the proof of
Theorem~\ref{thm:congr_strong_exact_bisim_prefix_choice}.
\qedhere
		\end{itemize}
	\end{proof}

%
\subsection{Compositionality of Weak Bisimilarities}
\label{sec:congruence_weak}
%

We start by proving that weak ordinary bisimilarity is a congruence with respect to all the operators of
PEPA except for choice.

	\begin{thm}[Congruence of $\approx_{\rm o}$]\label{thm:congr_weak_ordinary_bisim}
Let $P_{1}, P_{2} \in \cC$. If $P_{1} \approx_{\rm o} P_{2}$ then:

		\begin{itemize}
\item $a \, . \, P_{1} \approx_{\rm o} a \, . \, P_{2}$ for all $a \in \cA \times \cR$.

\item $P_{1} \sync{L} P \approx_{\rm o} P_{2} \sync{L} P$ and $P \sync{L} P_{1} \approx_{\rm o} P \sync{L}
P_{2}$ for all $L \subseteq \cA \setminus \{ \tau \}$ and $P \in \cC$.

\item $P_{1} \, / \, L \approx_{\rm o} P_{2} \, / \, L$ for all $L \subseteq \cA \setminus \{ \tau \}$.
		\end{itemize}
	\end{thm}

	\begin{proof}
Let $\cB$ be a weak ordinary bisimulation such that $(P_{1}, P_{2}) \in \cB$. We prove that:

		\begin{itemize}
\item $\cB' = (\cB \cup \{ ((\alpha, r) \, . \, Q_{1}, (\alpha, r) \, . \, Q_{2}) \mid (Q_{1}, Q_{2}) \in
\cB \})^{+}$ is a weak ordinary bisimulation too, where $^{+}$ denotes transitive closure. There are only
two interesting cases:

			\begin{itemize}
\item If we consider $(Q'_{1}, Q'_{2}) \in \cB' \cap \cB$ and $C = [(\alpha, r) \, . \, Q_{1}]_{\cB} \cup
[(\alpha, r) \, . \, Q_{2}]_{\cB}$ for some $(Q_{1}, Q_{2}) \in \cB$, then $q[Q'_{1}, C, \alpha'] =
q[Q'_{2}, C, \alpha']$ for all $\alpha' \in \cA$ such that either $\alpha' \neq \tau$, or $\alpha' = \tau$
and $Q'_{1}, Q'_{2} \notin C$, because $q[Q'_{i}, C, \alpha'] = q[Q'_{i}, [(\alpha, r) \, . \, Q_{1}]_{\cB},
\alpha'] + q[Q'_{i}, [(\alpha, r) \, . \, Q_{2}]_{\cB}, \alpha']$ for $i \in \{ 1, 2 \}$ with $(Q'_{1},
Q'_{2}) \in \cB$. Note that $Q'_{i} \notin C$ implies $Q'_{i} \notin [(\alpha, r) \, . \, Q_{1}]_{\cB}$ and
$Q'_{i} \notin [(\alpha, r) \, . \, Q_{2}]_{\cB}$.

\item If we consider $((\alpha, r) \, . \, Q_{1}, (\alpha, r) \, . \, Q_{2}) \in \cB'$, hence $(Q_{1},
Q_{2}) \in \cB'$, then $q[(\alpha, r) \, . \, Q_{1}, C, \alpha'] \linebreak = q[(\alpha, r) \, . \, Q_{2},
C, \alpha']$ for all $\alpha' \in \cA$ and $C \in \cC / \cB'$ such that either $\alpha' \neq \tau$, or
$\alpha' = \tau$ and $(\alpha, r) \, . \, Q_{1}, (\alpha, r) \, . \, Q_{2} \notin C$, because:

				\begin{itemize}
\item If $\alpha' = \alpha$ and $C$ contains $Q_{1}$ and $Q_{2}$, then $q[(\alpha, r) \, . \, Q_{i}, C,
\alpha'] = r$ for $i \in \{ 1, 2 \}$.

\item If $\alpha' \neq \alpha$ or $C$ does not contain $Q_{1}$ and $Q_{2}$, then $q[(\alpha, r) \, . \,
Q_{i}, C, \alpha'] = 0$ for $i \in \{ 1, 2 \}$.
				\end{itemize}
			\end{itemize}

\item $\cB' = \cI_{\cC} \cup \{ (Q_{1} \sync{L} Q, Q_{2} \sync{L} Q) \mid (Q_{1}, Q_{2}) \in \cB \}$ and
$\cB'' = \cI_{\cC} \cup \{ (Q \sync{L} Q_{1}, Q \sync{L} Q_{2}) \mid (Q_{1}, Q_{2}) \in \cB \}$ are weak
ordinary bisimulations too, where $\cI_{\cC}$ is the identity relation over~$\cC$. Due to the symmetry of
the operational semantic rules for cooperation, without loss of generality we can focus on a single
relation, say $\cB'$. The only interesting cases are those in which we consider $(Q_{1} \sync{L} Q, Q_{2}
\sync{L} Q) \in \cB'$, hence $(Q_{1}, Q_{2}) \in \cB$, and $C \in \cC / \cB'$ of the form $[Q']_{\cB} \,
\sync{L} \bar{Q}$, i.e., $\{ Q'' \sync{L} \bar{Q} \mid (Q'', Q') \in \cB \}$. There are two subcases based
on $\alpha \in \cA$ such that either $\alpha \neq \tau$, or $\alpha = \tau$ and $Q_{1} \sync{L} Q, Q_{2}
\sync{L} Q \notin C$:

			\begin{itemize}
\item If $\alpha \notin L$ then $q[Q_{1} \sync{L} Q, C, \alpha] = q[Q_{2} \sync{L} Q, C, \alpha]$ because
for $i \in \{ 1, 2 \}$ it holds that $q[Q_{i} \sync{L} Q, C, \alpha]$ is equal to:

				\begin{itemize}
\item $q[Q_{i}, [Q']_{\cB}, \alpha] + q[Q, \{ \bar{Q} \}, \alpha]$ if $Q_{i} \in [Q']_{\cB}$ and $Q =
\bar{Q}$ (not to be examined if $\alpha = \tau$)

\item $q[Q_{i}, [Q']_{\cB}, \alpha]$ if $Q_{i} \notin [Q']_{\cB}$ and $Q = \bar{Q}$ (hence $Q_{i} \sync{L} Q
\notin C$ when $\alpha = \tau$)

\item $q[Q, \{ \bar{Q} \}, \alpha]$ if $Q_{i} \in [Q']_{\cB}$ and $Q \neq \bar{Q}$

\item $0$ if $Q_{i} \notin [Q']_{\cB}$ and $Q \neq \bar{Q}$
				\end{itemize}

with $(Q_{1}, Q_{2}) \in \cB$.

\item If $\alpha \in L$ then $q[Q_{1} \sync{L} Q, C, \alpha] = q[Q_{2} \sync{L} Q, C, \alpha]$ because for
$i \in \{ 1, 2 \}$ it holds that $q[Q_{i} \sync{L} Q, C, \alpha] = \frac{q[Q_{i}, [Q']_{\cB},
\alpha]}{q[Q_{i}, \alpha]} \cdot \frac{q[Q, \{ \bar{Q} \}, \alpha]}{q[Q, \alpha]} \cdot \min(q[Q_{i},
\alpha], q[Q, \alpha])$ with $(Q_{1}, Q_{2}) \in \cB$. Note that $\alpha \in L$ implies $\alpha \neq \tau$,
hence $q[Q_{1}, [Q']_{\cB}, \alpha] = q[Q_{2}, [Q']_{\cB}, \alpha]$ -- from which $q[Q_{1}, \alpha] =
q[Q_{2}, \alpha]$ also follows -- when $Q_{1}, Q_{2} \in [Q']_{\cB}$.
			\end{itemize}

\item $\cB' = \cI_{\cC} \cup \{ (Q_{1} \, / \, L, Q_{2} \, / \, L) \mid (Q_{1}, Q_{2}) \in \cB \}$ is a weak
ordinary bisimulation too. The only interesting cases are those in which we consider $(Q_{1} \, / \, L,
Q_{2} \, / \, L) \in \cB'$, hence $(Q_{1}, Q_{2}) \in \cB$, and $C \in \cC / \cB'$ of the form $[Q]_{\cB} \,
/ \, L$, i.e., $\{ Q' \, / \, L \mid (Q, Q') \in \cB \}$. There are three subcases based on $\alpha \in \cA$
such that either $\alpha \neq \tau$, or $\alpha = \tau$ and $Q_{1} \, / \, L, Q_{2} \, / \, L \notin C$:

			\begin{itemize}
\item If $\alpha \in L$ then $q[Q_{1} \, / \, L, C, \alpha] = 0 = q[Q_{2} \, / \, L, C, \alpha]$.

\item If $\alpha = \tau$ then $q[Q_{1} \, / \, L, C, \alpha] = q[Q_{2} \, / \, L, C, \alpha]$ because
$q[Q_{i} \, / \, L, C, \tau] = \sum_{\alpha' \in L \cup \{ \tau \}} q[Q_{i},$ \linebreak $[Q]_{\cB},
\alpha']$ for $i \in \{ 1, 2 \}$ with $(Q_{1}, Q_{2}) \in \cB$. Note that $Q_{i} \, / \, L \notin C$ implies
$Q_{i} \notin [Q]_{\cB}$.

\item If $\alpha \notin L \cup \{ \tau \}$ then $q[Q_{1} / L, C, \alpha] = q[Q_{2} / L, C, \alpha]$ because
$q[Q_{i} / L, C, \alpha] = q[Q_{i}, [Q]_{\cB}, \alpha]$ for $i \in \{ 1, 2 \}$ with $(Q_{1}, Q_{2}) \in
\cB$.
\qedhere
			\end{itemize}
		\end{itemize}
	\end{proof}

To see why congruence with respect to choice does not hold for weak ordinary bisimilarity, it is sufficient
to examine a suitable variant of the usual counterexample for weak bisimilarity in nondeterministic process
algebra. Following~\cite{Mil89a}, to achieve compositionality for choice we introduce a \emph{rooted}
variant of weak ordinary bisimilarity -- which requires any two related components to precisely match their
first executed activities even when their type is $\tau$ and the reached $\approx_{\rm o}$-class contains
those two components -- and then prove that it is the coarsest congruence with respect to choice that is
contained in weak ordinary bisimilarity.

	\begin{exa}[$\approx_{\rm o}$ w.r.t.\ choice]\label{ex:no_congr_choice_weak_ordinary_bisim}
Consider the three non-recursive components:
\[
P''_{1} \: = \: (\tau, r_{\tau}) \, . \, \underline{0} \qquad 
P''_{2} \: = \: \underline{0} \qquad
P'' \: = \: (\alpha, r_{\alpha}) \, . \, \underline{0} 
\]
with $\alpha \neq \tau$ and $r_{\tau}, r_{\alpha} \in \mathbb{R}_{> 0}$. The first two are variants of those
in Example~\ref{ex:no_alpha_tau_p_nr}.

It holds that $P''_{1} \approx_{\rm o} P''_{2}$ as witnessed by the equivalence relation whose only
non-singleton equivalence class is $C = \{ P''_{1}, P''_{2} \}$. Indeed, either component in $C$ reaches no
other class via an activity of type different from $\tau$. Note that $P''_{1}$ and $P''_{2}$ would not be
equivalent if the initial $\tau$-transition of $P''_{1}$ had to be matched by $P''_{2}$ in the bisimulation
game.

However, $P''_{1} + P'' \not\approx_{\rm o} P''_{2} + P''$. The reason is that $P''_{1} + P''$ and
$\underline{0}$ are not weakly ordinary bisimilar due to the $\alpha$-activity of $P''$. Therefore, denoting
by $C$ the equivalence class of~$\underline{0}$ with respect to weak ordinary bisimilarity, we have that
$q[P''_{1} + P'', C, \tau] = r_{\tau} \neq 0 = q[P''_{2} + P'', C, \tau]$.
	\end{exa}

	\begin{defi}\label{def:rooted_weak_ordinary_bisim}
We say that $P_{1}, P_{2} \in \cC$ are \emph{rooted weakly ordinary bisimilar}, written $P_{1} \approx_{\rm
ro} P_{2}$, iff for all action types $\alpha \in \cA$ and for all equivalence classes $C \in \cC /
{\approx_{\rm o}}$:
\[
q[P_{1}, C, \alpha] \: = \: q[P_{2}, C, \alpha]
\]
	\end{defi}

	\begin{thm}[Coarsest congruence contained in $\approx_{\rm o}$]
	\label{thm:coarsest_congr_weak_ordinary_bisim}
Let $P_{1}, P_{2} \in \cC$. Then $P_{1} \approx_{\rm ro} P_{2}$ iff $P_{1} + P \approx_{\rm o} P_{2} + P$
for all $P \in \cC$.
	\end{thm}

	\begin{proof}
We prove the two implications separately:

		\begin{itemize}
\item If $P_{1} \approx_{\rm ro} P_{2}$ then $q[P_{1}, C, \alpha] = q[P_{2}, C, \alpha]$ for all $\alpha \in
\cA$ and $C \in \cC / {\approx_{\rm o}}$. Given an arbitrary $P \in \cC$, since for $i \in \{ 1, 2 \}$ we
have that $q[P_{i} + P, C, \alpha] = q[P_{i}, C, \alpha] + q[P, C, \alpha]$, it holds that $P_{1} + P
\approx_{\rm ro} P_{2} + P$, hence $P_{1} + P \approx_{\rm o} P_{2} + P$ because ${\approx_{\rm ro}} \subset
{\approx_{\rm o}}$.
 
\item As far as the reverse implication is concerned, we reason on the contrapositive. Suppose that $P_{1}
\not\approx_{\rm ro} P_{2}$, i.e., $q[P_{1}, C, \alpha] \neq q[P_{2}, C, \alpha]$ for some $\alpha \in \cA$
and $C \in \cC /{\approx_{\rm o}}$. Let $P$ be $(\alpha', r') \, . \, \underline{0}$ with $\alpha' \neq
\tau$ not occurring in $P_{1}$ and $P_{2}$, so that $\alpha$ does not occur in $P$. Then $P_{1} + P
\not\approx_{\rm o} P_{2} + P$ because $q[P_{1} + P, C, \alpha] = q[P_{1}, C, \alpha] \neq q[P_{2}, C,
\alpha] = q[P_{2} + P, C, \alpha]$. Note that if $\alpha = \tau$ then it cannot be the case that $P_{1} + P,
P_{2} + P \in C$ because otherwise, when computing $q[P_{1} + P, C, \alpha]$ and $q[P_{2} + P, C, \alpha]$,
for some $i \in \{ 1, 2 \}$ there would exist $P'_{i} \in \mathit{ds}(P_{i})$ such that $P_{i} + P, P'_{i}
\in C$, which is not possible because $\alpha'$ can be executed within $P_{i} + P$ only by $P$.
\qedhere
		\end{itemize}
	\end{proof}

Now we prove that weak exact bisimilarity and weak strict bisimilarity are congruences with respect to all
the operators of PEPA other than prefix and choice. The counterexamples are the same as those for strong
exact bisimilarity (Examples~\ref{ex:no_congr_prefix_exact_bisim} and~\ref{ex:no_congr_choice_exact_bisim})
and strong strict bisimilarity (Example~\ref{ex:no_congr_prefix_choice_strict_bisim}), respectively. To
achieve compositionality for those two operators, as in the strong case we need to restrict to components
with no recursion. Note that the counterexample for weak ordinary bisimilarity
(Example~\ref{ex:no_congr_choice_weak_ordinary_bisim}) does not apply to weak strict bisimilarity. Indeed,
$P''_{1} \not\approx_{\rm s} P''_{2}$ because if they were in the same equivalence class $C$ then $q[C \cap
\mathit{ds}((\tau, r_{\tau}) \, . \, \underline{0}), (\tau, r_{\tau}) \, . \, \underline{0}, \tau] -
q[(\tau, r_{\tau}) \, . \, \underline{0}, \tau] = 0 - r_{\tau} \neq 0 - 0 = q[C \cap
\mathit{ds}(\underline{0}), \underline{0}, \tau] - q[\underline{0}, \tau]$.

	\begin{thm}[Congruence of $\approx_{\rm e}$]\label{thm:congr_weak_exact_bisim}
Let $P_{1}, P_{2} \in \cC$. If $P_{1} \approx_{\rm e} P_{2}$ then:

		\begin{itemize}
\item $P_{1} \sync{L} P \approx_{\rm e} P_{2} \sync{L} P$ and $P \sync{L} P_{1} \approx_{\rm e} P \sync{L}
P_{2}$ for all $L \subseteq \cA \setminus \{ \tau \}$ and $P \in \cC$.

\item $P_{1} \, / \, L \approx_{\rm e} P_{2} \, / \, L$ for all $L \subseteq \cA \setminus \{ \tau \}$.
		\end{itemize}
	\end{thm}

	\begin{proof}
Let $\cB$ be a weak exact bisimulation such that $(P_{1}, P_{2}) \in \cB$. We prove that:

		\begin{itemize}
\item $\cB' = \cI_{\cC} \cup \{ (Q_{1} \sync{L} Q, Q_{2} \sync{L} Q) \mid (Q_{1}, Q_{2}) \in \cB \}$ and
$\cB'' = \cI_{\cC} \cup \{ (Q \sync{L} Q_{1}, Q \sync{L} Q_{2}) \mid (Q_{1}, Q_{2}) \in \cB \}$ are weak
exact bisimulations too, where $\cI_{\cC}$ is the identity relation over $\cC$. Due to the symmetry of the
operational semantic rules for cooperation, without loss of generality we can focus on a single relation,
say $\cB'$. The only interesting cases are those in which we consider $(Q_{1} \sync{L} Q, Q_{2} \sync{L} Q)
\in \cB'$, hence $(Q_{1}, Q_{2}) \in \cB$, and $C \in \cC / \cB'$ of the form $[Q']_{\cB} \, \sync{L}
\bar{Q}$, i.e., $\{ Q'' \sync{L} \bar{Q} \mid (Q'', Q') \in \cB \}$. There are two subcases based on $\alpha
\in \cA$:

			\begin{itemize}
\item If $\alpha \notin L$ then:

				\begin{itemize}
\item $q[Q_{1} \sync{L} Q, \alpha] = q[Q_{2} \sync{L} Q, \alpha]$ if $\alpha \neq \tau$ because for $i \in
\{ 1, 2 \}$ it holds that $q[Q_{i} \sync{L} Q, \alpha] \linebreak = q[Q_{i}, \alpha] + q[Q, \alpha]$ with
$(Q_{1}, Q_{2}) \in \cB$.

\item $q[C \cap \mathit{ds}(Q_{1} \sync{L} Q), Q_{1} \sync{L} Q, \alpha] = q[C \cap \mathit{ds}(Q_{2}
\sync{L} Q), Q_{2} \sync{L} Q, \alpha]$ if either $\alpha \neq \tau$, or $\alpha = \tau$ and $Q_{1} \sync{L}
Q, Q_{2} \sync{L} Q \notin C$, because for $i \in \{ 1, 2 \}$ it holds that $q[C \cap \mathit{ds}(Q_{i}
\sync{L} Q), Q_{i} \sync{L} Q, \alpha]$ is equal to:

					\begin{itemize}
\item $q[[Q']_{\cB} \cap \mathit{ds}(Q_{i}), Q_{i}, \alpha] + q[\{ \bar{Q} \}, Q, \alpha]$ if $Q_{i} \in
[Q']_{\cB}$ and $Q = \bar{Q}$ (not to be examined here if $\alpha = \tau$)

\item $q[[Q']_{\cB} \cap \mathit{ds}(Q_{i}), Q_{i}, \alpha]$ if $Q_{i} \notin [Q']_{\cB}$ and $Q = \bar{Q}$
(hence $Q_{i} \sync{L} Q \notin C$ when $\alpha = \tau$)

\item $q[\{ \bar{Q} \}, Q, \alpha]$ if $Q_{i} \in [Q']_{\cB}$ and $Q \neq \bar{Q}$

\item $0$ if $Q_{i} \notin [Q']_{\cB}$ and $Q \neq \bar{Q}$
					\end{itemize}

with $(Q_{1}, Q_{2}) \in \cB$.

\item $q[C \cap \mathit{ds}(Q_{1} \sync{L} Q), Q_{1} \sync{L} Q, \alpha] - q[Q_{1} \sync{L} Q, \alpha] = q[C
\cap \mathit{ds}(Q_{2} \sync{L} Q), Q_{2} \sync{L} Q, \alpha] - q[Q_{2} \sync{L} Q, \alpha]$ if $\alpha =
\tau$ and $Q_{1} \sync{L} Q, Q_{2} \sync{L} Q \in C$ because for $i \in \{ 1, 2 \}$ it holds that $q[C \cap
\mathit{ds}(Q_{i} \sync{L} Q), Q_{i} \sync{L} Q, \tau] - q[Q_{i} \sync{L} Q, \tau] = (q[[Q']_{\cB} \cap
\mathit{ds}(Q_{i}), Q_{i}, \tau] + q[\{ \bar{Q} \}, Q, \tau]) - (q[Q_{i}, \tau] + q[Q, \tau]) =
(q[[Q']_{\cB} \cap \mathit{ds}(Q_{i}), Q_{i}, \tau] - q[Q_{i}, \tau]) + (q[\{ \bar{Q} \}, Q, \tau] - q[Q,
\tau])$ with $(Q_{1}, Q_{2}) \in \cB$. Note that $Q_{i} \sync{L} Q \in C$ implies $Q_{i} \in [Q']_{\cB}$.
				\end{itemize}

\item If $\alpha \in L$ then $q[Q_{1} \sync{L} Q, \alpha] = q[Q_{2} \sync{L} Q, \alpha]$ and $q[C \cap
\mathit{ds}(Q_{1} \sync{L} Q), Q_{1} \sync{L} Q, \alpha] = q[C \cap \mathit{ds}(Q_{2} \sync{L} Q), Q_{2}
\sync{L} Q, \alpha]$ because for $i \in \{ 1, 2 \}$ it holds that $q[Q_{i} \sync{L} Q, \alpha] =
\min(q[Q_{i}, \alpha], q[Q, \alpha])$ and $q[C \cap \mathit{ds}(Q_{i} \sync{L} Q), Q_{i} \sync{L} Q, \alpha]
= \frac{q[[Q']_{\cB} \cap \mathit{ds}(Q_{i}), Q_{i}, \alpha]}{q[Q', \alpha]} \cdot \frac{q[\{ \bar{Q} \}, Q,
\alpha]}{q[\bar{Q}, \alpha]} \cdot \min(q[Q', \alpha], q[\bar{Q}, \alpha])$ with $(Q_{1}, Q_{2}) \in \cB$.
Note that $\alpha \in L$ implies $\alpha \neq \tau$, hence $q[Q_{1}, \alpha] \linebreak = q[Q_{2}, \alpha]$
and $q[[Q']_{\cB} \cap \mathit{ds}(Q_{1}), Q_{1}, \alpha] = q[[Q']_{\cB} \cap \mathit{ds}(Q_{2}), Q_{2},
\alpha]$ when $Q_{1}, Q_{2} \in [Q']_{\cB}$.
			\end{itemize}

\item $\cB' = \cI_{\cC} \cup \{ (Q_{1} \, / \, L, Q_{2} \, / \, L) \mid (Q_{1}, Q_{2}) \in \cB \}$ is a weak
exact bisimulation too. The only interesting cases are those in which we consider $(Q_{1} \, / \, L, Q_{2}
\, / \, L) \in \cB'$, hence $(Q_{1}, Q_{2}) \in \cB$, and $C \in \cC / \cB'$ of the form $[Q]_{\cB} \, / \,
L$, i.e., $\{ Q' \, / \, L \mid (Q, Q') \in \cB \}$. There are three subcases based on $\alpha \in \cA$:

			\begin{itemize}
\item If $\alpha \in L$ then $q[Q_{1} \, / \, L, \alpha] = 0 = q[Q_{2} \, / \, L, \alpha]$ and $q[C \cap
\mathit{ds}(Q_{1} \, / \, L), Q_{1} \, / \, L, \alpha] = 0 = q[C \cap \mathit{ds}(Q_{2} \, / \, L), Q_{2} \,
/ \, L, \alpha]$.

\item If $\alpha = \tau$ then:

				\begin{itemize}
\item $q[C \cap \mathit{ds}(Q_{1} \, / \, L), Q_{1} \, / \, L, \alpha] = q[C \cap \mathit{ds}(Q_{2} \, / \,
L), Q_{2} \, / \, L, \alpha]$ if $Q_{1} \, / \, L, Q_{2} \, / \, L \notin C$ because $q[C \cap
\mathit{ds}(Q_{i} \, / \, L), Q_{i} \, / \, L, \tau] = \sum_{\alpha' \in L \cup \{ \tau \}} q[[Q]_{\cB} \cap
\mathit{ds}(Q_{i}), Q_{i}, \alpha']$ for $i \in \{ 1, 2 \}$ with $(Q_{1}, Q_{2}) \in \cB$. Note that $Q_{i}
\, / \, L \notin C$ implies $Q_{i} \notin [Q]_{\cB}$.

\item $q[C \cap \mathit{ds}(Q_{1} \, / \, L), Q_{1} \, / \, L, \alpha] - q[Q_{1} \, / \, L, \alpha] = q[C
\cap \mathit{ds}(Q_{2} \, / \, L), Q_{2} \, / \, L, \alpha] - q[Q_{2} \, / \, L, \alpha]$ if $Q_{1} \, / \,
L, Q_{2} \, / \, L \in C$ because $q[C \cap \mathit{ds}(Q_{i} \, / \, L), Q_{i} \, / \, L, \tau] - q[Q_{i}
\, / \, L, \tau] = \sum_{\alpha' \in L \cup \{ \tau \}} q[[Q]_{\cB} \linebreak \cap \mathit{ds}(Q_{i}),
Q_{i}, \alpha'] - \sum_{\alpha' \in L \cup \{ \tau \}} q[Q_{i}, \alpha'] = \sum_{\alpha' \in L \cup \{ \tau
\}} (q[[Q]_{\cB} \cap \mathit{ds}(Q_{i}), Q_{i}, \alpha'] - q[Q_{i}, \alpha'])$ for $i \in \{ 1, 2 \}$ with
$(Q_{1}, Q_{2}) \in \cB$. Note that $Q_{i} \, / \, L \in C$ implies $Q_{i} \in [Q]_{\cB}$, hence
$q[[Q]_{\cB} \cap \mathit{ds}(Q_{1}), Q_{1}, \tau] - q[Q_{1}, \tau] = q[[Q]_{\cB} \cap \mathit{ds}(Q_{2}),
Q_{2}, \tau] - q[Q_{2}, \tau]$.
				\end{itemize}

\item If $\alpha \notin L \cup \{ \tau \}$ then $q[Q_{1} \, / \, L, \alpha] = q[Q_{2} \, / \, L, \alpha]$
and $q[C \cap \mathit{ds}(Q_{1} \, / \, L), Q_{1} \, / \, L, \alpha] = q[C \cap \mathit{ds}(Q_{2} \, / \,
L), Q_{2} \, / \, L, \alpha]$ because $q[Q_{i} \, / \, L, \alpha] = q[Q_{i}, \alpha]$ and $q[Q_{i} / L, C,
\alpha] = q[[Q]_{\cB} \cap \mathit{ds}(Q_{i}), Q_{i}, \alpha]$ for $i \in \{ 1, 2 \}$ with $(Q_{1}, Q_{2})
\in \cB$.
\qedhere
			\end{itemize}
		\end{itemize}
	\end{proof}

	\begin{thm}[Congruence of $\approx_{\rm s}$]\label{thm:congr_weak_strict_bisim}
Let $P_{1}, P_{2} \in \cC$. If $P_{1} \approx_{\rm s} P_{2}$ then:

		\begin{itemize}
\item $P_{1} \sync{L} P \approx_{\rm s} P_{2} \sync{L} P$ and $P \sync{L} P_{1} \approx_{\rm s} P \sync{L}
P_{2}$ for all $L \subseteq \cA \setminus \{ \tau \}$ and $P \in \cC$.

\item $P_{1} \, / \, L \approx_{\rm s} P_{2} \, / \, L$ for all $L \subseteq \cA \setminus \{ \tau \}$.
		\end{itemize}
	\end{thm}

	\begin{proof}
Let $\cB$ be a weak strict bisimulation such that $(P_{1}, P_{2}) \in \cB$. We prove that:

		\begin{itemize}
\item $\cB' = \cI_{\cC} \cup \{ (Q_{1} \sync{L} Q, Q_{2} \sync{L} Q) \mid (Q_{1}, Q_{2}) \in \cB \}$ and
$\cB'' = \cI_{\cC} \cup \{ (Q \sync{L} Q_{1}, Q \sync{L} Q_{2}) \mid (Q_{1}, Q_{2}) \in \cB \}$ are weak
strict bisimulations too, where $\cI_{\cC}$ is the identity relation over $\cC$. Due to the symmetry of the
operational semantic rules for cooperation, without loss of generality we can focus on a single relation,
say $\cB'$. The only interesting cases are those in which we consider $(Q_{1} \sync{L} Q, Q_{2} \sync{L} Q)
\in \cB'$, hence $(Q_{1}, Q_{2}) \in \cB$, and $C \in \cC / \cB'$ of the form $[Q']_{\cB} \, \sync{L}
\bar{Q}$, i.e., $\{ Q'' \sync{L} \bar{Q} \mid (Q'', Q') \in \cB \}$.

\noindent
As for outgoing conditions, there are two subcases based on $\alpha \in \cA$ such that either $\alpha \neq
\tau$, or $\alpha = \tau$ and $Q_{1} \sync{L} Q, Q_{2} \sync{L} Q \notin C$:

			\begin{itemize}
\item If $\alpha \notin L$ then $q[Q_{1} \sync{L} Q, C, \alpha] = q[Q_{2} \sync{L} Q, C, \alpha]$ because
for $i \in \{ 1, 2 \}$ it holds that $q[Q_{i} \sync{L} Q, C, \alpha]$ is equal to:

				\begin{itemize}
\item $q[Q_{i}, [Q']_{\cB}, \alpha] + q[Q, \{ \bar{Q} \}, \alpha]$ if $Q_{i} \in [Q']_{\cB}$ and $Q =
\bar{Q}$ (not to be examined if $\alpha = \tau$)

\item $q[Q_{i}, [Q']_{\cB}, \alpha]$ if $Q_{i} \notin [Q']_{\cB}$ and $Q = \bar{Q}$ (hence $Q_{i} \sync{L} Q
\notin C$ when $\alpha = \tau$)

\item $q[Q, \{ \bar{Q} \}, \alpha]$ if $Q_{i} \in [Q']_{\cB}$ and $Q \neq \bar{Q}$

\item $0$ if $Q_{i} \notin [Q']_{\cB}$ and $Q \neq \bar{Q}$
				\end{itemize}

with $(Q_{1}, Q_{2}) \in \cB$.

\item If $\alpha \in L$ then $q[Q_{1} \sync{L} Q, C, \alpha] = q[Q_{2} \sync{L} Q, C, \alpha]$ because for
$i \in \{ 1, 2 \}$ it holds that $q[Q_{i} \sync{L} Q, C, \alpha] = \frac{q[Q_{i}, [Q']_{\cB},
\alpha]}{q[Q_{i}, \alpha]} \cdot \frac{q[Q, \{ \bar{Q} \}, \alpha]}{q[Q, \alpha]} \cdot \min(q[Q_{i},
\alpha], q[Q, \alpha])$ with $(Q_{1}, Q_{2}) \in \cB$. Note that $\alpha \in L$ implies $\alpha \neq \tau$,
hence $q[Q_{1}, [Q']_{\cB}, \alpha] = q[Q_{2}, [Q']_{\cB}, \alpha]$ -- from which $q[Q_{1}, \alpha] =
q[Q_{2}, \alpha]$ also follows -- when $Q_{1}, Q_{2} \in [Q']_{\cB}$.
			\end{itemize}

\noindent
As for incoming conditions, we proceed like in the proof of the corresponding result of
Theorem~\ref{thm:congr_weak_exact_bisim}.

\item $\cB' = \cI_{\cC} \cup \{ (Q_{1} \, / \, L, Q_{2} \, / \, L) \mid (Q_{1}, Q_{2}) \in \cB \}$ is a weak
strict bisimulation too. The only interesting cases are those in which we consider $(Q_{1} \, / \, L, Q_{2}
\, / \, L) \in \cB'$, hence $(Q_{1}, Q_{2}) \in \cB$, and $C \in \cC / \cB'$ of the form $[Q]_{\cB} \, / \,
L$, i.e., $\{ Q' \, / \, L \mid (Q, Q') \in \cB \}$.

\noindent
As for outgoing conditions, there are three subcases based on $\alpha \in \cA$ such that either $\alpha \neq
\tau$, or $\alpha = \tau$ and $Q_{1} \, / \, L, Q_{2} \, / \, L \notin C$:

			\begin{itemize}
\item If $\alpha \in L$ then $q[Q_{1} \, / \, L, C, \alpha] = 0 = q[Q_{2} \, / \, L, C, \alpha]$.

\item If $\alpha = \tau$ then $q[Q_{1} \, / \, L, C, \alpha] = q[Q_{2} \, / \, L, C, \alpha]$ because
$q[Q_{i} \, / \, L, C, \tau] = \sum_{\alpha' \in L \cup \{ \tau \}} q[Q_{i},$ \linebreak $[Q]_{\cB},
\alpha']$ for $i \in \{ 1, 2 \}$ with $(Q_{1}, Q_{2}) \in \cB$. Note that $Q_{i} \, / \, L \notin C$ implies
$Q_{i} \notin [Q]_{\cB}$.

\item If $\alpha \notin L \cup \{ \tau \}$ then $q[Q_{1} / L, C, \alpha] = q[Q_{2} / L, C, \alpha]$ because
$q[Q_{i} / L, C, \alpha] = q[Q_{i}, [Q]_{\cB}, \alpha]$ for $i \in \{ 1, 2 \}$ with $(Q_{1}, Q_{2}) \in
\cB$.
			\end{itemize}

\noindent
As for incoming conditions, we proceed like in the proof of the corresponding result of
Theorem~\ref{thm:congr_weak_exact_bisim}.
\qedhere
		\end{itemize}
	\end{proof}

	\begin{thm}[Congruence of $\approx_{\rm e}$ over $\cC_{\rm nr}$]
	\label{thm:congr_weak_exact_bisim_prefix_choice}
Let $P_{1}, P_{2} \in \cC_{\rm nr}$. If $P_{1} \approx_{\rm e} P_{2}$ then:

		\begin{itemize}
\item $a \, . \, P_{1} \approx_{\rm e} a \, . \, P_{2}$ for all $a \in \cA \times \cR$.

\item $P_{1} + P \approx_{\rm e} P_{2} + P$ and $P + P_{1} \approx_{\rm e} P + P_{2}$ for all $P \in
\cC_{\rm nr}$.
		\end{itemize}
	\end{thm}

	\begin{proof}
Let $\cB$ be a weak exact bisimulation over $\cC_{\rm nr}$ such that $(P_{1}, P_{2}) \in \cB$. We prove
that:

		\begin{itemize}
\item $\cB' = (\cB \cup \{ ((\alpha, r) \, . \, Q_{1}, (\alpha, r) \, . \, Q_{2}) \mid (Q_{1}, Q_{2}) \in
\cB \})^{+}$ is a weak exact bisimulation over $\cC_{\rm nr}$ too, where $^{+}$ denotes transitive closure.
There are only two interesting cases:

			\begin{itemize}
\item If we consider $(Q'_{1}, Q'_{2}) \in \cB' \cap \cB$ -- which thus already satisfies $q[Q'_{1}, \alpha]
= q[Q'_{2}, \alpha]$ for all $\alpha \in \cA \setminus \{ \tau \}$ -- and $C = [(\alpha, r) \, . \,
Q_{1}]_{\cB} \cup [(\alpha, r) \, . \, Q_{2}]_{\cB}$ for some $(Q_{1}, Q_{2}) \in \cB$, then:

				\begin{itemize}
\item $q[C \cap \mathit{ds}(Q'_{1}), Q'_{1}, \alpha'] = q[C \cap \mathit{ds}(Q'_{2}), Q'_{2}, \alpha']$ for
all $\alpha' \in \cA$ such that either $\alpha' \neq \tau$, or $\alpha' = \tau$ and $Q'_{1}, Q'_{2} \notin
C$, because $q[C \cap \mathit{ds}(Q'_{i}), Q'_{i}, \alpha'] = q[[(\alpha, r) \, . \, Q_{1}]_{\cB} \cap
\mathit{ds}(Q'_{i}), Q'_{i}, \alpha'] + q[[(\alpha, r) \, . \, Q_{2}]_{\cB} \cap \mathit{ds}(Q'_{i}),
Q'_{i}, \alpha']$ for $i \in \{ 1, 2 \}$ with $(Q'_{1}, Q'_{2}) \in \cB$. Note that $Q'_{i} \notin C$
implies $Q'_{i} \notin [(\alpha, r) \, . \, Q_{1}]_{\cB}$ and $Q'_{i} \notin [(\alpha, r) \, . \,
Q_{2}]_{\cB}$.

\item $q[C \cap \mathit{ds}(Q'_{1}), Q'_{1}, \alpha'] - q[Q'_{1}, \alpha'] = q[C \cap \mathit{ds}(Q'_{2}),
Q'_{2}, \alpha'] - q[Q'_{2}, \alpha']$ if $\alpha' = \tau$ and $Q'_{1}, Q'_{2} \in C$ because $q[C \cap
\mathit{ds}(Q'_{i}), Q'_{i}, \tau] - q[Q'_{i}, \tau] = q[[(\alpha, r) \, . \, Q_{1}]_{\cB} \cap
\mathit{ds}(Q'_{i}), Q'_{i}, \tau] + q[[(\alpha, r) \, . \, Q_{2}]_{\cB} \cap \mathit{ds}(Q'_{i}), Q'_{i},
\tau] - q[Q'_{i}, \tau]$ for $i \in \{ 1, 2 \}$ with $(Q'_{1}, Q'_{2}) \in \cB$. Note that $Q'_{i} \in C$
implies either $Q'_{i} \in [(\alpha, r) \, . \, Q_{1}]_{\cB}$ or $Q'_{i} \in [(\alpha, r) \, . \,
Q_{2}]_{\cB}$.
				\end{itemize}

\item If we consider $((\alpha, r) \, . \, Q_{1}, (\alpha, r) \, . \, Q_{2}) \in \cB'$, hence $(Q_{1},
Q_{2}) \in \cB$, then the result follows from:

				\begin{itemize}
\item $q[(\alpha, r) \, . \, Q_{1}, \alpha] = r = q[(\alpha, r) \, . \, Q_{2}, \alpha]$.

\item $q[(\alpha, r) \, . \, Q_{1}, \alpha'] = 0 = q[(\alpha, r) \, . \, Q_{2}, \alpha']$ for all $\alpha'
\in \cA \setminus \{ \alpha \}$.

\item $q[C \cap \mathit{ds}((\alpha, r) \, . \, Q_{1}), (\alpha, r) \, . \, Q_{1}, \alpha'] = 0 = q[C \cap
\mathit{ds}((\alpha, r) \, . \, Q_{2}), (\alpha, r) \, . \, Q_{2}, \alpha']$ for all $\alpha' \in \cA$ and
$C \in \cC / \cB'$ such that either $\alpha' \neq \tau$, or $\alpha' = \tau$ and $(\alpha, r) \, . \, Q_{1},
(\alpha, r) \, . \, Q_{2} \notin C$, because we are considering components with no recursion.

\item $q[C \cap \mathit{ds}((\alpha, r) \, . \, Q_{1}), (\alpha, r) \, . \, Q_{1}, \alpha'] - q[(\alpha, r)
\, . \, Q_{1}, \alpha'] = q[C \cap \mathit{ds}((\alpha, r) \, . \, Q_{2}), (\alpha, r) \, . \, Q_{2},
\linebreak \alpha'] - q[(\alpha, r) \, . \, Q_{2}, \alpha']$ for $\alpha' = \tau$ and $C \in \cC / \cB'$
such that $(\alpha, r) \, . \, Q_{1}, (\alpha, r) \, . \, Q_{2} \in C$ because $q[C \cap
\mathit{ds}((\alpha, r) \, . \, Q_{i}), (\alpha, r) \, . \, Q_{i}, \tau] - q[(\alpha, r) \, . \, Q_{i},
\tau] = 0 - q[(\alpha, r) \, . \, Q_{i}, \tau]$ for $i \in \{ 1, 2 \}$ as we are considering components with
no recursion, where $q[(\alpha, r) \, . \, Q_{i}, \tau]$ is $r$ or $0$ depending on whether $\alpha = \tau$
or not.
				\end{itemize}
			\end{itemize}

\item $\cB' = (\cB \cup \{ (Q_{1} + Q, Q_{2} + Q) \mid (Q_{1}, Q_{2}) \in \cB \})^{+}$ and $\cB'' = (\cB
\cup \{ (Q + Q_{1}, Q + Q_{2}) \mid (Q_{1}, Q_{2}) \in \cB \})^{+}$ are weak exact bisimulations over
$\cC_{\rm nr}$ too. Due to the symmetry of the operational semantic rules for choice, without loss of
generality we can focus on a single relation, say $\cB'$. There are only two interesting cases:

			\begin{itemize}
\item If we consider $(Q'_{1}, Q'_{2}) \in \cB' \cap \cB$ -- which thus already satisfies $q[Q'_{1}, \alpha]
= q[Q'_{2}, \alpha]$ for all $\alpha \in \cA \setminus \{ \tau \}$ -- and $C = [Q_{1} + Q]_{\cB} \cup [Q_{2}
+ Q]_{\cB}$ for some $(Q_{1}, Q_{2}) \in \cB$, then:

				\begin{itemize}
\item $q[C \cap \mathit{ds}(Q'_{1}), Q'_{1}, \alpha] = q[C \cap \mathit{ds}(Q'_{2}), Q'_{2}, \alpha]$ for
all $\alpha \in \cA$ such that either $\alpha \neq \tau$, or $\alpha = \tau$ and $Q'_{1}, Q'_{2} \notin C$,
because $q[C \cap \mathit{ds}(Q'_{i}), Q'_{i}, \alpha] = q[[Q_{1} + Q]_{\cB} \cap \mathit{ds}(Q'_{i}),
Q'_{i}, \alpha] + q[[Q_{2} + Q]_{\cB} \cap \mathit{ds}(Q'_{i}), Q'_{i}, \alpha]$ for $i \in \{ 1, 2 \}$ with
$(Q'_{1}, Q'_{2}) \in \cB$. Note that $Q'_{i} \notin C$ implies $Q'_{i} \notin [Q_{1} + Q]_{\cB}$ and
$Q'_{i} \notin [Q_{2} + Q]_{\cB}$.

\item $q[C \cap \mathit{ds}(Q'_{1}), Q'_{1}, \alpha] - q[Q'_{1}, \alpha] = q[C \cap \mathit{ds}(Q'_{2}),
Q'_{2}, \alpha] - q[Q'_{2}, \alpha]$ if $\alpha = \tau$ and $Q'_{1}, Q'_{2} \in C$ because $q[C \cap
\mathit{ds}(Q'_{i}), Q'_{i}, \tau] - q[Q'_{i}, \tau] = q[[Q_{1} + Q]_{\cB} \cap \mathit{ds}(Q'_{i}), Q'_{i},
\tau] + q[[Q_{2} + Q]_{\cB} \cap \mathit{ds}(Q'_{i}), Q'_{i}, \tau] - q[Q'_{i}, \tau]$ for $i \in \{ 1, 2
\}$ with $(Q'_{1}, Q'_{2}) \in \cB$. Note that $Q'_{i} \in C$ implies either $Q'_{i} \in [Q_{1} + Q]_{\cB}$
or $Q'_{i} \in [Q_{2} + Q]_{\cB}$.
				\end{itemize}

\item If we consider $(Q_{1} + Q, Q_{2} + Q) \in \cB'$, hence $(Q_{1}, Q_{2}) \in \cB$, then the result
follows from:

				\begin{itemize}
\item $q[Q_{1} + Q, \alpha] = q[Q_{2} + Q, \alpha]$ for all $\alpha \in \cA \setminus \{ \tau \}$ because
$q[Q_{i} + Q, \alpha] = q[Q_{i}, \alpha] + q[Q, \alpha]$ for $i \in \{ 1, 2 \}$ with $(Q_{1}, Q_{2}) \in
\cB$.

\item $q[C \cap \mathit{ds}(Q_{1} + Q), Q_{1} + Q, \alpha] = 0 = q[C \cap \mathit{ds}(Q_{2} + Q), Q_{2} + Q,
\alpha]$ for all $\alpha \in \cA$ and $C \in \cC / \cB'$ such that either $\alpha \neq \tau$, or $\alpha =
\tau$ and $Q_{1} + Q, Q_{2} + Q \notin C$, because we are considering components with no recursion.

\item $q[C \cap \mathit{ds}(Q_{1} + Q), Q_{1} + Q, \alpha] - q[Q_{1} + Q, \alpha] = q[C \cap
\mathit{ds}(Q_{2} + Q), Q_{2} + Q, \alpha] - q[Q_{2} + Q, \alpha]$ for $\alpha = \tau$ and $C \in \cC /
\cB'$ such that $Q_{1} + Q, Q_{2} + Q \in C$ because for $i \in \{ 1, 2 \}$ it holds that $q[C \cap
\mathit{ds}(Q_{i} + Q), Q_{i} + Q, \tau] - q[Q_{i} + Q, \tau] = 0 - q[Q_{i} + Q, \tau]$ as we are
considering components with no recursion, where $q[Q_{i} + Q, \tau] = q[Q_{i}, \tau] + q[Q, \tau]$ with
$(Q_{1}, Q_{2}) \in \cB$. Note that $q[Q_{1}, \tau] = q[Q_{2}, \tau]$ because, if we denote by $C'$ the
equivalence class of $Q_{1}$ and $Q_{2}$ with respect to $\cB$, then it must be the case that $q[C' \cap
\mathit{ds}(Q_{1}), Q_{1}, \tau] - q[Q_{1}, \tau] = q[C' \cap \mathit{ds}(Q_{2}), Q_{2}, \tau] - q[Q_{2},
\tau]$, hence $0 - q[Q_{1}, \tau] = 0 - q[Q_{2}, \tau]$ by exploiting again the fact that we are considering
components with no recursion.
\qedhere
				\end{itemize}
			\end{itemize}
		\end{itemize}
	\end{proof}

	\begin{thm}[Congruence of $\approx_{\rm s}$ over $\cC_{\rm nr}$]
	\label{thm:congr_weak_strict_bisim_prefix_choice}
Let $P_{1}, P_{2} \in \cC_{\rm nr}$. If $P_{1} \approx_{\rm s} P_{2}$ then:

		\begin{itemize}
\item $a \, . \, P_{1} \approx_{\rm s} a \, . \, P_{2}$ for all $a \in \cA \times \cR$.

\item $P_{1} + P \approx_{\rm s} P_{2} + P$ and $P + P_{1} \approx_{\rm s} P + P_{2}$ for all $P \in
\cC_{\rm nr}$.
		\end{itemize}
	\end{thm}

	\begin{proof}
Let $\cB$ be a weak strict bisimulation over $\cC_{\rm nr}$ such that $(P_{1}, P_{2}) \in \cB$. We prove
that:

		\begin{itemize}
\item $\cB' = (\cB \cup \{ ((\alpha, r) \, . \, Q_{1}, (\alpha, r) \, . \, Q_{2}) \mid (Q_{1}, Q_{2}) \in
\cB \})^{+}$ is a weak strict bisimulation over $\cC_{\rm nr}$ too, where $^{+}$ denotes transitive closure.

\noindent
As for outgoing conditions, there are only two interesting cases:

			\begin{itemize}
\item If we consider $(Q'_{1}, Q'_{2}) \in \cB' \cap \cB$ and $C = [(\alpha, r) \, . \, Q_{1}]_{\cB} \cup
[(\alpha, r) \, . \, Q_{2}]_{\cB}$ for some $(Q_{1}, Q_{2}) \in \cB$, then $q[Q'_{1}, C, \alpha'] =
q[Q'_{2}, C, \alpha']$ for all $\alpha' \in \cA$ such that either $\alpha' \neq \tau$, or $\alpha' = \tau$
and $Q'_{1}, Q'_{2} \notin C$, because $q[Q'_{i}, C, \alpha'] = q[Q'_{i}, [(\alpha, r) \, . \, Q_{1}]_{\cB},
\alpha'] + q[Q'_{i}, [(\alpha, r) \, . \, Q_{2}]_{\cB}, \alpha']$ for $i \in \{ 1, 2 \}$ with $(Q'_{1},
Q'_{2}) \in \cB$. Note that $Q'_{i} \notin C$ implies $Q'_{i} \notin [(\alpha, r) \, . \, Q_{1}]_{\cB}$ and
$Q'_{i} \notin [(\alpha, r) \, . \, Q_{1}]_{\cB}$.

\item If we consider $((\alpha, r) \, . \, Q_{1}, (\alpha, r) \, . \, Q_{2}) \in \cB'$, hence $(Q_{1},
Q_{2}) \in \cB'$, then $q[(\alpha, r) \, . \, Q_{1}, C, \alpha'] \linebreak = q[(\alpha, r) \, . \, Q_{2},
C, \alpha']$ for all $\alpha' \in \cA$ and $C \in \cC / \cB'$ such that either $\alpha' \neq \tau$, or
$\alpha' = \tau$ and $(\alpha, r) \, . \, Q_{1}, (\alpha, r) \, . \, Q_{2} \notin C$, because:

				\begin{itemize}
\item If $\alpha' = \alpha$ and $C$ contains $Q_{1}$ and $Q_{2}$, then $q[(\alpha, r) \, . \, Q_{i}, C,
\alpha'] = r$ for $i \in \{ 1, 2 \}$.

\item If $\alpha' \neq \alpha$ or $C$ does not contain $Q_{1}$ and $Q_{2}$, then $q[(\alpha, r) \, . \,
Q_{i}, C, \alpha'] = 0$ for $i \in \{ 1, 2 \}$.
				\end{itemize}
			\end{itemize}

\noindent
As for incoming conditions, we proceed like in the proof of the corresponding result of
Theorem~\ref{thm:congr_weak_exact_bisim_prefix_choice}.

\item $\cB' = (\cB \cup \{ (Q_{1} + Q, Q_{2} + Q) \mid (Q_{1}, Q_{2}) \in \cB \})^{+}$ and $\cB'' = (\cB
\cup \{ (Q + Q_{1}, Q + Q_{2}) \mid (Q_{1}, Q_{2}) \in \cB \})^{+}$ are weak strict bisimulations over
$\cC_{\rm nr}$ too. Due to the symmetry of the operational semantic rules for choice, without loss of
generality we can focus on a single relation, say $\cB'$.

\noindent
As for outgoing conditions, there are only two interesting cases:

			\begin{itemize}
\item If we consider $(Q'_{1}, Q'_{2}) \in \cB' \cap \cB$ and $C = [Q_{1} + Q]_{\cB} \cup [Q_{2} + Q]_{\cB}$
for some $(Q_{1}, Q_{2}) \in \cB$, then $q[Q'_{1}, C, \alpha] = q[Q'_{2}, C, \alpha]$ for all $\alpha \in
\cA$ such that either $\alpha \neq \tau$, or $\alpha = \tau$ and $Q'_{1}, Q'_{2} \notin C$, because
$q[Q'_{i}, C, \alpha] = q[Q'_{i}, [Q_{1} + Q]_{\cB}, \alpha] + q[Q'_{i}, [Q_{2} + Q]_{\cB}, \alpha]$ for $i
\in \{ 1, 2 \}$ with $(Q'_{1}, Q'_{2}) \in \cB$. Note that $Q'_{i} \notin C$ implies $Q'_{i} \notin [Q_{1} +
Q]_{\cB}$ and $Q'_{i} \notin [Q_{2} + Q]_{\cB}$.

\item If we consider $(Q_{1} + Q, Q_{2} + Q) \in \cB'$, hence $(Q_{1}, Q_{2}) \in \cB$, then $q[Q_{1} + Q,
C, \alpha] = q[Q_{2} + Q, C, \alpha]$ for all $\alpha \in \cA$ and $C \in \cC / \cB'$ such that either
$\alpha \neq \tau$, or $\alpha = \tau$ and $Q_{1} + Q, Q_{2} + Q \notin C$, because $q[Q_{i} + Q, C, \alpha]
= q[Q_{i}, C, \alpha] + q[Q, C, \alpha]$ for $i \in \{ 1, 2 \}$ with $(Q_{1}, Q_{2}) \in \cB$.
			\end{itemize}

\noindent
As for incoming conditions, we proceed like in the proof of
Theorem~\ref{thm:congr_weak_exact_bisim_prefix_choice}.
\qedhere
		\end{itemize}
	\end{proof}

%
%
\section{Conclusion}
\label{sec:conclusion}
%
%

In this paper we have collected within a single, uniform framework the three strong bisimilarities and the
three weak bisimilarities definable over PEPA that respectively induce ordinary, exact, and strict lumpings.
Then we have organised them into a lumpability-driven taxonomy exhibiting all and only the inclusions
holding among them and shown how the taxonomy changes in three special cases: components whose underlying
CTMCs are time reversible, components with no $\tau$-activities, and components with no recursion. Finally
we have proved that all the six bisimilarities are congruences with respect to the cooperation and hiding
operators, while some of them are not congruences with respect to the prefix and/or choice operators. In
that case we have singled out either a set of components over which congruence with respect to those
operators is achieved, or the coarsest congruence with respect to them that is contained in the considered
bisimilarity.


As for future work, on the theoretical side we would like to develop sound and complete axiomatizations for
the six bisimilarities. The laws for strong ordinary bisimilarity over PEPA are already
known~\cite{hillston:book}, as well as those of the three strong bisimilarities over a reversible stochastic
process algebra~\cite{BR23}. Equational laws are very useful to precisely understand which components are
identified by a behavioral equivalence. They can also be employed as rewriting rules for algebraically
manipulating the components themselves in a way that preserves their semantics according to the considered
equivalence.

On the application side, the main direction stemming from our systematisation of knowledge concerns
noninterference analysis of stochastic systems. Noninterference, originally formalised by Goguen and
Meseguer~\cite{GM82}, requires that low-level agents be unable to deduce any information about the behaviour
of high-level agents, thus ruling out information leakage between the two levels. In the literature, weak
behavioral equivalences have been employed to define noninterference properties according to a common
scheme: a system is deemed secure when its behaviour in isolation, i.e., when it is prevented from
interacting on confidential actions, is observationally equivalent to its behaviour when cooperating with an
arbitrary high-level attacker. If the weak equivalence is stochastic, such properties capture not only
functional information flows but also timing covert channels, whereby the attacker infers confidential
information from the timing of the observable behaviour, thus protecting the system against timing attacks.

Persistent Stochastic Non-Interference (PSNI)~\cite{PSNI_strong} instantiates this scheme with lumpable
bisimilarity, i.e., weak ordinary bisimilarity in our terminology, whereas its recent refinement named Exact
Persistent Stochastic Non-Interference (EPSNI)~\cite{valuetools25-pepa} adopts exact lumpable bisimilarity,
i.e., weak exact bisimilarity. Related lines of work have also addressed interference-sensitive behavioural
relations and probabilistic energy-aware modelling for mobile ad-hoc
networks~\cite{BugliesiGMRH12,GallinaHMR11}, further witnessing the relevance of quantitative
noninterference analysis beyond the setting considered here. The other equivalences studied in this paper,
most notably the strict ones, have never been exploited in this context.

In the light of the equivalence spectrum presented here, it would be natural to investigate the taxonomy
that arises at the level of the noninterference properties induced by these equivalences, i.e., to what
extent the inclusions and incomparabilities of Figure~\ref{fig:spectrum} transfer to the corresponding
security properties. The works closest to this programme are those addressing noninterference in the context
of reversible computation, which cover nondeterministic processes~\cite{EABR25}, probabilistic
processes~\cite{EAB24}, deterministically timed processes~\cite{EAB25b}, and stochastically timed
processes~\cite{EAB25a}. Among them, the stochastically timed setting of~\cite{EAB25a} is the closest to
ours, although it relies on Interactive Markov Chains~\cite{hermanns:book} rather than on PEPA and makes use
of a stochastic variant of branching bisimilarity~\cite{GW96} rather than weak bisimilarity. In this
respect, the fact that strict and exact lumpabilities coincide over reversible CTMCs suggests that the
strict bisimilarities introduced here may be natural candidates for the noninterference analysis of
reversible stochastic systems modeled according to an integrated view of activities.

\section*{Acknowledgement}
This study was carried out within the PE0000014 -- ``Security and Rights in the CyberSpace (SERICS)'' and
received funding from the European Union Next-GenerationEU -- National Recovery and Resilience Plan (NRRP)
-- MISSION 4, COMPONENT 2, INVESTMENT 1.3 -- CUP N. H73C22000890001. This work has been also partially
supported by the Research Project INDAM GNCS 2025 -- CUP E53C24001950001 -- ``Modelli e Analisi per sistemi
Reversibili e Quantistici (MARQ)'' and the Research Project PRIN 2020 -- CUP 20202FCJMH -- ``Noninterference
and Reversibility Analysis in Private Blockchains (NiRvAna)''. This manuscript reflects only the authors'
views and opinions; neither the European Union nor the European Commission can be considered responsible for
them.

\bibliographystyle{alphaurl}
\bibliography{performance,security,concurrency}

\end{document}